\documentclass[letterpaper,11pt,leqno]{article}

\usepackage{paper}
\hypersetup{pdftitle={Technology Incidence}}

\newcommand{\bib}{paper.bib}
\begin{document}

\title{The Labor Market Incidence of New Technologies}
\author{Tianyu Fan\\ Yale University \\ \textit{Job Market Paper} 
\thanks{Tianyu Fan: Yale University. Email: tianyu.fan@yale.edu. Website: \url{https://tianyu-fan.com}. \\ 
I am deeply grateful to Michael Peters, Pascual Restrepo, and Fabrizio Zilibotti for their invaluable guidance and unwavering support throughout this project. 
I thank Serdar Birinci, Lorenzo Caliendo, Yu-Ting Chiang, Maxim Dvorkin, Ana Cecilia Fieler, Mayara Felix, Joel Flynn, David Hémous, Zhen Huo, Sam Kortum, Julian Kozlowski, Danial Lashkari, Ricardo Marto, Kjetil Storesletten, Michael Song, Kailan Tian, Sharon Traiberman, Laura Veldkamp, and Anton Yang for insightful discussions and constructive feedback. 
I also benefited from helpful comments from participants at the CCER Summer Institute, CUFE Workshops, Dissertation Fellow Workshop (Federal Reserve Bank of St. Louis), Growth and Development Conference (Federal Reserve Bank of Minneapolis), Growth and Institution Workshop (Tsinghua University), TRAIN Conference, and Yale Macro/Trade Lunch Workshops. 
Financial support from the Stripe Economics of AI Fellowship is gratefully acknowledged.}
}
\date{October 30, 2025 \\ \href{https://www.tianyu-fan.com/files/FAN_Technology_Incidence.pdf}{(Click here for the most recent version)}}

\begin{titlepage}\maketitle
This paper develops a new framework to analyze the incidence of labor market shocks, focusing on automation and artificial intelligence. Central to our theory is the distance-dependent elasticity of substitution (DIDES), where worker mobility between occupations declines with their distance in skill space. Mapping 306 occupations into cognitive, manual, and interpersonal skill dimensions, we estimate a low-dimensional latent skill model that preserves granular substitution patterns. We show that both automation and artificial intelligence cluster within skill-adjacent occupations, constraining employment adjustment and amplifying wage effects. The clustering nature of technologies generates unequal outcomes: 20--50\% of labor demand shocks translate to wages (versus 30\% under standard models), while mobility recovers only 20\% of losses (versus 30\% from standard estimates). 

\end{titlepage}

\section{Introduction} \label{s:intro}

Recent technological advances in automation and artificial intelligence have transformed labor markets, spurring productivity gains while reshaping the distribution of economic outcomes.\footnote{We follow \cite{Acemoglu2022-lv,Restrepo2024-or} in defining automation technologies as industrial robots, machinery, and software without AI capability.} These technologies share a defining characteristic that fundamentally determines their labor market impact: they do not strike randomly but concentrate in task-similar occupations. Automation targets routine manual tasks such as manufacturing, assembly, and transportation, while AI affects cognitive work including data analysis, research, and decision-making. This clustering creates a mobility constraint for displaced workers. Workers can easily transition to task-similar occupations, yet these alternatives face similar technological displacement and offer no economic refuge. An assembly worker displaced by automation could shift to construction or welding, occupations where task requirements overlap substantially, but finds these jobs similarly automated. A data analyst threatened by AI discovers that financial analysis and market research, their most accessible alternatives, face equivalent AI exposure. Workers are most mobile precisely where mobility provides the least benefit.

This paper develops a novel framework for evaluating how labor market adjustment absorbs unequal shocks across occupations, focusing on automation and artificial intelligence. To assess the labor market incidence of these technologies, we need to understand how easily workers can move to differentially affected occupations rather than similarly affected ones: a data analyst threatened by AI benefits from transitioning to management (which faces lower AI exposure) rather than to market research (which faces comparable AI threats). When workers can easily move from negatively affected occupations to those benefiting from technological change, the labor market absorbs shocks through employment reallocation; when mobility is limited, shocks manifest as wage inequality. While this problem has been recognized in the literature, existing frameworks either impose strong restrictions on substitution patterns to retain tractability or group occupations to attain realistic substitution at the cost of removing heterogeneity.\footnote{The former includes constant elasticity of substitution and nested constant elasticity of substitution frameworks. The latter comprises structural labor frameworks, especially Roy models with skill heterogeneity that flexibly model worker choices. At the extreme, skill-biased technical change frameworks define production over a small set of aggregate skill groups, achieving tractability by collapsing occupational heterogeneity entirely.} We argue that accounting for heterogeneity in workers' abilities to reallocate across occupations is crucial, as it shapes the impact of technological shocks on both employment and inequality. To address this, we develop a Roy model with latent skills that preserves rich heterogeneity in substitution patterns while remaining empirically tractable.

Our central contribution is providing a methodology for empirically assessing the incidence of labor market shocks across granular occupations. We build our theoretical analysis on a Roy model of occupational choice augmented with distance-dependent elasticity of substitution (DIDES)\footnote{\cite{Teulings2005-rz} first introduces this concept in a one-dimensional assignment framework for analyzing wage determination. We extend it to a multi-dimensional empirical setting for analyzing shock propagation across occupations.}, where worker mobility declines with skill distance between occupations—that is, differences in occupational skill intensities. Formally, workers draw correlated productivity across occupations, where correlation declines with skill distance in a multi-dimensional skill space. This correlation structure directly generates substitution patterns: high correlation between skill-similar occupations creates strong substitutability, while low correlation between skill-distant occupations limits substitution. When technological shocks cluster in skill-adjacent occupations—as we document for both automation and AI—this structure constrains employment adjustment while amplifying wage effects, generating severe and persistent inequality.

A key innovation of our framework is achieving dimensionality reduction while preserving flexible substitution patterns for a granular occupation structure. Rather than estimating hundreds of thousands of bilateral elasticities among occupations, we parameterize the entire substitution structure through a low-dimensional skill space. The correlation structure governing productivity draws depends on just $S+1$ parameters for $S$ skill dimensions: one cross-skill elasticity $\theta$ and $S$ within-skill correlation parameters $\{\rho_s\}_{s=1}^S$. This parsimony makes estimation feasible from standard aggregate data while maintaining the flexibility to capture how technological clustering constrains employment adjustment.

Our empirical implementation proceeds in three integrated steps that connect measurement to theory to structural estimation. First, we map 306 detailed occupations into a three-dimensional skill space using O*NET data, extracting cognitive, manual, and interpersonal skill intensities through principal component analysis. These skill measures operationalize the theoretical framework: they fully parameterize the labor supply structure through three interpretable skills. Second, we follow the task framework to measure how automation and AI differentially affect occupations by having ChatGPT evaluate the automation and AI feasibility of 19,200 tasks across 862 occupations. This reveals the clustering patterns central to our analysis (automation concentrates in manual-intensive occupations while AI targets cognitive-intensive jobs, with both technologies showing systematic concentration within skill-adjacent occupations rather than random dispersion). Third, we estimate the structural parameters governing substitution by leveraging how occupational employment and average wages responded to historical automation between 1980 and 2010.

The estimation reveals striking departures from standard models that fundamentally alter our understanding of labor market adjustment. Under conventional CES assumptions, which are nested in our framework, we estimate an average elasticity of 3.12, suggesting substantial worker mobility across all occupations. Allowing for skill-based correlation through our DIDES framework changes this picture dramatically. The cross-skill elasticity plummets to 1.10, implying that moving across skill boundaries is far more difficult than standard models assume. Within-skill elasticities show substantial heterogeneity: 4.8 for cognitive occupations, 4.4 for interpersonal occupations, but only 2.1 for manual occupations. The correlation parameters driving these differences reveal that cognitive skills prove most transferable ($\rho_{\text{cog}} = 0.77$), while manual skills exhibit limited transferability ($\rho_{\text{man}} = 0.48$). These estimates reveal that two-thirds of observed occupational substitution occurs within skill dimensions rather than across them, a pattern that becomes crucial when technological shocks themselves cluster.

With the substitution structure estimated, we quantify how automation and AI reshape labor market outcomes. The clustering of technological shocks within skill domains constrains adjustment: when entire skill clusters face negative shocks simultaneously, workers have limited escape options. This manifests in heterogeneous wage pass-through across occupations. While standard models predict uniform 30\% pass-through from demand shocks to wages, we find pass-through rates ranging from 20\% to 50\%. Production workers facing automation experience 40--45\% pass-through, with nearly half their occupational labor demand shocks translating directly to wage declines. The variation directly reflects how clustering eliminates escape options: when skill-similar occupations face simultaneous threats, the effective elasticity of substitution collapses, forcing wage absorption rather than employment reallocation.

Occupational mobility provides limited insurance against these wage losses. Workers recover only 20\% of automation-induced wage declines through occupational transitions, compared to 30\% predicted by standard models. This limited mobility gain emerges from technological shocks concentrating precisely where skill transferability is weakest. Automation targets manual occupations where workers have the lowest transferability ($\rho_{\text{man}} = 0.48$), creating large losses with minimal recovery options. AI affects cognitive occupations where higher transferability ($\rho_{\text{cog}} = 0.77$) offers better prospects, yet clustering still constrains escape options because natural transition targets face similar AI threats. The interaction between shock distribution and heterogeneous mobility (absent from models assuming uniform elasticity) drives the severe distributional consequences we document.

Our static analysis estimates parameters from decadal wage responses, which we interpret as long-run elasticities reflecting equilibrium adjustment over extended periods. However, one key dimension is the speed of labor market adjustment. We extend this static analysis to examine transitions by embedding DIDES into a dynamic discrete choice framework. Our examination of historical automation reveals remarkably persistent effects: gradual adoption since 1985 generated wage gaps up to 50\% between high and low exposure occupations. Employment shifts absorbed two-thirds of demand changes over this period. Under a counterfactual scenario where AI rapidly reaches automation's scale by 2030, adjustment proves much more constrained. The labor market initially absorbs less than one-third of shocks, generating sharp wage declines with mobility recovering only one-third of losses during the transition. The clustering that constrains static adjustment also slows dynamic transitions, with forward-looking behavior providing limited relief because improved outside options are offset by similar threats to alternative occupations.


These findings reshape our understanding of how technological progress affects workers and carry immediate policy implications. Conventional estimates overstate the extent to which labor market flexibility mitigates technological disruption. When technical changes cluster in skill-adjacent occupations, as our evidence establishes for both automation and AI, they systematically target rigidities in occupational substitution. Workers cannot escape to unaffected occupations because skill intensities create barriers, and the occupations they can reach face similar technological threats. This interaction between technology clustering and substitution structure, absent in standard frameworks, explains why technological change generates such pronounced and persistent inequality. Standard policy prescriptions for worker retraining miss this fundamental constraint: displaced workers' natural transition targets face similar technological risks. As AI deployment accelerates, understanding these mechanisms becomes essential for designing policies that facilitate necessary economic transitions while protecting vulnerable workers from concentrated disruption.

\paragraph{Related Literature}

Our paper contributes to four interconnected literatures: skill-biased technical change, labor reallocation dynamics, the Roy model tradition, and assignment theory.
 
\textit{Skill-biased Technical Change:} Our work builds on the extensive literature examining labor market consequences of technological change. Early research established that technological advances disproportionately benefit skilled workers \citep{Katz1992-zr, Autor1998-tu, Acemoglu2002-gk, Autor2013-zm}. Recent task-based frameworks provide a more granular understanding of how automation technologies generate unequal labor demand shifts across occupations \citep{Acemoglu2018-zi, Acemoglu2020-bj, Acemoglu2022-lv}. Our contribution complements this demand-side focus by modeling supply-side adjustment, developing a framework that captures how workers reallocate across occupations and how the interaction between substitution patterns and demand shocks determines equilibrium incidence.

In parallel, emerging research explores AI's distinct disruptive potential. \cite{webb2019impact} and \cite{Acemoglu2022-rp} demonstrate that AI affects both routine and non-routine cognitive tasks, while experimental studies by \cite{Noy2023-yp} and \cite{brynjolfsson2025generative} document how generative AI transforms knowledge-based and creative work. Recent analyses primarily examine AI's demand-side impact through task frameworks \citep{Eloundou2024-cu, brynjolfsson2025canaries, hampole2025artificial, freund2025job}. We provide the first systematic assessment of worker mobility constraints under AI exposure.

\textit{Labor Reallocation and Mobility Frictions:} A growing literature emphasizes how worker reallocation mitigates unequal demand shocks. Recent work documents how occupational mobility constraints amplify wage inequality during transitions \citep{Lee2006-sj, dvorkin2019occupation, Traiberman2019-fe}, with dynamic models studying the regulation policies \citep{Guerreiro2022-tq, lehr2022optimal, Beraja2024-og}. Regarding the source of slow adjustment, \cite{bocquet2024network} examines adjustment through job transition networks, while \cite{Adao2024-ub} highlights skill specialization as a constraint on reallocation. We extend this literature by examining how the distribution of technologies interacts with heterogeneous worker mobility in determining incidence, highlighting the importance of a flexible substitution structure.\footnote{While \cite{Bohm2025} also highlights the importance of heterogeneous labor supply elasticities, their heterogeneity stems solely from differences in employment shares across occupations, not from underlying variation in substitution structure.}


\textit{Roy Models and Multidimensional Skills:} Following the Roy tradition of selection on comparative advantage \citep{Heckman1985-gf}, recent work incorporates multidimensional skills to study business cycle dynamics \citep{grigsby2022skill},\footnote{In \cite{grigsby2022skill}'s framework, the notion of skill is equivalent to jobs (that is, one occupation is one skill). Therefore, he needs to group occupations into 15 clusters to be estimable from job transitions.} discrimination \citep{Hurst2024-jw}, and occupation choices \citep{Lise2020-hm}. We build on \cite{Lise2020-hm}'s insight about multidimensional skill structure but embed it into aggregate labor supply across granular occupations. Our innovation is mapping 300+ occupations into three latent skill dimensions while preserving rich substitution patterns, estimable from standard aggregate employment and average wage data. By adopting a Roy-Fréchet structure with copula-based correlation, we focus directly on substitution patterns.\footnote{This approach circumvents the well-known identification challenges of unobserved heterogeneity that plague selection models \citep{Heckman1990-np, French2011-dt, erosa2025labor}.} This approach yields a tractable framework that uses occupational skill intensities to parameterize substitution structure and aggregate employment shares as sufficient statistics, enabling estimation without relying on individual-level data.

\textit{Assignment Theory and DIDES:} The distance-dependent elasticity of substitution emerges naturally from assignment models where workers sort based on comparative advantage \citep{Sattinger1993-av, Teulings1995-hc, Teulings2005-rz}. These models establish that substitutability declines with skill distance, a theoretical result we operationalize empirically. While \cite{Lindenlaub2017-eu} explores multidimensional assignment theoretically, we provide the first empirical implementation that quantifies DIDES using occupational data, estimates its parameters from observed labor market responses, and demonstrates its crucial role in technological incidence.

\paragraph{Road Map} 
Section \ref{s:theory} develops a static model featuring distance-dependent elasticity of substitution (DIDES). Section \ref{s:empirics} implements the framework empirically, estimating a flexible substitution structure for granular occupations. Section \ref{s:incidence} quantifies the incidence of automation and AI. Section \ref{s:dynamic_model} extends to dynamic adjustment, embedding DIDES into a dynamic discrete choice framework. Section \ref{s:add_extension} addresses extensions including alternative specifications and heterogeneous groups. Section \ref{s:conclusion} concludes.

\section{Theoretical Framework}\label{s:theory}

Our framework combines a parsimonious task-based production structure with a flexible model of worker sorting across occupations. On the production side, each occupation represents a collection of tasks. Technological change shifts task assignment, generating occupation-specific labor demand shocks. On the labor supply side, the economy consists of workers who draw correlated productivities for performing tasks across occupations. Workers choose occupations competitively, sorting based on their comparative advantage in performing occupation-specific task bundles. Central to our analysis is the correlation structure of these productivity draws, which determines occupational substitution patterns.

We introduce a latent skill structure to parameterize the correlation in productivity across occupations. In our framework, "skills" are not primitive worker attributes but rather a dimensionality reduction device: they provide a low-dimensional representation that summarizes how productivity correlates across occupations. This parametric approach proves essential for capturing realistic substitution patterns while maintaining empirical tractability: rather than estimating thousands of bilateral elasticities, we estimate a handful of parameters governing the substitution structure.

\subsection{Model Setup} \label{ss:static_model}

\paragraph{Production and Labor Demand}

Since our focus is on labor supply responses, we adopt a deliberately parsimonious representation of labor demand. Labor demand derives from a task-based production framework following \cite{Acemoglu2018-zi, Acemoglu2022-lv}. In the underlying model (detailed in Appendix \ref{app_ss:production}), occupations perform distinct task sets that can be produced using either labor or capital, with technological change shifting task allocation between these inputs. This yields the reduced-form representation:
\begin{equation}
    y = \mathcal{A}\left(\sum_{o=1}^O \alpha_o^{\frac{1}{\sigma}} L_o^{\frac{\sigma-1}{\sigma}}\right)^{\frac{\sigma}{\sigma-1}}
\end{equation}
where $L_o$ denotes employment in occupation $o$, $\sigma$ is the elasticity of substitution between occupations (the labor demand elasticity), $\mathcal{A}$ captures aggregate productivity, and $\alpha_o$ represents the share of tasks performed by labor in occupation $o$.

The parameter $\alpha_o$ serves as a sufficient statistic for occupation-specific labor demand shocks. When automation or AI replaces labor in specific tasks, the corresponding $\alpha_o$ declines: $d\ln\alpha_o < 0$ for occupations whose tasks become automated. Conversely, if a new technology increases demand for a particular occupation, $d\ln\alpha_o > 0$.\footnote{We focus on labor demand shocks from automation and AI in both theory and measurement. In contrast, \cite{Autor2025-vd} study how automation can also change the occupational supply of workers.} This parsimonious representation captures technology's distributional effects without explicitly tracking task assignments, as the demand shifters $\{\alpha_o\}$ fully summarize technological impacts across occupations.\footnote{The aggregate productivity effect $d\ln\mathcal{A}$ represents a level shift that affects all occupations proportionally. Since our focus is on distributional incidence across occupations, this term cancels out in relative wage analysis and is omitted from subsequent analysis.}

From profit maximization, occupational wages equal marginal products:
\begin{equation*}
    w_o = \frac{\partial y}{\partial L_o} = y^{\frac{1}{\sigma}} \alpha_o^{\frac{1}{\sigma}} L_o^{-\frac{1}{\sigma}}
\end{equation*}

This labor demand equation, combined with the labor supply framework developed below, determines equilibrium wage and employment responses to technological change.

\paragraph{Workers and Labor Supply}

The economy consists of a continuum of workers indexed by $i$. Each worker draws a productivity vector $\boldsymbol{\epsilon}(i) = \{\epsilon_o(i)\}_{o=1}^O$ across occupations from a generalized multivariate Fréchet distribution:
\begin{equation}
    \operatorname{Pr}[\epsilon_1(i) \leq \epsilon_1, \ldots, \epsilon_O(i) \leq \epsilon_O] = \exp\left[-F(A_1\epsilon_1^{-\theta}, \ldots, A_O\epsilon_O^{-\theta})\right] \label{eq:prod_distr}
\end{equation}
where $A_o > 0$ captures average productivity in occupation $o$ and $\theta > 0$ governs productivity dispersion across workers. The marginal distributions are Fréchet: $\operatorname{Pr}[\epsilon_o(i) \leq \epsilon_o] = \exp(-A_o\epsilon_o^{-\theta})$, standard in Roy models with extreme value distributions. The correlation function $F$ is the central primitive of our framework, governing how productivity correlates across occupations and thereby determining the entire structure of occupational substitution.\footnote{The correlation function $F$ is related to the copula of the productivity distribution and satisfies three key properties: homogeneity of degree one, unboundedness, and the sign-switching property (ensuring occupations are gross substitutes). See Appendix~\ref{app_sss:correlation_properties} for formal definitions.}

Workers choose occupations to maximize utility. Worker $i$ receives utility $u_o(i) = w_o\epsilon_o(i)$ from occupation $o$, where $w_o$ is the wage and $\epsilon_o(i)$ represents both productivity and inverse effort cost.\footnote{Formally, workers consume $c_o = w_o$ and supply effort $\ell_o(i) = 1/\epsilon_o(i)$, yielding utility $u_o(i) = c_o/\ell_o(i) = w_o\epsilon_o(i)$.} The optimal occupational choice is:
\begin{equation*}
    o^*(i) = \arg\max_{o \in \{1,\ldots,O\}} \{w_o\epsilon_o(i)\}
\end{equation*}

The correlation function $F: \mathbb{R}_+^O \rightarrow \mathbb{R}_+$ determines substitution patterns between occupations. When productivity draws are highly correlated across occupations, workers transition more readily between them in response to wage changes. When $F$ is additive ($F = \sum_o x_o$), productivity draws are independent and the model reduces to standard CES with uniform elasticity. 

\begin{proposition}[Occupational Employment Shares]
Given the multivariate Fréchet productivity distribution in equation \eqref{eq:prod_distr}, the share of workers selecting occupation $o$ is:
\begin{equation*}
    \pi_o = \frac{A_ow_o^{\theta}F_o(A_1w_1^{\theta}, \ldots, A_Ow_O^{\theta})}{F(A_1w_1^{\theta}, \ldots, A_Ow_O^{\theta})}
\end{equation*}
where $F_o = \partial F/\partial x_o$ denotes the partial derivative with respect to the $o$-th argument.
\end{proposition}

\begin{proof}
See Appendix~\ref{app_ss:emp_share_derivation}.
\end{proof}

The employment share expression reveals that occupation $o$'s share depends on three factors: average productivity $A_o$, wage raised to the dispersion parameter ($w_o^{\theta}$), and how the correlation function characterizes occupation's relative attractiveness ($F_o/F$). This last term breaks the independence of irrelevant alternatives (IIA) property, allowing realistic substitution patterns where wage changes in one occupation affect employment shares differently across other occupations.\footnote{When $F(x_1, \ldots, x_O) = \sum_o x_o$ (independent productivity draws), $F_o/F = 1/\sum_j x_j$ for all $o$, restoring IIA and reducing to standard CES with uniform elasticity $\theta$.}

Total labor supply to occupation $o$ is $L_o = \pi_o\bar{L}$, where $\bar{L}$ is the total workforce. In the baseline specification, we assume idiosyncratic productivity reduces the cost of working but does not enter production.\footnote{Section~\ref{ss:efficiency} extends the framework to incorporate efficiency effects, where a fraction $\delta$ of workers contribute productivity directly to production. This extension reveals that when $\delta > 0$, labor supply elasticities decrease and wage pass-through increases, implying our baseline specification provides conservative estimates of technological incidence.} The correlation function $F$ fully characterizes substitution patterns through its effect on employment share responses to wage changes.\footnote{Section \ref{ss:dides} parameterizes $F$ to capture distance-dependent elasticity of substitution (DIDES), where substitutability declines with skill distance between occupations.}

\paragraph{Market Equilibrium}
A competitive equilibrium consists of a wage vector $\boldsymbol{w}^* = \{w_o^*\}_{o=1}^O$ and allocation $\boldsymbol{L}^* = \{L_o^*\}_{o=1}^O$ such that:
\begin{enumerate}
    \item \textbf{Profit maximization:} Firms choose labor to maximize profits, yielding demand:
    \begin{equation*}
        L_o^d(\boldsymbol{w}) = \left(\frac{\alpha_o}{w_o}\right)^{\sigma} y(\boldsymbol{L})
    \end{equation*}
    
    \item \textbf{Utility maximization:} Workers choose occupations optimally, yielding supply:
    \begin{equation*}
        L_o^s(\boldsymbol{w}) = \pi_o(\boldsymbol{w})\bar{L} = \frac{A_ow_o^{\theta}F_o(A_1w_1^{\theta}, \ldots, A_Ow_O^{\theta})}{F(A_1w_1^{\theta}, \ldots, A_Ow_O^{\theta})}\bar{L}
    \end{equation*}
    
    \item \textbf{Market clearing:} Labor markets clear in all occupations: 
    \begin{equation*}
        L_o^d(\boldsymbol{w}^*) = L_o^s(\boldsymbol{w}^*) = L_o^* \quad \forall o
    \end{equation*}
\end{enumerate}

\begin{proof}
Existence and uniqueness are established in Appendix~\ref{app_ss:equilibrium_proof}.
\end{proof}

\subsection{Technological Shocks and Labor Market Incidence}

We model technological change as shifts in the share of tasks performed by labor across occupations, $d\ln\boldsymbol{\alpha} = \{d\ln\alpha_o\}_{o=1}^O$, and changes in aggregate productivity, $d\ln\mathcal{A}$. While occupation-specific shifts may displace labor, aggregate productivity gains increase total output—the core tension in technological incidence. The distributional question is how these aggregate gains and occupation-specific changes are shared across workers.

\begin{proposition}[Equilibrium Responses to Technology]\label{prop:tech_incidence}
Consider a technological shock characterized by task share changes $\{d\ln\alpha_o\}_{o=1}^O$. To first order:

\noindent (i) Wage and employment responses satisfy:
\begin{align}
    d\ln\boldsymbol{w} + \frac{1}{\sigma}d\ln\boldsymbol{L}&= \frac{1}{\sigma}d\ln y \cdot \mathbf{1} + \frac{1}{\sigma}d\ln\boldsymbol{\alpha} \label{eq:labor_demand_change}\\
    d\ln\boldsymbol{L} &= \Theta \cdot d\ln\boldsymbol{w} \label{eq:labor_supply_change}
\end{align}

\noindent (ii) Equilibrium wage incidence is:
\begin{equation}
    d\ln\boldsymbol{w} = \frac{1}{\sigma}d\ln y \cdot \mathbf{1} + \Delta \cdot \frac{d\ln\boldsymbol{\alpha}}{\sigma}
    \label{eq:wage_incidence}
\end{equation}
where $\Delta = (\mathbf{I} + \Theta/\sigma)^{-1}$ is the pass-through matrix and $\Theta$ is the matrix of labor supply elasticities:
\begin{equation}
    \Theta_{oo'} = \begin{cases}
        \theta\left[\frac{x_{o'}F_{oo'}}{F_o}\bigg|_{x_j = A_jw_j^{\theta}} - \pi_{o'}\right] & \text{if } o \neq o' \\
        \theta\left[\frac{x_oF_{oo}}{F_o}\bigg|_{x_j = A_jw_j^{\theta}} + 1 - \pi_o\right] & \text{if } o = o'
    \end{cases}
    \label{eq:elasticity_matrix}
\end{equation}
\end{proposition}

\begin{proof}
Part (i) follows from log-differentiating first-order conditions and employment shares. Part (ii) combines wage and employment responses. See Appendix \ref{app_ss:elasticity_derivation}.
\end{proof}

This proposition reveals how technological incidence depends on the interaction between shock distribution and the matrix of substitution elasticities. The aggregate output effect $(d\ln y/\sigma)$ raises all wages uniformly. The distributional effect, captured by the pass-through matrix $\Delta$, depends on both demand elasticity $\sigma$ and substitution matrix $\Theta$. This matrix embeds substitution patterns through two components: the correlation term $\theta x_{o'}F_{oo'}/F_o$ reflects productivity correlation between occupations, while the share term $-\theta\pi_{o'}$ represents independent substitution that depends only on employment shares. When productivities are independent ($F = \sum_o x_o$), only the share term remains, reducing to standard CES.\footnote{Rows of $\Theta$ sum to zero, confirming that only relative wage changes induce reallocation. This property follows from the homogeneity of $F$. See Appendix~\ref{app_ss:zero_row_sum}.}

The pass-through matrix $\Delta$ embodies the labor market capability to absorb distributional shocks: greater worker mobility (larger $\|\Theta\|$) enables employment adjustment that dampens wage effects, while limited mobility (smaller $\|\Theta\|$) translates shocks directly into wage disparities. In the limit where $\|\Theta\| \to 0$ (no mobility) or $\sigma \to \infty$ (perfectly elastic demand), the pass-through matrix approaches identity, yielding complete wage incidence. Conversely, as $\theta \to \infty$ (no productivity dispersion), workers become perfectly substitutable and unequal demand changes dissipate through employment reallocation, with pass-through approaching zero.

\paragraph{Mobility Gains and Welfare Recovery}

While equation~\eqref{eq:wage_incidence} captures wage effects for workers remaining in their occupations, a key aspect of demand shock incidence is workers' ability to shield themselves from negative shocks by switching occupations. Some workers benefit from such transitions, partially recovering losses through reallocation to less-affected occupations.

\begin{proposition}[Mobility Gains from Reallocation]\label{prop:mobility}
The expected welfare gain for workers initially in occupation $o$ from occupational transitions is:
\begin{equation}
    \text{Mobility Gain}_o = \sum_{o': d\ln w_{o'} > d\ln w_o} \mu_{oo'}(d\ln w_{o'} - d\ln w_o)
    \label{eq:mobility_gains}
\end{equation}
where $\mu_{oo'} = -\Theta_{oo'}(d\ln w_{o'} - d\ln w_o)$ is the fraction of workers reallocating from $o$ to $o'$.
\end{proposition}

\begin{proof}
See Appendix~\ref{app_ss:mobility_gain}.
\end{proof}

To build intuition for the determinants of mobility gains, we decompose equation~\eqref{eq:mobility_gains} into average and correlation effects. Substituting $\mu_{oo'} = -\Theta_{oo'}(d\ln w_{o'} - d\ln w_o)$ yields:
\begin{align*}
\text{Mobility Gain}_o 
&= \sum_{o':\, d\ln w_{o'} > d\ln w_o} 
\left|\Theta_{oo'}\right|(d\ln w_{o'} - d\ln w_o)^2 \nonumber\\
&= n_o \left[
\underbrace{\overline{\left|\Theta_{oo'}\right|} \cdot 
\overline{(d\ln w_{o'} - d\ln w_o)^2}}_{\text{Average effect}}
+ \underbrace{\text{Cov}\!\left(\left|\Theta_{oo'}\right|,\,
(d\ln w_{o'} - d\ln w_o)^2\right)}_{\text{Correlation effect}}
\right]
\label{eq:mobility_decomp}
\end{align*}
where $n_o=\#\{o': d \ln w_{o'}>d \ln w_o\}$. When technological shocks cluster in occupations that are close substitutes, workers have high mobility precisely to occupations facing similar negative shocks—a data analyst threatened by AI can easily transition to financial analysis, but that occupation faces comparable AI exposure. This negative correlation implies that standard models with uniform elasticities overstate welfare recovery through reallocation while understating inequality.
\subsection{Spectral Analysis of Technological Incidence} \label{ss:spectral}

We employ spectral analysis to understand how the distribution of technological shocks interacts with the substitution matrix to determine labor market incidence. This approach decomposes any shock into fundamental components (eigenshocks) given the occupation substitution structure. Each eigenshock has its own effective elasticity of substitution, revealing the capacity for employment adjustment and the associated wage effects.

\subsubsection{Eigendecomposition and Pass-Through}

The wage incidence equation \eqref{eq:wage_incidence} can be reformulated using the eigenstructure of the labor supply elasticity matrix $\Theta$. While $\Theta$ is not generally symmetric, it admits an eigendecomposition $\Theta = U\Lambda U^{-1}$ where $\Lambda = \text{diag}(\lambda_1, \ldots, \lambda_O)$ contains eigenvalues in ascending order and $U = [\boldsymbol{u}_1, \ldots, \boldsymbol{u}_O]$ contains corresponding eigenvectors.\footnote{The non-symmetry of $\Theta$ requires distinguishing between right eigenvectors (columns of $U$) and left eigenvectors (rows of $U^{-1}$). Empirically, all eigenvalues are distinct with $O$ linearly independent eigenvectors, ensuring: (i) diagonalizability, (ii) a complete basis spanning $\mathbb{R}^O$, and (iii) unique projection of shocks onto this basis given our normalization $\|\boldsymbol{u}_n\| = 1$.}

Each eigenvalue $\lambda_n$ represents the labor supply elasticity along its corresponding eigenvector $\boldsymbol{u}_n$—that is, how readily workers reallocate when relative wages change in the direction $\boldsymbol{u}_n$. This transforms the complex $O \times O$ substitution matrix into $O$ independent dimensions, each with its own elasticity.

\begin{lemma}[Eigenvalue Properties]\label{lemma:eigenvalues}
The labor supply elasticity matrix $\Theta$ satisfies:
\begin{enumerate}
    \item All eigenvalues are non-negative: $\lambda_n \geq 0$ for all $n$
    \item Exactly one zero eigenvalue: $\lambda_1 = 0$ with eigenvector $\boldsymbol{u}_1 \propto \mathbf{1}$
    \item Remaining eigenvalues are strictly positive: $\lambda_n > 0$ for $n > 1$
\end{enumerate}
\end{lemma}

\begin{proof}
The zero eigenvalue follows from the row sum property $\sum_{o'}\Theta_{oo'} = 0$. Non-negativity follows from gross substitutes. See Appendix \ref{app_ss:eigenvalue_proof}.
\end{proof}

The zero eigenvalue $\lambda_1 = 0$ reflects that uniform wage changes ($\boldsymbol{u}_1 \propto \mathbf{1}$) induce no labor reallocation since only relative wages matter for occupational choice. Positive eigenvalues $\lambda_n > 0$ measure labor supply elasticities for different directions of relative wage changes. Large eigenvalues indicate shock directions enabling extensive reallocation—workers have many unaffected alternatives. Small eigenvalues indicate limited mobility options—affected occupations and their natural alternatives face similar shocks.

\begin{proposition}[Spectral Decomposition of Incidence]\label{prop:spectral}
Any technological shock decomposes uniquely into eigenshocks:
\begin{equation*}
    \frac{d\ln\boldsymbol{\alpha}}{\sigma} = \sum_{n=1}^O b_n \boldsymbol{u}_n
\end{equation*}
where weights $b_n$ can be recovered as the coefficients in a linear projection of the shocks onto basis $\boldsymbol{b}=\left(U^{\prime} U\right)^{-1} U^{\prime} \cdot (d\ln\boldsymbol{\alpha}/\sigma)$. The wage response is:
\begin{equation*}
    d\ln\boldsymbol{w} = \frac{d\ln y}{\sigma}\mathbf{1} + \sum_{n=1}^O \underbrace{\frac{\sigma}{\sigma + \lambda_n}}_{\text{pass-through}} b_n\boldsymbol{u}_n
\end{equation*}
\end{proposition}

\begin{proof}
Apply eigendecomposition to $\Delta = (\mathbf{I} + \Theta/\sigma)^{-1} = U(\mathbf{I} + \Lambda/\sigma)^{-1}U^{-1}$. See Appendix \ref{app_ss:spectral_proof}.
\end{proof}

The pass-through factor $\sigma/(\sigma + \lambda_n)$ generalizes the classic one-dimensional incidence formula to a multi-dimensional occupational setting. Our spectral decomposition reveals that each shock direction has its own effective elasticity $\lambda_n$, generating heterogeneous incidence across different shock distributions. When technological shocks align with low-elasticity dimensions (small $\lambda_n$), workers cannot escape through reallocation, generating near-complete pass-through to wages. When shocks align with high-elasticity dimensions (large $\lambda_n$), extensive worker mobility dissipates the impact through employment adjustment. This decomposition shows why shock distribution matters: technological changes loading heavily on low-elasticity eigenvectors—those affecting clusters of skill-similar occupations—create maximal wage effects with minimal offsetting mobility.

\subsubsection{Illustration: Clustered versus Dispersed Shocks}

To illustrate these results, consider four occupations organized in two skill clusters: cognitive ($c_1, c_2$) and manual ($m_1, m_2$). Workers' productivity follows a nested structure with within-cluster correlation $\rho \in [0,1)$:
\begin{equation*}
    \Pr[\boldsymbol{\epsilon}(i) \leq \boldsymbol{\epsilon}] = \exp\left[-\left(\epsilon_{c_1}^{\frac{-\theta}{1-\rho}} + \epsilon_{c_2}^{\frac{-\theta}{1-\rho}}\right)^{1-\rho} - \left(\epsilon_{m_1}^{\frac{-\theta}{1-\rho}} + \epsilon_{m_2}^{\frac{-\theta}{1-\rho}}\right)^{1-\rho}\right]
\end{equation*}

This structure generates high substitutability within clusters but limited substitution across them. With equal initial employment shares, the eigendecomposition yields:
\begin{equation*}
    \boldsymbol{\lambda} = \begin{pmatrix}
        0 \\
        \theta \\
        \theta/(1-\rho) \\
        \theta/(1-\rho)
    \end{pmatrix}, \quad
    U = \frac{1}{2}\begin{pmatrix}
        1 & 1 & 1 & 1 \\
        1 & 1 & -1 & -1 \\
        1 & -1 & 1 & -1 \\
        1 & -1 & -1 & 1
    \end{pmatrix}
\end{equation*}

Three distinct shock patterns emerge:
\begin{itemize}
    \item $\boldsymbol{u}_1 = (1,1,1,1)'$: Uniform shocks ($\lambda_1 = 0$) with complete pass-through
    \item $\boldsymbol{u}_2 = (1,1,-1,-1)'$: Cross-cluster shocks ($\lambda_2 = \theta$) affecting cognitive and manual occupations oppositely
    \item $\boldsymbol{u}_3, \boldsymbol{u}_4$: Within-cluster shocks ($\lambda = \theta/(1-\rho)$) with differential effects within each cluster
\end{itemize}

The cross-cluster shock $\boldsymbol{u}_2$ has the smallest positive eigenvalue, yielding pass-through $\sigma/(\sigma + \theta)$. When $\theta$ is small (limited overall mobility) or $\sigma$ is large (flexible demand), this approaches complete pass-through. Crucially, workers displaced from cognitive occupations find their natural alternatives—other cognitive occupations—similarly affected, constraining mobility and amplifying wage disparity.

Within-cluster shocks achieve better adjustment. With eigenvalue $\theta/(1-\rho)$, pass-through becomes $\sigma(1-\rho)/[\sigma(1-\rho) + \theta]$. Higher within-cluster correlation $\rho$ increases the eigenvalue, enabling more reallocation because workers can transition to unaffected occupations in the same cluster. When one cognitive occupation faces a negative shock while another remains stable, high correlation within the cognitive cluster facilitates movement between them.

This example crystallizes why technological clustering matters. When automation or AI concentrates in skill-adjacent occupations, aligning with low-eigenvalue eigenvectors, it generates maximal wage adjustment with minimal offsetting mobility. The next section formalizes this intuition through a distance-dependent substitution structure in high-dimensional occupational space.

\subsection{Distance-Dependent Elasticity of Substitution} \label{ss:dides}

The spectral analysis revealed why technological shocks clustered in skill space generate severe wage inequality. We now move from the illustrative 2×2 example to the full complexity of real labor markets with hundreds of occupations and multiple skill dimensions. The key challenge is maintaining tractability while capturing realistic substitution patterns. We achieve this through a DIDES framework with a cross-nested constant elasticity of substitution (CNCES) functional form \citep{Lind2023-rl} that embeds distance-dependent substitution via a low-dimensional latent skill structure. 

\subsubsection{Latent Skill Formulation}

\paragraph{Microfoundation: Skills and Occupational Productivity}

Workers possess a vector of latent skills $s \in \mathcal{S}$. For each skill, they draw productivity across occupations from a correlated Fréchet distribution:
\begin{equation*}
    \Pr[\epsilon_1^s(i) \leq \epsilon_1^s, \ldots, \epsilon_O^s(i) \leq \epsilon_O^s] = \exp\left[-\left(\sum_{o=1}^O (\epsilon_o^s)^{\frac{-\theta}{1-\rho_s}}\right)^{1-\rho_s}\right]
\end{equation*}
where skill-specific correlation coefficient $\rho_s \in [0,1)$ governs skill transferability. This parameter captures a fundamental aspect of human capital: some skills transfer seamlessly across occupations while others are context-specific. General cognitive abilities (problem-solving, analytical thinking) typically exhibit high transferability (large $\rho_s$), while occupation-specific manual techniques (operating particular machinery, specialized surgical procedures) show low transferability (small $\rho_s$).

Occupations differ in their skill utilization. Let $A_o^s$ denote occupation $o$'s productivity when employing skill $s$. Workers optimally deploy their skills, achieving productivity:
\begin{equation*}
    \epsilon_o(i) = \max_{s \in \mathcal{S}} A_o^s \cdot \epsilon_o^s(i)
\end{equation*}

This max operator captures how workers sort into occupations based on comparative advantage. Different occupations require different skill combinations: data analysis demands strong cognitive skills, construction requires manual dexterity, and sales positions need interpersonal abilities. The parameters $\{A_o^s\}$ encode these occupation-specific skill productivity. Workers with exceptional manual dexterity but modest cognitive skills achieve the highest productivity in manual-intensive occupations where $A_o^{\text{manual}}$ is large. Conversely, cognitively gifted workers maximize productivity in occupations with high $A_o^{\text{cognitive}}$. This generates endogenous sorting: workers self-select into occupations that best utilize their skill endowments, with the occupation-skill match determining productivity.

\paragraph{DIDES Structure}

The microfoundation yields a tractable aggregate structure:

\begin{proposition}[DIDES through Cross-Nested CES]\label{prop:cnces}
The joint productivity distribution across occupations follows:
\begin{equation*}
    \Pr[\epsilon_1(i) \leq \epsilon_1, \ldots, \epsilon_O(i) \leq \epsilon_O] = \exp[-F(A_1\epsilon_1^{-\theta}, \ldots, A_O\epsilon_O^{-\theta})]
\end{equation*}
with correlation function:
\begin{equation}
    F(x_1, \ldots, x_O) = \sum_{s \in \mathcal{S}} \left[\sum_{o=1}^O (\omega_o^s x_o)^{\frac{1}{1-\rho_s}}\right]^{1-\rho_s} \label{eq:corr_cnces}
\end{equation}
where $A_o = \sum_s (A_o^s)^{\theta}$ is occupation $o$'s overall labor productivity and $\omega_o^s = (A_o^s)^{\theta}/A_o$ represents occupation $o$'s skill intensity in dimension $s$.
\end{proposition}

\begin{proof}
See Appendix \ref{app_ss:cnces_proof}.
\end{proof}

The skill intensities $\{\omega_o^s\}$ map occupations into skill space. Each $\omega_o^s$ measures how intensively occupation $o$ relies on skill $s$: data analysts and financial analysts both exhibit high cognitive intensity, locating them near each other in this space, while construction workers have high manual intensity, placing them in a distant region. This skill space geography determines substitution patterns through two mechanisms:

\begin{itemize}
    \item \textbf{Proximity effect:} Occupations with similar skill intensities are strong substitutes.
    \item \textbf{Transferability effect:} High $\rho_s$ amplifies substitution between occupations sharing skill $s$.
\end{itemize}

The key feature of Proposition \ref{prop:cnces} is that it achieves remarkable dimensionality reduction. The full substitution matrix requires $O^2$ parameters—with 300 occupations, this means 90,000 bilateral elasticities. Our framework collapses this to $S + 1$ structural parameters ($S$ skill-specific correlation parameters $\{\rho_s\}$ and one cross-skill dispersion parameter $\theta$) plus $300 \times S$ skill intensities $\{\omega_o^s\}$. Crucially, as I show below, the skill intensities can be measured directly from occupational data, leaving only the structural parameters to be estimated. For three skills (cognitive, manual, interpersonal), we estimate just four parameters while capturing rich substitution patterns across hundreds of occupations.

To be clear, "skills" in this framework are not primitive worker characteristics but rather a parsimonious device for parameterizing how productivity correlates across occupations based on skill similarity.

\subsubsection{Employment and Substitution Structure}

The DIDES framework generates explicit expressions for employment shares and substitution elasticities, revealing how distance in skill space governs labor market outcomes.

\begin{proposition}[Employment and Elasticities]\label{prop:cnces_elasticity}
Under DIDES, occupational employment shares decomposed as:
\begin{equation}
    \pi_o = \sum_{s \in \mathcal{S}} \pi_o^s = \sum_{s \in \mathcal{S}} \underbrace{\pi_o^{s,W}}_{\text{within-skill share}} \cdot \underbrace{\pi^{s}}_{\text{between-skill share}}
    \label{eq:cnces_shares}
\end{equation}
where:
\begin{align*}
    \pi_o^{s,W} &= \frac{(\omega_o^s A_o w_o^{\theta})^{\frac{1}{1-\rho_s}}}{\sum_{o'} (\omega_{o'}^s A_{o'} w_{o'}^{\theta})^{\frac{1}{1-\rho_s}}} \quad \text{(occupation $o$'s share among skill-$s$ users)} \\
    \pi^{s} &= \frac{\left[\sum_{o'} (\omega_{o'}^s A_{o'} w_{o'}^{\theta})^{\frac{1}{1-\rho_s}}\right]^{1-\rho_s}}{\sum_{s'} \left[\sum_{o'} (\omega_{o'}^{s'} A_{o'} w_{o'}^{\theta})^{\frac{1}{1-\rho_{s'}}}\right]^{1-\rho_{s'}}} \quad \text{(skill $s$'s share of workforce)}
\end{align*}

The correlated substitution component in \eqref{eq:elasticity_matrix} is:
\begin{equation}
    \theta\frac{x_{o'}F_{oo'}}{F_o}\bigg|_{x_j = A_jw_j^{\theta}} = -\theta \sum_{s \in \mathcal{S}} \frac{\rho_s}{1-\rho_s} \cdot \pi_o^{s,W} \pi_{o'}^{s,W} \cdot \frac{\pi^{s}}{\pi_o}
    \label{eq:cnces_elasticity}
\end{equation}
\end{proposition}

\begin{proof}
See Appendix \ref{app_ss:cnces_elasticity_proof}.
\end{proof}

The employment decomposition in equation \eqref{eq:cnces_shares} shows that occupational employment share $\pi_o$ aggregates skill-specific contributions $\pi_o^s$, each equaling the product of within-skill share $\pi_o^{s,W}$ (occupation $o$'s share among skill-$s$ users) and between-skill share $\pi^{s}$ (skill $s$'s workforce share). 

The elasticity formula \eqref{eq:cnces_elasticity} reveals how skill distance determines substitutability. The product $\pi_o^{s,W} \pi_{o'}^{s,W}$ measures skill overlap between occupations, while $\rho_s/(1-\rho_s)$ scales this overlap by transferability. High $\rho_s$ amplifies substitution even with modest overlap, while low $\rho_s$ limits substitution despite substantial overlap. Two data analysts at different firms (high overlap, high transferability) are strong substitutes; a data analyst and welder (low overlap, low transferability) are not.

This structure explains why technological clustering in skill-adjacent occupations limits employment adjustment. When automation concentrates in manual-intensive occupations, displaced workers face a mobility trap. Their high within-skill shares ($\pi_o^{\text{manual},W}$ large) indicate concentration in manual occupations. Clustering ensures their natural alternatives (other manual occupations) face similar negative shocks, forcing small employment absorption with large wage adjustments.

The framework nests standard models as special cases. When $\rho_s = 0$ for all skills (no correlation), the model reduces to CES with uniform elasticity $\theta$. Our framework generalizes nested CES models where each occupation belongs exclusively to one nest. Traditional nested CES requires pre-specifying exclusive occupation groups (manufacturing versus services, routine versus non-routine). In contrast, DIDES allows occupations to draw from multiple skills with varying intensities $\{\omega_o^s\}$, measured directly from occupational data. This flexibility proves crucial: data reveal that most occupations blend multiple skills, and these continuous skill intensities (rather than discrete categories) determine substitution patterns.

\subsection{Heterogeneous Workers} \label{ss:hetero_workers}

Our baseline model assumes workers are ex-ante identical, differing only in their idiosyncratic productivity draws. However, individuals may differ systematically across demographic groups, age cohorts, or education levels, possessing different comparative advantages across occupations. We now extend the framework to incorporate such systematic heterogeneity, allowing us to study how technological change affects different segments of the workforce differently. In Section~\ref{ss:distortions}, we examine how automation changed inequality between demographic groups.

Consider demographic groups $g \in G$ (e.g., race $\times$ gender combinations) that differ in their occupational productivity distributions. Each group draws productivity from:
\begin{equation*}
\Pr[\boldsymbol{\epsilon}^g(i) \leq \boldsymbol{\epsilon}] = \exp\left[-F\left(A_1^g \epsilon_1^{-\theta}, \ldots, A_O^g \epsilon_O^{-\theta}\right)\right]
\end{equation*}
where $A_o^g$ represents group $g$'s average productivity in occupation $o$. While comparative advantages $\{A_o^g\}$ vary across groups, the correlation function $F$ and dispersion parameter $\theta$ remain common, preserving the underlying substitution structure.\footnote{Group-specific employment shares are denoted $\pi_o^g$, yielding group-specific elasticity matrices $\Theta^g$.}

The productivity differences $\{A_o^g\}$ can arise from multiple sources—labor market discrimination, differences in skill endowments, or heterogeneous preferences for job amenities. The source of these differences does not affect the substitution patterns: given observed employment distributions, groups with identical employment shares $\{\pi_o^g\}$ exhibit identical substitution elasticities, regardless of whether these shares arise from discrimination, productivity, or preferences. The elasticity matrix $\Theta^g$ depends only on the equilibrium employment distribution\footnote{The group employment shares are sufficient statistic as shown in Appendix~\ref{app_ss:hat_algebra}.}, not on its underlying causes.

This group heterogeneity serves two purposes in our analysis. First, it enables us to study heterogeneous impacts of technological change across demographic groups. A group concentrated in manual occupations experiences automation differently than one concentrated in cognitive occupations, revealing how clustering interacts with initial employment distributions to generate unequal outcomes. Second, this heterogeneity provides identifying variation for estimation: different groups exhibit distinct substitution patterns based on their occupational employment distributions, with their differential reallocation responses to the same wage changes helping to identify elasticity parameters.

\section{Measurement and Estimation}\label{s:empirics}

This section empirically formulates and estimates the DIDES framework through two steps. First, we measure key model inputs: occupational skill intensities ($\omega_o^s$) from O*NET descriptors and technological exposures, denoted $\boldsymbol{z}^{\text{Automation}}$ and $\boldsymbol{z}^{\text{AI}}$, through task-level evaluations of automation and AI feasibility. These measurements reveal that both technologies cluster within skill-adjacent occupations—automation concentrates in manual-intensive jobs while AI concentrates in cognitive-intensive ones. Second, we estimate the structural parameters $\{\theta, \{\rho_s\}_{s\in\mathcal{S}}\}$ by exploiting how occupational employment responded to automation-induced wage changes between 1980 and 2010.

\subsection{Data and Measurement}

The primary data source for measuring both skill intensities and occupational exposure to technologies is O*NET (the Occupational Information Network).\footnote{The O*NET database, maintained by the U.S. Department of Labor, provides comprehensive data on occupational characteristics, worker skills, and job requirements across a wide range of professions: \url{https://www.onetonline.org/}.} O*NET provides two key elements: (i) skill intensities, which define an occupation's location in the skill space of labor supply, and (ii) task descriptions, which allow measurement of exposure to automation and AI.

\paragraph{Occupational Skill Intensities}

The theoretical framework requires measures of skill intensities $\{\omega_o^s\}$ that map occupations into a low-dimensional skill space. To operationalize this concept, we follow \cite{Lise2020-hm} and extract skill intensities directly from O*NET data rather than estimating them (see Appendix \ref{b:appendix:skills} for detailed methodology).

To extract the main skill dimensions and reduce dimensionality, we apply Principal Component Analysis (PCA) to approximately 200 O*NET descriptors covering skills, abilities, knowledge, work activities, and work context. Following \cite{Lise2020-hm}, we reduce these to three interpretable dimensions through exclusion restrictions: (i) mathematics scores load exclusively onto \textit{cognitive} intensity, (ii) mechanical knowledge onto \textit{manual} intensity, and (iii) social perceptiveness onto \textit{interpersonal} intensity.\footnote{The three principal components explain 58\% of total variation, with cognitive skills accounting for 35.6\%, manual skills 15.2\%, and interpersonal skills 6.9\%.} These orthogonal dimensions align with the model's assumption of independent skill-specific productivity distributions. The O*NET skill descriptors have cardinal meaning, measuring skill intensity on quantitative scales that correspond to our theoretical object $\omega_o^s$—the share of occupational productivity attributable to each skill dimension.

To construct the skill intensity parameters $\omega_o^s$ that enter the correlation function $F$, we first rescale principal component loadings to skill indices $r_o^s \in [0,1]$ using linear transformations that preserve relative distances between occupations.\footnote{Linear transformations preserve the distance metric in skill space—a key feature for DIDES. Converting to ranks would impose uniform spacing between adjacent occupations, eliminating meaningful variation in skill proximity.} We then compute the final skill intensities as variance-weighted shares:
\begin{equation*}
\omega_o^s = \frac{r_o^s \times \text{Var}_s}{\sum_{s' \in \mathcal{S}} r_o^{s'} \times \text{Var}_{s'}}
\end{equation*}
where the weight $\text{Var}_s$ is the variance explained by skill $s$ that preserves the empirical salience of each component. This formulation ensures $\sum_s \omega_o^s = 1$ for each occupation, consistent with the theoretical requirement that $\omega_o^s = (A_o^s)^\theta/A_o$ represents relative skill intensity. Table \ref{tab:occupation_examples} provides illustrative examples.

\begin{table}[ht]
\centering
\caption{Skill Intensities and Technological Exposures for Selected Occupations}
\label{tab:occupation_examples}
\resizebox{\linewidth}{!}{
\begin{tabular}{lcccccc}
\toprule
& \multicolumn{3}{c}{Skill Intensities} & & \multicolumn{2}{c}{Technological Exposure} \\
\cmidrule{2-4} \cmidrule{6-7}
Occupation & Cognitive & Manual & Interpersonal & & AI & Automation \\
\midrule
Chief Executives & 0.71 & 0.11 & 0.18 & & 0.28 & 0.03 \\
Electrical Engineers & 0.73 & 0.19 & 0.08 & & 0.71 & 0.19 \\
Economists & 0.79 & 0.07 & 0.14 & & 0.86 & 0.31 \\
Licensed Practical Nurses & 0.52 & 0.26 & 0.22 & & 0.08 & 0.47 \\
Textile Machine Operators & 0.52 & 0.47 & 0.01 & & 0.02 & 0.51 \\
\bottomrule
\end{tabular}
}
\note{\textit{Notes:} Skill intensities ($\omega_o^s$) represent the relative importance of cognitive, manual, and interpersonal skills for each occupation, with values summing to 1.0 across the three dimensions. Technological exposure measures indicate the share of tasks within each occupation that can potentially be performed by AI (generative models) or automation (robots, machines, and rule-based software) without human intervention.}
\end{table}

\paragraph{Occupational Exposure to Technologies}

To estimate structural parameters and assess incidence, we construct measures of occupational exposure to automation and AI ($\boldsymbol{z}^{\text{Automation}}$ and $\boldsymbol{z}^{\text{AI}}$). 

Several related measures exist for occupational exposure to automation \citep{Autor2013-zm, Acemoglu2022-lv, Autor2024-ay}. In contrast, measuring occupational exposure to AI presents unique challenges, as its full labor market impact has yet to materialize. To construct forward-looking measures, we follow \cite{Eloundou2024-cu} and leverage ChatGPT to evaluate task-level automation and AI feasibility.\footnote{This LLM-based approach has been validated by subsequent studies. \cite{Bick2024-eu} and \cite{Tomlinson2025-va} demonstrate high correlations between LLM task evaluations and ex-post real-world generative AI adoption patterns. Most notably, \cite{brynjolfsson2025canaries} find that LLM exposure measures predict actual employment declines: early-career workers (ages 22-25) in the most AI-exposed occupations have experienced a 13\% relative decline in employment since widespread AI adoption.}

Specifically, we query ChatGPT on whether each task in O*NET's database (covering 19,200 tasks across 862 occupations) can be performed without human intervention by: (i) industrial robots, machinery, and software without AI capabilities (representing traditional automation exposure) or (ii) generative AI models like ChatGPT (representing AI exposure). ChatGPT estimates that approximately 6,000 tasks—one-third of the total—can potentially be performed by AI, a magnitude comparable to automation technologies.

\begin{table}[ht]
\centering
\caption{Task-Level Evaluation of Automation and AI Exposure}
\label{tab:task_eval_example}
\begin{tabular}{lcc}
\toprule
Task Description & Automation & AI \\
\midrule
\textbf{Economists, Market and Survey Researchers} & & \\
\quad Explain economic impact of policies to the public & No & Yes \\
\quad Supervise research projects and students' study projects & No & No \\
\quad Teach theories, principles, and methods of economics & No & Yes \\
\addlinespace
\textbf{Textile Sewing Machine Operators} & & \\
\quad Remove holding devices and finished items from machines & Yes & No \\
\quad Cut materials according to specifications, using tools & Yes & No \\
\quad Record quantities of materials processed & Yes & Yes \\
\bottomrule
\end{tabular}
\note{\textit{Notes:} This table presents examples of task-level evaluations using ChatGPT. Automation exposure is assessed by asking: ``Can industrial robots, machines, and computers (no AI capability) perform this task without human intervention?'' AI exposure is determined by querying: ``Can generative AI (e.g., large language models like ChatGPT) potentially perform this task without human intervention?'' Each task receives a binary classification.}
\end{table}

Table \ref{tab:task_eval_example} provides examples of task evaluations for two occupations: economists and sewing machine operators. This classification distinguishes automation-exposed tasks, which involve well-defined, rule-based processes susceptible to mechanization, from AI-exposed tasks, which primarily involve inductive reasoning, complex decision-making, and non-physical cognitive work. The latter pattern aligns with Polanyi's Paradox—many cognitive tasks resist codification into explicit rules, making them more amenable to AI than traditional automation \citep{Autor2015-ff}.

Using these task-level evaluations, we compute the share of tasks within each occupation that is either automatable or AI-exposed according to ChatGPT's evaluation, forming our occupational exposure measures $\boldsymbol{z}^{\text{Automation}}$ and $\boldsymbol{z}^{\text{AI}}$. Table \ref{tab:occupation_examples} reports automation and AI exposure levels for selected occupations. Additional methodological details and validation against existing measures are provided in Appendix \ref{b:appendix:expos}.

\paragraph{Technological Exposure in Skill Space}

We now demonstrate that technological exposure clusters in skill space: occupations with similar skill intensities face similar levels of automation and AI exposure. Consistent with existing research showing that manual-intensive occupations are more susceptible to automation \citep{Autor2003-jz}, our ChatGPT evaluations confirm this relationship. Panel (a) of Figure \ref{f:expo_skill_space} demonstrates that automation exposure increases with manual skill index and decreases with cognitive index. Conversely, Panel (b) reveals that AI exposure follows the opposite pattern: cognitive-intensive occupations face greater vulnerability to AI, as these technologies increasingly perform complex analytical and decision-making tasks.

\begin{figure}[ht]
    \centering
    \subcaptionbox{skill intensities vs. Automation Exposure}{\includegraphics[scale=0.15]{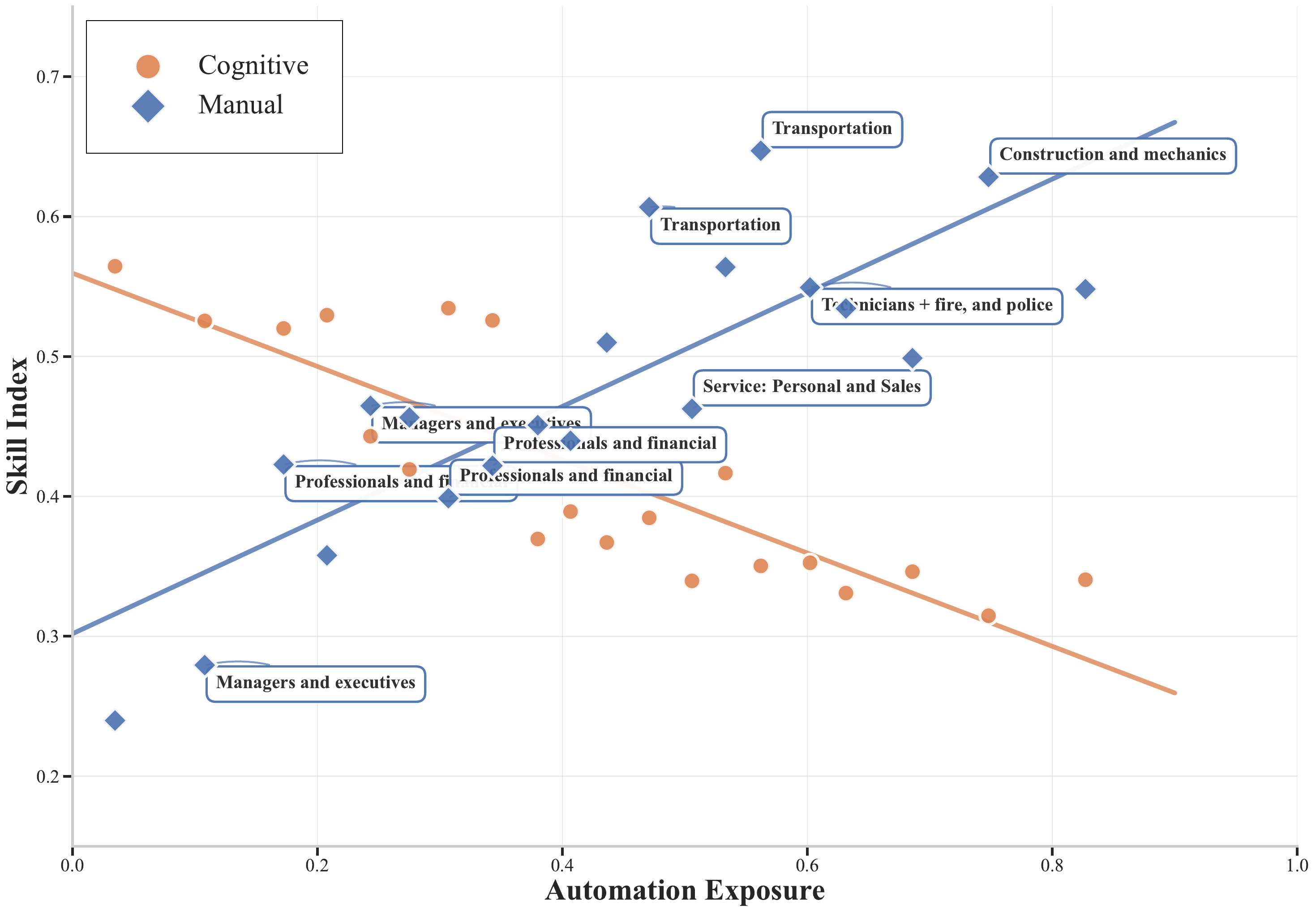}}\hfill
    \subcaptionbox{skill intensities vs. AI Exposure}{\includegraphics[scale=0.15]{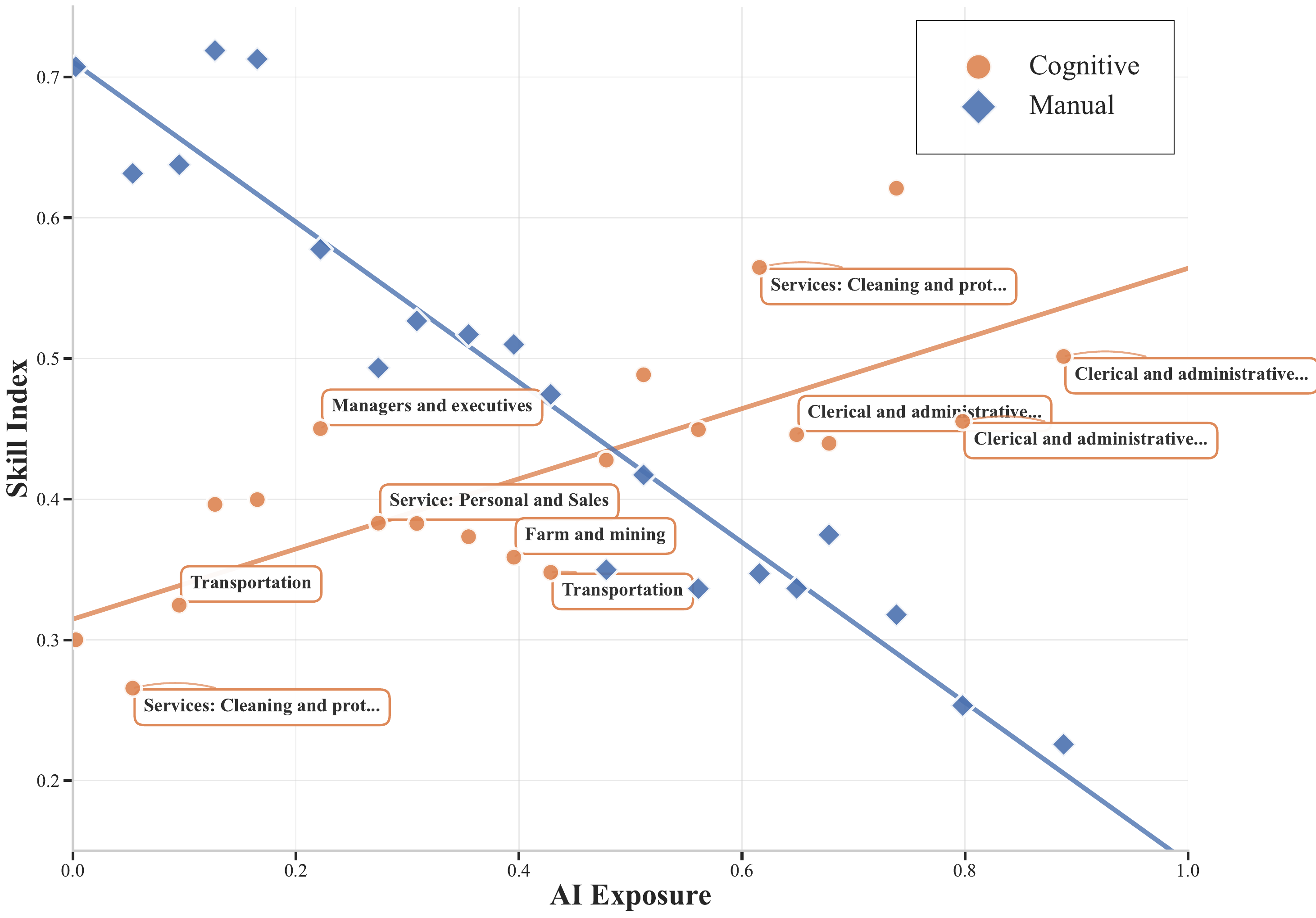}}\hfill
    \subcaptionbox{Automation Exposure in Skill Space}{\includegraphics[scale=0.25]{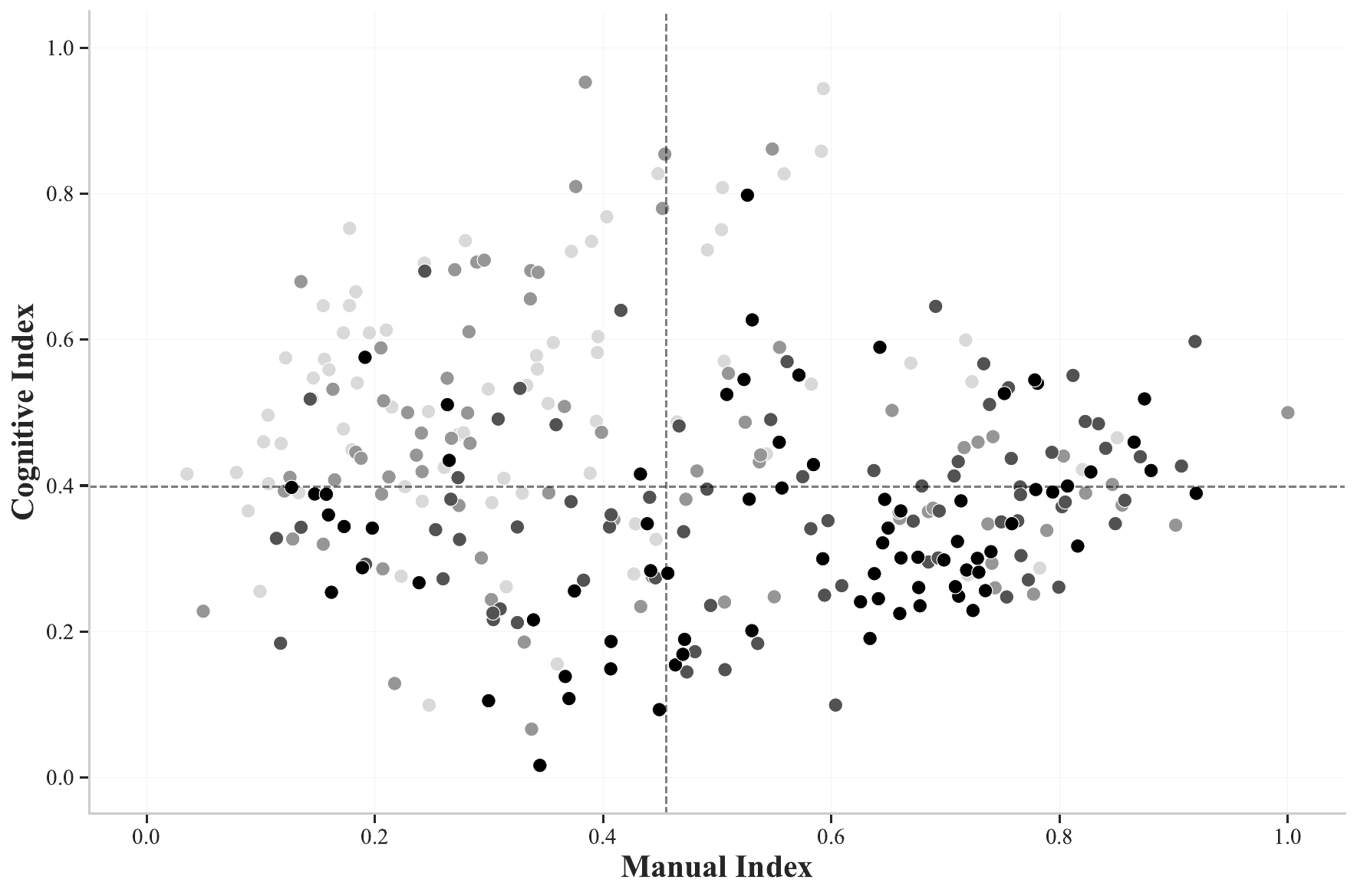}}\hfill
    \subcaptionbox{AI Exposure in Skill Space}{\includegraphics[scale=0.25]{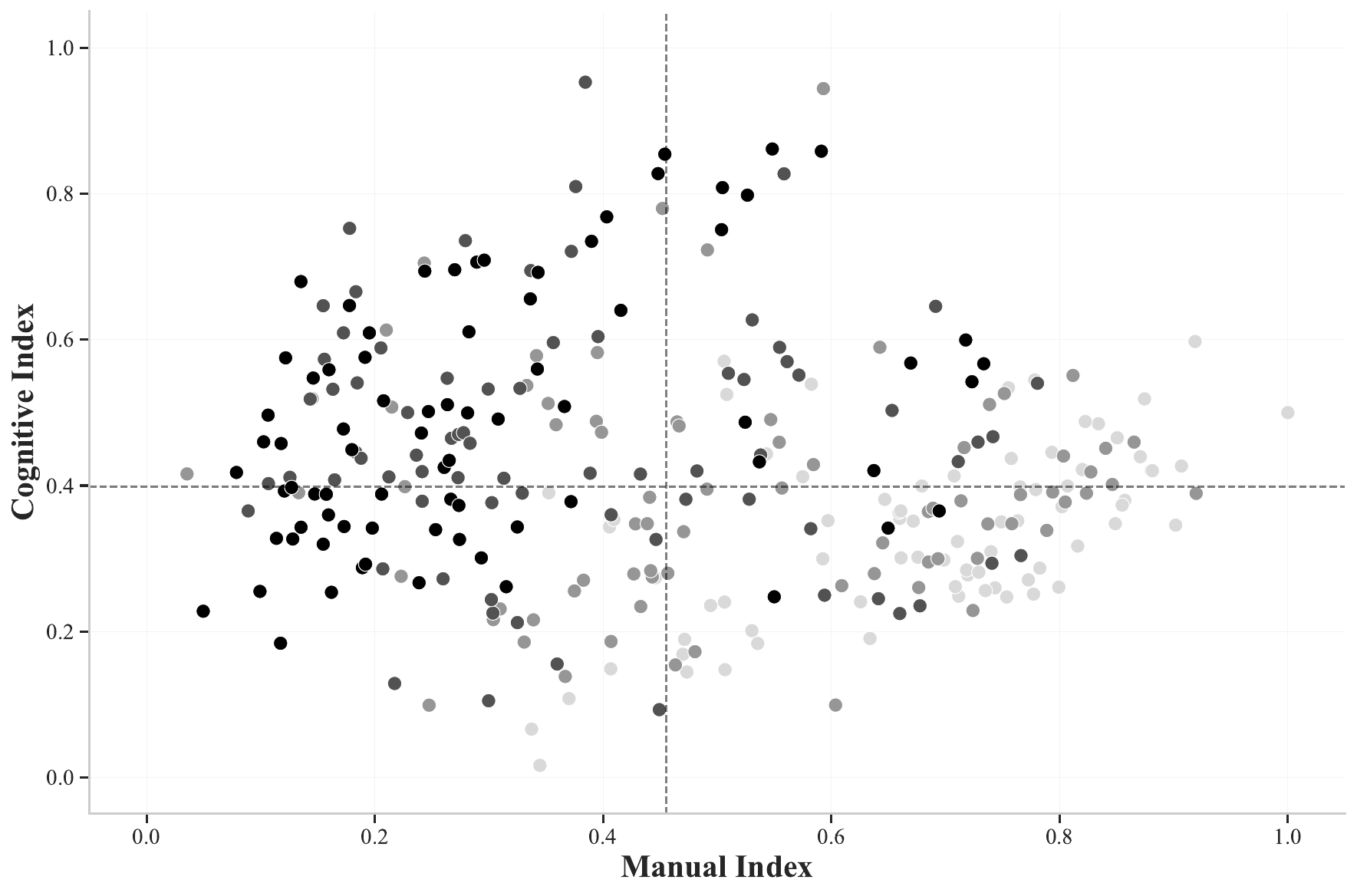}}
    \caption{Technological Exposure in Skill Space}
    \note{\textit{Notes:} This figure illustrates the distribution of automation and AI exposure across occupational skill space. Panels (a) and (b) present binscatter plots of occupational skill indices obtained from PCA ($r_o^s$) against technological exposure, each bin contains 15 occupations. Panels (c) and (d) visualize the same exposure patterns in two-dimensional cognitive-manual skill space, where darker shading indicates higher exposure levels.}
    \label{f:expo_skill_space}
\end{figure}

While automation and AI target distinct occupational segments, they share a critical feature: both technologies cluster within skill-adjacent occupations. Panels (c) and (d) of Figure \ref{f:expo_skill_space} visualize this clustering in cognitive-manual skill space, where darker shading indicates higher exposure. Automation concentrates in the lower-right region (high manual, low cognitive requirements), while AI clusters in the upper-left region (high cognitive, low manual requirements). This concentration has profound implications for labor market adjustment: as illustrated in Section \ref{ss:spectral}, clustering restricts worker mobility because displaced workers' natural alternatives—occupations requiring similar skills—face similar technological threats.

The choice of cognitive and manual dimensions reflects their empirical importance: together they account for 88\% of total skill variance across occupations.\footnote{Since cognitive and manual skills dominate occupational differentiation, they largely determine substitution patterns and mobility constraints.} Given this dominance, our descriptive analysis focuses on these two dimensions, while Appendix \ref{b:appendix:expos_int} examines technological exposure along the interpersonal dimension.

\subsection{Estimation of Structural Parameters}

With occupational skill intensities $\{\omega_o^s\}$ and automation exposure $\boldsymbol{z}^{\text{Automation}}$ measured, we now estimate the structural parameters $\{\theta, \{\rho_s\}_{s\in\mathcal{S}}\}$ that govern occupational substitution. Our estimation strategy exploits long-run employment responses to automation-induced wage changes across demographic groups between 1980 and 2010.

\subsubsection{Wage and Employment Effects of Automation}

We estimate automation-induced wage changes using the Panel Study of Income Dynamics (PSID) from 1976-2019.\footnote{Wage data for salaried workers are only available starting in 1976. The sample includes individuals aged 16–64, employed in nonagricultural, nonmilitary jobs, who are part of the core PSID sample (SRC). We exclude the oversample of low-income households (SEO sample) and the immigrant samples added in the 1990s to maintain sample consistency over time.} Following \cite{Cortes2016-nb}, we exploit within-individual job spell variation to address selection concerns and worker composition changes that plague cross-sectional average wage comparisons:\footnote{\cite{Cortes2016-nb} classifies occupations into three discrete groups: low-skill services, manufacturing, and high-skill services. We instead use continuous automation exposure for 306 occupations, providing richer variation for estimating supply elasticities.}

\begin{equation*}
\ln w_{i(o),t} = \beta_t \cdot z^{\text{Automation}}_o + \mathbf{X}'_{it}\gamma + \delta_{i,o} + u_{i,o,t}
\end{equation*}
where $\delta_{i,o}$ represents individual-occupation spell fixed effects and $\mathbf{X}_{it}$ includes year effects and time-varying individual characteristics. The individual-occupation spell fixed effects are crucial for addressing both selection and composition concerns. By tracking the same worker in the same occupation over time, we control for worker-specific occupational productivity—addressing selection bias from workers with different productivity systematically sorting into automation-exposed occupations. 

\begin{figure}[ht]
    \centering
    \subcaptionbox{Wage Effects from PSID}{\includegraphics[scale=0.25]{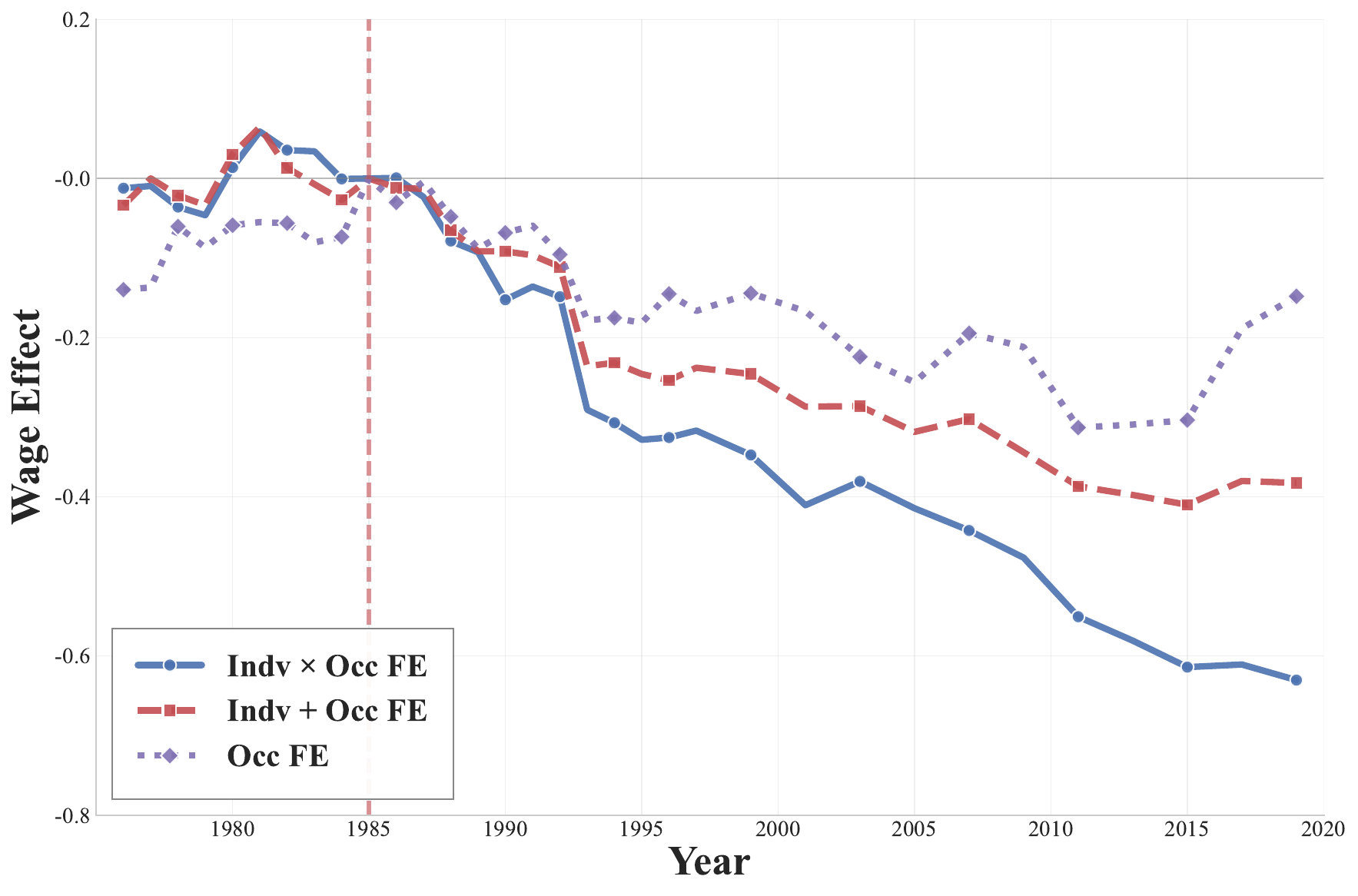}}\hfill
    \subcaptionbox{Employment Effects from Census}{\includegraphics[scale=0.25]{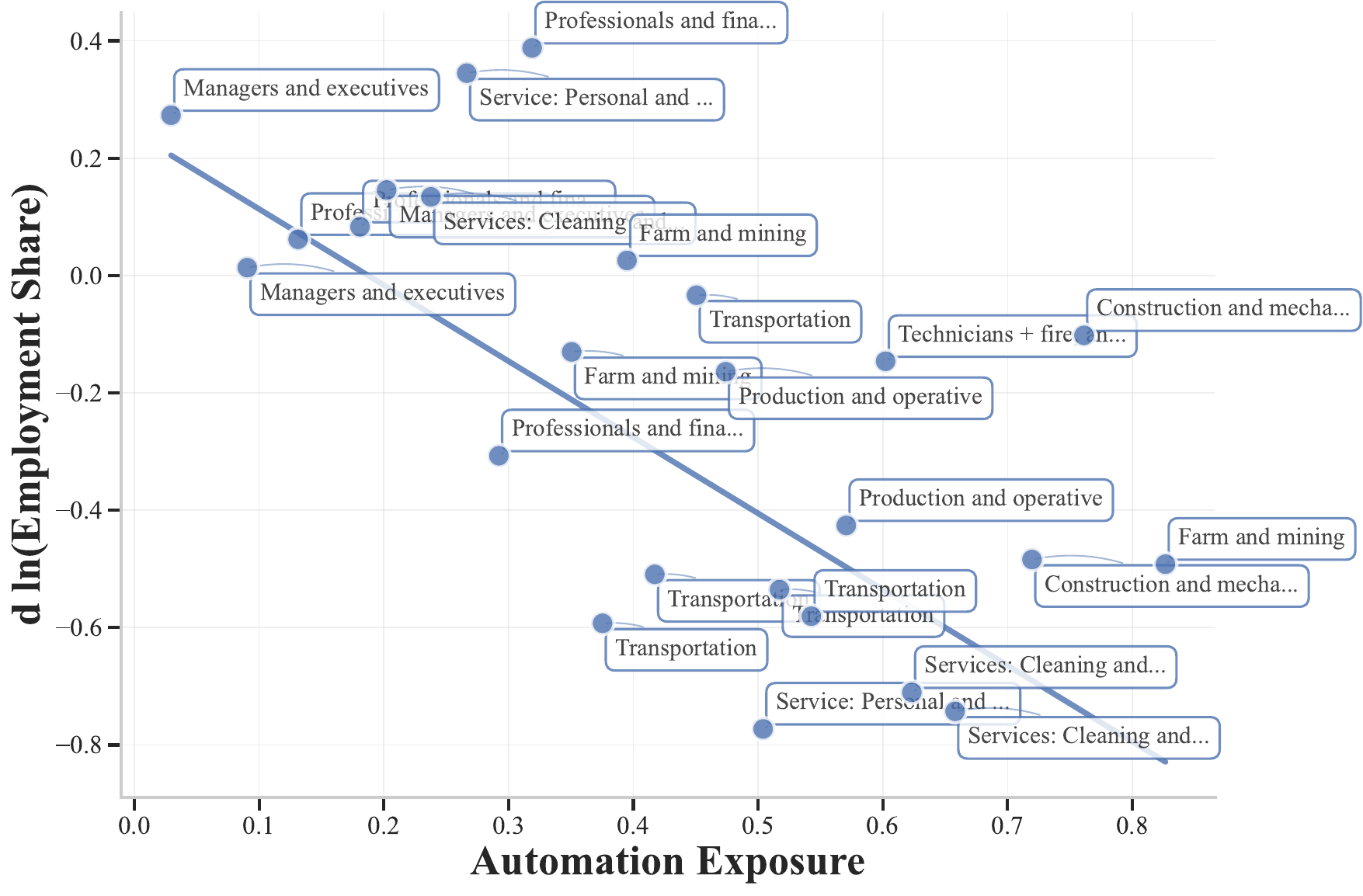}}
    \caption{Effects of Automation on Wages and Employment}
    \note{\textit{Notes:} Panel (a) shows estimated wage effects of automation exposure using PSID data. The solid blue line uses individual-occupation spell fixed effects to control for selection, while dashed lines show estimates with occupation fixed effects only (purple) or individual and occupation fixed effects separately (red). Panel (b) presents employment share changes from Census (1980-2000) and ACS (2010) data using a binscatter plot where each bin contains 10 occupations.}
    \label{f:auto_wage_emp}
\end{figure}

Figure \ref{f:auto_wage_emp} demonstrates the importance of controlling for selection. Panel (a) compares three specifications: (i) our preferred specification with job spell fixed effects (solid blue line), (ii) occupation fixed effects only (dashed purple line), and (iii) individual and occupation fixed effects separately (dashed red line). The specification without individual controls (purple line) yields only a 30 log point wage decline for maximally exposed occupations, while our preferred specification (blue line) shows a 60 log point decline by 2019. This stark difference reveals the extent of selection bias and worker composition changes: workers remaining in automation-exposed occupations have systematically higher occupational productivity than those who exit, causing cross-sectional comparisons to understate true wage effects. By purging this selection bias through job spell fixed effects, the blue line isolates the wage impact of automation on individual workers, providing unbiased estimates of true wage declines experienced within job matches. The red line, which includes individual and occupation fixed effects separately, accounts for composition changes but cannot fully control for selection into specific occupations, yielding intermediate estimates. Wage effects are remarkably similar for men and women, as shown in Appendix~\ref{b:appendix:wage_emp_auto}.

The timing of wage effects is consistent with automation accelerating in the mid-1980s.\footnote{The mid-1980s marked a turning point in automation adoption, coinciding with the widespread introduction of microprocessor-based manufacturing technologies and computerized machinery \citep{acemoglu2011skills}. \cite{Autor2003-jz} documents that routine task-intensive occupations began experiencing relative employment declines around this period.} Panel (a) shows no divergence in wage trends prior to 1985, followed by persistent and growing wage gaps between high- and low-exposure occupations. This pattern validates our identification strategy: occupations with different automation exposure experienced similar pre-trends before widespread automation adoption, then diverged systematically based on their exposure—precisely the pattern expected if automation causally affects wages in automatable occupations.

Panel (b) shows corresponding employment effects using Census and ACS data. Occupations with maximum automation exposure experienced a 100 log point decline in employment share relative to unexposed occupations. Crucially, the simultaneous decline in both wages and employment—rather than opposing movements—is consistent with automation operating as a negative labor demand shock. If automation were instead a supply phenomenon (workers voluntarily leaving certain occupations), we would observe rising wages in exposed occupations due to reduced labor supply. The parallel declines in wages and employment thus identify automation as reducing firms' demand for labor in affected occupations. These wage and employment responses provide the key moments for structural estimation.

\subsubsection{Estimation Strategy and Results}
\paragraph{Linear Representation and Identification}
While the full model is nonlinear, first-order approximation clarifies our identification strategy. Log-linearizing the labor supply system yields:
\begin{equation*}
\hat{\mathbf{L}}^g = \Theta^g(\theta, \boldsymbol{\rho}, \boldsymbol{\omega}, \mathbf{L}^g) \cdot \hat{\mathbf{w}} + \boldsymbol{\mu}^g
\end{equation*}
where the elasticity matrix $\Theta^g$ depends on structural parameters $\{\theta, \boldsymbol{\rho}\}$ and group-specific employment shares $\mathbf{L}^g$ (see equation \eqref{eq:cnces_elasticity}).

\begin{assumption}[Exclusion Restriction]
Automation exposure $\mathbf{z}^{\text{Automation}}$ is orthogonal to unobserved labor supply shocks:
\begin{equation*}
\mathbb{E}[\boldsymbol{\mu}^g|\mathbf{z}^{\text{Automation}}] = 0
\end{equation*}
where $\boldsymbol{\mu}^g$ represents unobserved labor supply shocks for group $g$.
\end{assumption}

In terms of our model's primitives, this assumption requires that automation exposure is uncorrelated with changes in occupation-specific productivity parameters $A_o$. The assumption is justified by the nature of automation exposure: it reflects technological feasibility—whether tasks can be codified into programmable rules—rather than shifts in workers' occupation-specific productivities. Moreover, the simultaneous decline in both wages and employment documented in Figure \ref{f:auto_wage_emp} confirms that automation operates as a labor demand shock rather than a supply shift. The timing evidence further supports this interpretation: Panel (a) shows no divergence in wage trends prior to 1985, followed by persistent wage gaps between high- and low-exposure occupations, consistent with automation adoption rather than pre-existing productivity trends.

\paragraph{Structural Estimation}

Given estimated wage changes $\{\hat{w}_o\} = \hat{\beta} \cdot \boldsymbol{z}^{\text{Automation}}$ from automation, we estimate elasticity parameters $\{\theta, \boldsymbol{\rho}\}$ via structural estimation. We start with a proposition establishing that observed employment shares and structural parameters are sufficient for computing employment changes given wage changes.

\begin{proposition}[Hat Algebra] \label{prop:hat_algebra}
Given relative wage changes $\hat{\mathbf{w}}$, correlation function $F$, and parameters $\{\theta, \boldsymbol{\rho}\}$, observed employment shares $\{\boldsymbol{\pi}_t^g\}_{g\in G}$ serve as sufficient statistics for counterfactual employment shares $\{\boldsymbol{\pi}_{t+1}^g\}_{g\in G}$ without requiring levels of wages or productivities.
\end{proposition}

\begin{proof}
See Appendix \ref{app_ss:hat_algebra} for proof and algorithm.
\end{proof}

This result allows us to express employment changes in terms of observed employment shares, generalizing the local elasticity result from Proposition \ref{prop:tech_incidence}. For each demographic group $g$, model-implied employment changes are:
\begin{equation*}
\hat{\pi}_{o,t}^g = \hat{\pi}_o^g\left(\theta, \boldsymbol{\rho}, \{\boldsymbol{\omega}^s\}_{s \in \mathcal{S}}, \boldsymbol{\pi}_{t}^g, \hat{\boldsymbol{w}}\right)
\end{equation*}
where hat variables denote proportional changes, $\hat{x}_t = x_{t+1} / x_{t}$. Importantly, wage changes $\{\hat{w}_o\}$ are common across all demographic groups, as automation operates as an occupational demand shock affecting all workers in a given occupation equally. Appendix \ref{b:appendix:wage_emp_auto} confirms that estimated wage effects are remarkably similar for men and women.\footnote{While wage changes are identical across groups, employment responses differ due to group-specific initial distributions $\{\pi_{o,t}^g\}$ across occupations. Groups concentrated in different parts of the occupational skill space respond differently to the same wage changes, providing the cross-group variation that identifies our structural parameters.}

We estimate parameters using pseudo-Poisson maximum likelihood, commonly used in gravity estimation \citep{Fally2015-tn,Lind2023-rl}:

\begin{equation*}
\{\hat{\theta}, \hat{\boldsymbol{\rho}}\} = \arg\min_{\theta, \boldsymbol{\rho}} \sum_{o,g} \kappa\left(\pi_{o,t+1}^g, \pi_{o,t}^g \cdot \hat{\pi}_{o,t}^g(\theta, \boldsymbol{\rho})\right)
\end{equation*}
where $\kappa(x, \hat{x}) = 2[x\ln(x/\hat{x}) - (x - \hat{x})]$ is the PPML objective function. The estimation embeds the exclusion restriction through the moment condition:
\begin{equation*}
\mathbb{E}\left[v_{o,t}^g \mid \boldsymbol{z}^{\text{Automation}}, \boldsymbol{\pi}_t^g\right] = 0
\end{equation*}
where $v_{o,t}^g = \pi_{o,t+1}^g/\pi_{o,t}^g\hat{\pi}_{o,t}^{g} - 1$ is the prediction error.

\begin{table}[ht]
\centering
\caption{PPML Estimation Results of Labor Supply Elasticities across Demographic Groups}
\label{tab:PPMLestimates}
\resizebox{\textwidth}{!}{%
\begin{tabular}{lcccccccc}
\hline\hline
                      & \multicolumn{4}{c}{1980--2000} & \multicolumn{4}{c}{1980--2010} \\ 
\cmidrule(lr){2-5} \cmidrule(lr){6-9}
                      & CES     & DIDES   & DIDES      & DIDES   & CES     & DIDES   & DIDES      & DIDES   \\
                      & All     & All     & White Men  & White   & All     & All     & White Men  & White   \\
                      &         &         & \& Women   & Women   &         &         & \& Women   & Women   \\
\hline
$\theta$              & 3.12    & 1.10    & 1.07       & 2.29    & 2.85    & 1.02    & 1.06       & 1.97    \\
                      & (0.20)  & (0.23)  & (0.47)     & (2.08)  & (0.20)  & (0.50)  & (0.46)     & (0.80)  \\
$\rho_{\text{Cog}}$   & 0       & 0.77    & 0.78       & 0.65    & 0       & 0.76    & 0.75       & 0.65    \\
                      & --      & (0.13)  & (0.09)     & (0.36)  & --      & (0.17)  & (0.22)     & (0.67)  \\
$\rho_{\text{Man}}$   & 0       & 0.48    & 0.50       & 0.39    & 0       & 0.44    & 0.45       & 0.38    \\
                      & --      & (0.18)  & (0.15)     & (0.44)  & --      & (0.23)  & (0.19)     & (0.52)  \\
$\rho_{\text{Int}}$   & 0       & 0.75    & 0.77       & 0.62    & 0       & 0.72    & 0.74       & 0.62    \\
                      & --      & (0.13)  & (0.14)     & (0.41)  & --      & (0.19)  & (0.14)     & (0.27)  \\
\hline
Observations          & 2,448   & 2,448   & 612        & 306     & 2,448   & 2,448   & 612        & 306     \\
\hline\hline
\end{tabular}%
}
\note{\textit{Notes:} Standard errors in parentheses. Following the literature, we scale the Poisson deviance by the mean-variance ratio of the data to obtain standard errors. While scaling does not affect the estimates, it aligns the deviance with the data variance. The CES specification imposes $\rho = 0$ across all occupation groups, while the DIDES specification allows for heterogeneous distance-dependent elasticities. Columns 1 and 5 report CES estimates for the full sample. Columns 2 and 6 report DIDES estimates for the full sample. Columns 3 and 7 restrict the sample to white men and women, while columns 4 and 8 further restrict to white women only. The estimates for the CES specification are close to those from a simple OLS regression with predicted wage effect as a regressor.}
\end{table}

Table \ref{tab:PPMLestimates} presents PPML estimation results. Column 1 shows the CES benchmark, which imposes $\rho_s = 0$ and yields $\hat{\theta} = 3.12$ (s.e. = 0.20)—the average elasticity under independent productivity draws across occupations. Column 2 presents our main DIDES specification using all demographic groups, which dramatically alters the results.

Three key findings emerge. First, cross-skill elasticity falls to $\hat{\theta} = 1.10$ (s.e. = 0.23), implying that approximately two-thirds of observed substitution occurs within skill clusters rather than across them. Second, skill transferability varies substantially: cognitive skills show the highest correlation parameter ($\hat{\rho}_{\text{Cog}} = 0.77$, s.e. = 0.13), followed by interpersonal skills ($\hat{\rho}_{\text{Int}} = 0.75$, s.e. = 0.13), while manual skills exhibit the lowest transferability ($\hat{\rho}_{\text{Man}} = 0.48$, s.e. = 0.18). Third, these parameters imply heterogeneous within-skill elasticities: $\theta/(1-\rho_s)$ equals 4.8 for cognitive occupations, 4.4 for interpersonal occupations, but only 2.1 for manual occupations.

Columns 3-8 demonstrate robustness across subsamples and time periods. Columns 3-4 show that estimates remain stable when restricting to white workers, though standard errors increase with smaller samples.\footnote{The larger but imprecise estimate for white women alone ($\hat{\theta} = 2.29$, s.e. = 2.08) likely reflects both identification and compositional effects. Cross-skill identification comes primarily from men, whose employment concentrates in manual occupations, providing clearer variation in response to clustering shocks. Rising female labor force participation during our sample period means cross-sectional variation partly captures new entrants who are inherently more flexible in their occupational choices, inflating the estimated elasticity. Women's more dispersed employment across skill clusters, however, provides valuable identifying variation for correlation parameters $\{\rho_s\}$. This complementarity in identification across demographic groups strengthens our overall estimates.} Columns 5-8 replicate the analysis using 1980--2010 employment changes, yielding nearly identical point estimates with slightly larger standard errors, confirming the robustness of our parameter estimates.

A natural question is whether a simpler approach might suffice. One could specify labor supply using nested CES, where occupations are partitioned into mutually exclusive categories—such as broad sectors (low-skill service, high-skill service, manufacturing) or skill intensity exclusive groupings (cognitive, manual, interpersonal)—with higher elasticities of substitution within nests than across them. To assess this alternative, we estimate nested CES specifications using the same procedure and test both nesting structures. Results detailed in Appendix \ref{b:appendix:nested_ces} reveal that most within-nest correlation parameters are statistically insignificant, while cross-nest elasticities (2.05--2.67) converge toward our CES benchmark of 3.12, effectively reducing nested CES to standard CES. The contrast with our DIDES estimates—where $\theta = 1.10$ and substantial skill-specific correlations emerge—demonstrates a key advantage of our approach: rather than requiring researchers to impose arbitrary categorical restrictions, we allow occupations to draw continuously from multiple skills with varying intensities $\{\omega_o^s\}$ measured directly from occupational data. This flexibility proves empirically critical, as most occupations blend multiple skills and substitution patterns reflect continuous skill proximity rather than discrete categories.
\subsection{The Topology of Occupational Substitution}

Our estimated parameters reveal the fundamental structure of labor market substitution. Figure~\ref{fig:occupation_topology_1980} visualizes the substitution topology among 306 occupations based on implied cross-wage elasticities from our DIDES model using 1980--2000 data.

\begin{figure}[ht]
    \centering
    \includegraphics[width = \linewidth]{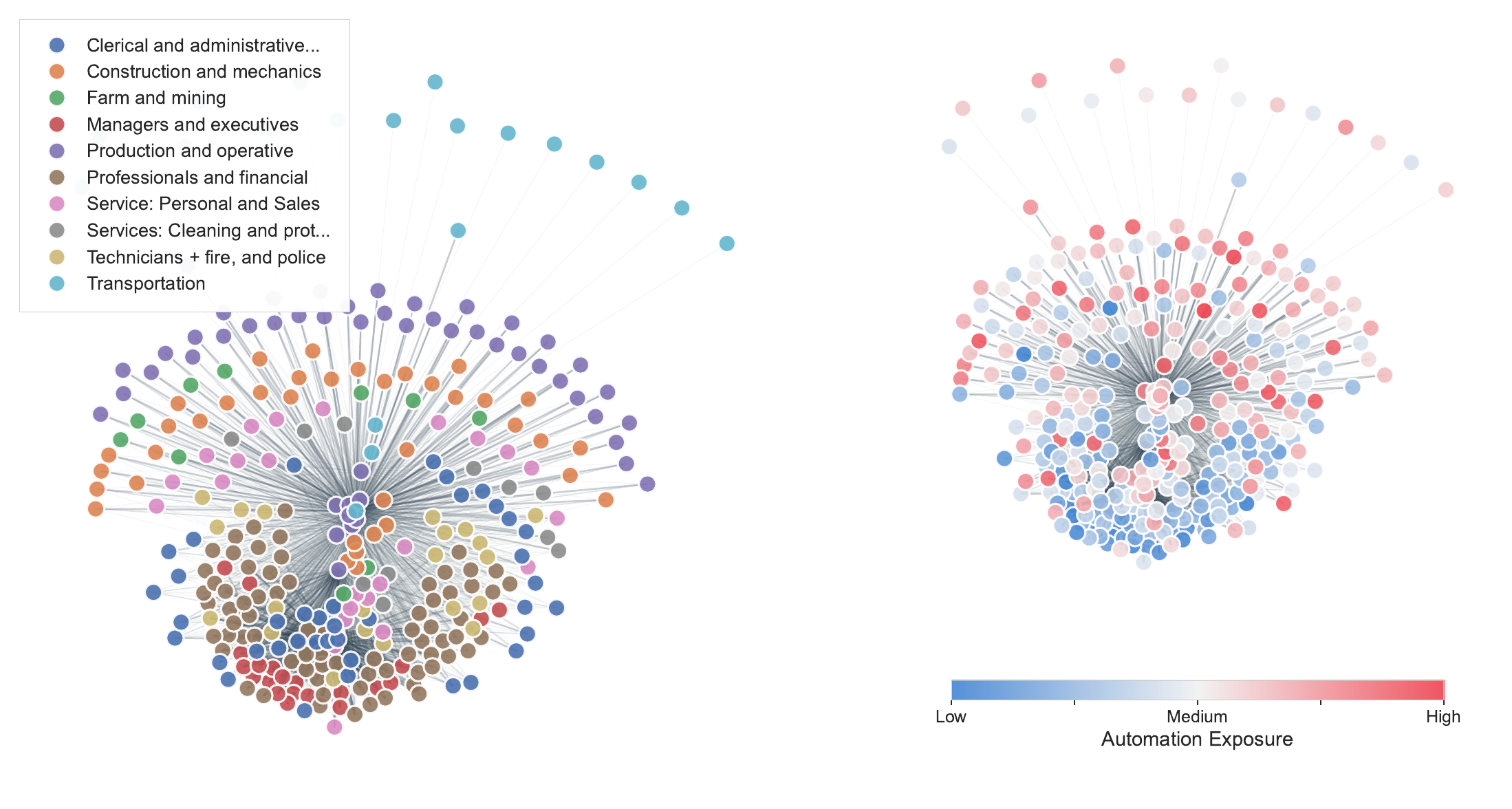}
    \caption{The Structure of Occupational Substitutability and Automation Exposure}
    \label{fig:occupation_topology_1980}
    \note{\textit{Notes:} This figure visualizes the substitution structure among 306 occupations based on estimated cross-wage elasticities from the DIDES model. Each node represents one occupation. Edges connect occupation pairs, with darker and thicker lines indicating stronger substitution relationships (top 20\% of elasticities shown). Node positions are determined using a force-directed layout algorithm that places more substitutable occupations closer together. The left panel colors nodes by one-digit Census occupation categories (e.g., professional, production, service occupations). The right panel maps automation exposure onto the same network structure, with colors ranging from blue (low exposure) to red (high exposure).}
\end{figure}

The left panel reveals how estimated elasticities organize occupations into distinct groupings. Each node represents one occupation, colored by its Census occupation category. Edges connect occupation pairs, with darker and thicker lines indicating stronger cross-wage elasticities—we display only the top 20\% of elasticities to highlight the strongest substitution relationships. Node positions emerge from a force-directed layout algorithm that places more substitutable occupations closer together, allowing the estimated elasticities themselves to determine the network structure.

The visualization reveals distinct groupings with strong within-group substitutability. Production and operative occupations (lower right) form dense interconnections, as do professional and financial occupations (upper left), while service occupations show more dispersed patterns. This structure emerges from our estimated correlation parameters—$\rho_{\text{cog}} = 0.77$, $\rho_{\text{man}} = 0.48$, and $\rho_{\text{int}} = 0.75$—rather than from imposed categorical assumptions. The varying density of connections across different regions reflects heterogeneous mobility constraints: workers can more easily transition between tightly connected occupations than between loosely connected ones.

This topology directly illustrates distance-dependent elasticity of substitution (DIDES). Dense connections within occupational groupings indicate high substitutability among similar occupations, while sparse connections across groupings reveal limited substitution across different occupational domains. This contrasts sharply with nested CES specifications, which impose rigid categorical boundaries. Our empirical tests reveal why nested CES fails: within-nest correlations are statistically insignificant and cross-nest elasticities (2.05--2.67) converge to the CES benchmark (3.12).\footnote{Standard Nested-CES forces arbitrary boundaries that may group dissimilar occupations together or separate similar ones. Our DIDES framework avoids this by allowing occupations to draw from multiple skills with varying intensities. See Appendix~\ref{b:appendix:nested_ces} for detailed comparisons.} The DIDES framework's flexibility—with $\theta = 1.10$ for cross-skill substitution but within-skill elasticities ranging from 2.1 to 4.8—captures heterogeneous mobility constraints that discrete nesting misses.

The right panel overlays automation exposure onto this substitution graph, revealing how technological clustering constrains adjustment. Automation concentrates in production and operative occupations (shown in red), creating a mobility trap: these occupations form a tightly connected group, meaning workers' most natural transition targets face similar automation threats. Dense connections that normally facilitate adjustment instead propagate technological shocks throughout the group. When automation affects manufacturing workers, they could easily transition to construction or transportation occupations—but those alternatives face comparable automation exposure. Low substitutability between differentially affected groups amplifies rather than dissipates wage effects.\footnote{Appendix~\ref{appendix:ai_topology} shows AI exhibits analogous clustering in professional and technical occupations, suggesting this pattern characterizes skill-biased technological change generally.} This visualization crystallizes our core theoretical insight: the interaction between technological clustering and substitution structure—not average labor market flexibility—determines distributional consequences.
\section{The Incidence of Automation and AI} \label{s:incidence}

Having estimated the structural parameters governing occupational substitution, we now quantify the labor market incidence of automation and AI. We adopt the labor demand elasticity estimate $\sigma = 1.34$ from \cite{Caunedo2023-db}. Our incidence analysis proceeds in three steps: we first apply spectral decomposition to formalize how technological exposures cluster in skill space, then document heterogeneous wage pass-through rates across occupations, and finally quantify how occupational mobility mitigates welfare losses. Throughout this analysis, we employ the parameter estimates from column 2 of Table \ref{tab:PPMLestimates}, which uses the 1980–2000 period.

\subsection{Spectral Decomposition of Technological Shocks}

We apply the spectral framework from Section~\ref{ss:spectral} to decompose measured automation and AI exposures, $\boldsymbol{z}^{\text{Automation}}$ and $\boldsymbol{z}^{\text{AI}}$, into eigenshocks—fundamental patterns revealing the labor market's adjustment capacity.\footnote{The occupational demand shock relates to automation exposure via $d\ln \boldsymbol{\alpha} = \sigma(\mathbf{I} + \Theta / \sigma) \cdot \beta \cdot \boldsymbol{z}^{\text{Automation}}$, where $d\ln \boldsymbol{w} = \beta \cdot \boldsymbol{z}^{\text{Automation}}$ is the wage effect of the automation shock. Therefore, when we apply eigendecomposition to the exposure vector, $\boldsymbol{z} = \sum_{n = 1}^O b_n \cdot \boldsymbol{u}_n$, the implied demand shock becomes $d\ln \boldsymbol{\alpha} = \sum_{n = 1}^O (1 + \lambda_n/\sigma) b_n \cdot \boldsymbol{u}_n$.} Each eigenshock's eigenvalue determines reallocation possibilities: smaller eigenvalues indicate rigid adjustment channels and larger wage effects.\footnote{We compute the substitution matrix by inverting employment shares to obtain implied productivity levels $A_{o,t}w_{o,t}^{\theta}$. Eigenshocks are eigenvectors of this matrix, with eigenvalues measuring absorption capacity.}

\begin{figure}[ht]
    \centering
    \subcaptionbox{Automation and AI Exposures}{\includegraphics[width=0.48\linewidth]{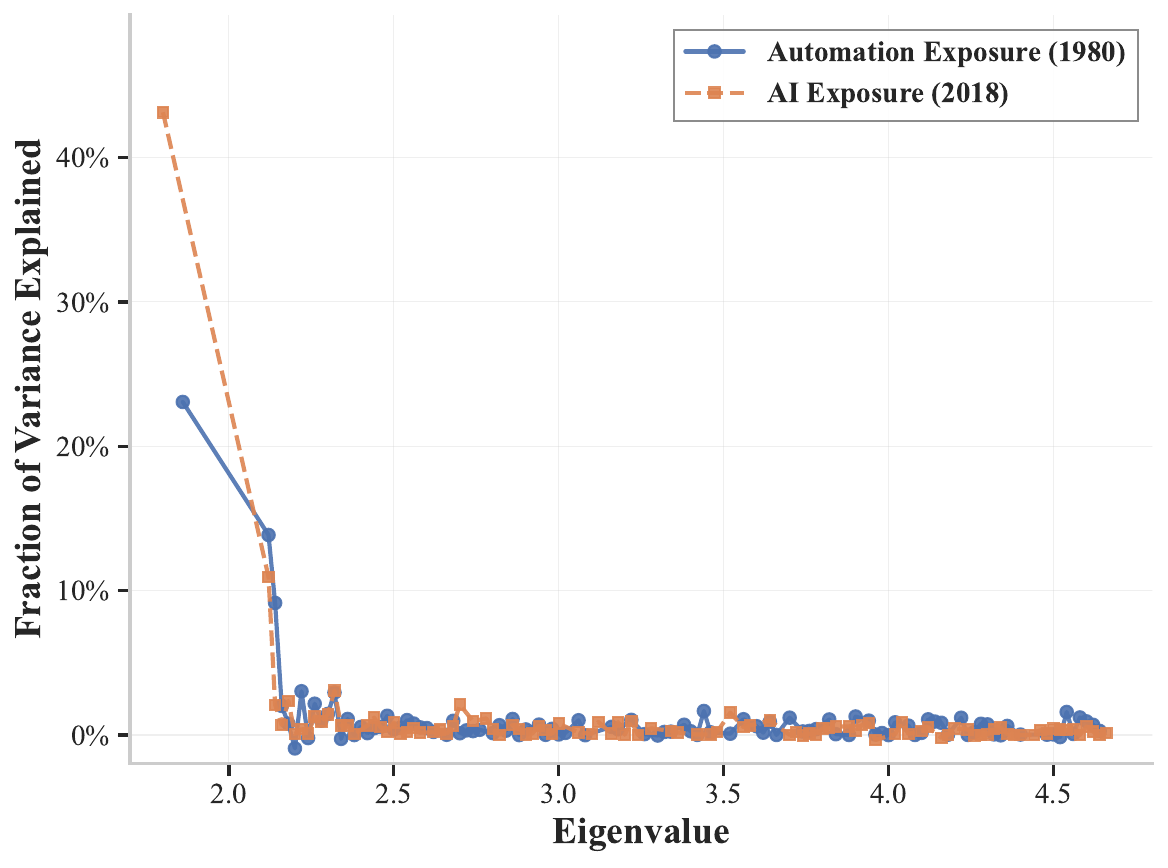}}\hfill
    \subcaptionbox{Trade and Demographic Shocks}{\includegraphics[width=0.48\linewidth]{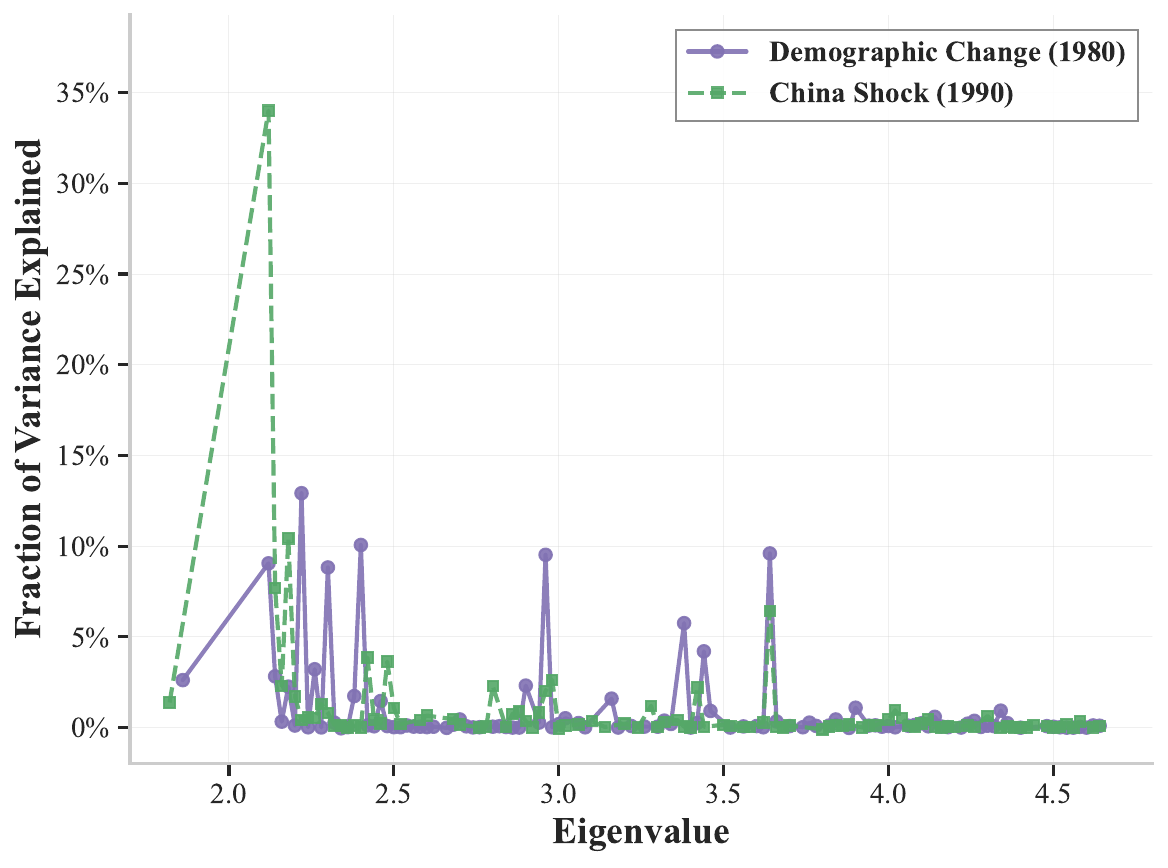}}
    \caption{Variance Decomposition of Labor Demand Shocks}
    \label{f:spectral_decomposition}
    \note{\textit{Notes:} Decomposition of labor demand shocks into eigenshocks ordered by eigenvalue magnitude using substitution matrices based on 1980 (automation, China Shock, Demographic Change) and 2018 (AI) employment shares. Each point represents the share of variance explained by the corresponding eigenshock. Smaller eigenvalues indicate limited employment reallocation capacity. Panel A: Automation and AI concentrate on low-eigenvalue eigenshocks (1.8--2.0). Panel B: Trade shocks from Chinese import competition (China shock) and demographic shocks from population aging distribute across higher eigenvalues. Demand shift measures obtained from \cite{Autor2024-ay}.}
\end{figure}

Figure~\ref{f:spectral_decomposition} reveals why technological shocks generate severe wage adjustments. The figure displays the share of variance contributed by each eigenshock, ordered by the magnitude of its associated eigenvalue. Panel A shows that both automation and AI load disproportionately onto eigenshocks with the smallest eigenvalues (1.8--2.0). Automation concentrates 23\% of its variance on the smallest eigenvalue; AI concentrates 44\%—dwarfing other components. These small eigenvalues represent shock patterns affecting clusters of similar occupations simultaneously, leaving workers with minimal escape options. When technological shocks concentrate on such rigid adjustment channels, the labor market cannot dissipate them through employment reallocation. Appendix~\ref{app:spatial_visualization} visualizes this clustering mechanism, showing how AI exposure aligns precisely with the low-eigenvalue eigenshock pattern in cognitive-manual skill space.

Clustering patterns vary across demographic groups, reflecting heterogeneous employment distributions across skill space. Appendix~\ref{app:gender_spectral} demonstrates that automation constrains male workers more severely—loading 31\% of variance on the smallest eigenvalue versus 10\% for women—due to men's concentration in manual-intensive occupations. This gender difference disappears for AI, where both groups face similar extreme concentration on low eigenvalues, suggesting cognitive task clustering affects all workers regardless of occupational segregation patterns.\footnote{The convergence in AI's gender impact contrasts with automation's differentiated effects, implying future technological shocks may generate more uniform demographic impacts while maintaining severe absolute clustering effects.} In Section~\ref{ss:distortions}, we discuss the role of automation and changing discrimination in between-group inequality.

Panel B contrasts this with two alternative occupational labor demand shocks: trade shocks from Chinese import competition and demand shocks from population aging.\footnote{The China shock represents occupation-level demand changes from increased Chinese import competition in manufacturing during 1991–2014, while demographic changes capture demand shifts from population aging (particularly the Baby Boom generation) over 1980–2018. We obtain these occupation-level demand changes from \cite{Autor2024-ay}, who construct them by combining industry-level shocks with occupations' employment distributions across industries. The China shock primarily affects manufacturing occupations, while demographic changes reflect consumption pattern shifts as the population ages. See Appendix~\ref{app:alternative_shocks} for details.} These distribute variance across eigenshocks with moderate-to-high eigenvalues. The China shock's largest loading (34\%) occurs at eigenvalue 2.1, while demographic changes load substantially on eigenvalues above 2.5. These patterns create flexible adjustment pathways—workers displaced from declining industries or occupations can transition to expanding ones, a mechanism that technological clustering eliminates.
\subsection{The Structure of Wage Incidence}

The spectral analysis revealed that automation and AI generate limited employment adjustment and large wage effects to first order; we now document their heterogeneous manifestation across occupations. Our central finding: technological incidence depends not on average substitutability (corresponding to CES estimates) but on the structure of substitution—specifically, how clustering interacts with skill-based mobility constraints. 

With the estimated wage changes $\hat{\boldsymbol{w}}_{1980\rightarrow2010}^{\text{Automation}} = \hat{\beta}_{2010} \cdot \boldsymbol{z}^{\text{Automation}}$ and the estimated labor supply parameters $\{\hat{\theta}, \hat{\boldsymbol{\rho}}\}$, the model implies employment changes $\hat{\boldsymbol{L}}_{1980\rightarrow2010}^{\text{Automation}}$. We construct the share of wage incidence for automation from equation~\eqref{eq:wage_incidence} as:
\begin{equation*}
    \text{Wage Pass-Through}^{\text{Automation}}_o = \frac{\ln \left( \hat{w}_{o,1980\rightarrow2010}^{\text{Automation}} / \hat{W}_{1980\rightarrow2010}^{\text{Automation}}\right) }{\ln \left( \hat{w}_{o,1980\rightarrow2010}^{\text{Automation}} / \hat{W}_{1980\rightarrow2010}^{\text{Automation}}\right) + \ln \hat{L}_{o,1980\rightarrow2010}^{\text{Automation}}/\sigma}
\end{equation*}
where $\hat{W}_{1980\rightarrow2010}^{\text{Automation}}$ is the change in the aggregate wage index under automation.\footnote{The aggregate wage index is $W = F(A_1w_1^{\theta}, \ldots, A_Ow_O^{\theta})^{1/\theta}$, where $F$ is the correlation function from the DIDES structure. This aggregates occupational wages weighted by their productivity parameters and correlation structure. Importantly, with observed employment shares and wage changes, we can compute aggregate wage index changes without requiring the productivity levels $\{A_o\}$ via hat algebra (see Appendix~\ref{app_ss:hat_algebra}).} For AI, we construct incidence measures similarly, assuming AI exposure $\boldsymbol{z}^{\text{AI}}$ generates the same proportional wage effects as automation (i.e., $\hat{\boldsymbol{w}}^{\text{AI}} = \hat{\beta}_{2010} \cdot \boldsymbol{z}^{\text{AI}}$), allowing us to predict prospective AI incidence. Under the CES framework, the share of wage adjustment is constant across all occupations and equal to $\frac{\sigma}{\sigma + \hat{\theta}} = 0.3$.
\begin{figure}[ht]
    \centering
    \subcaptionbox{Automation Pass-Through}{\includegraphics[scale=0.20]{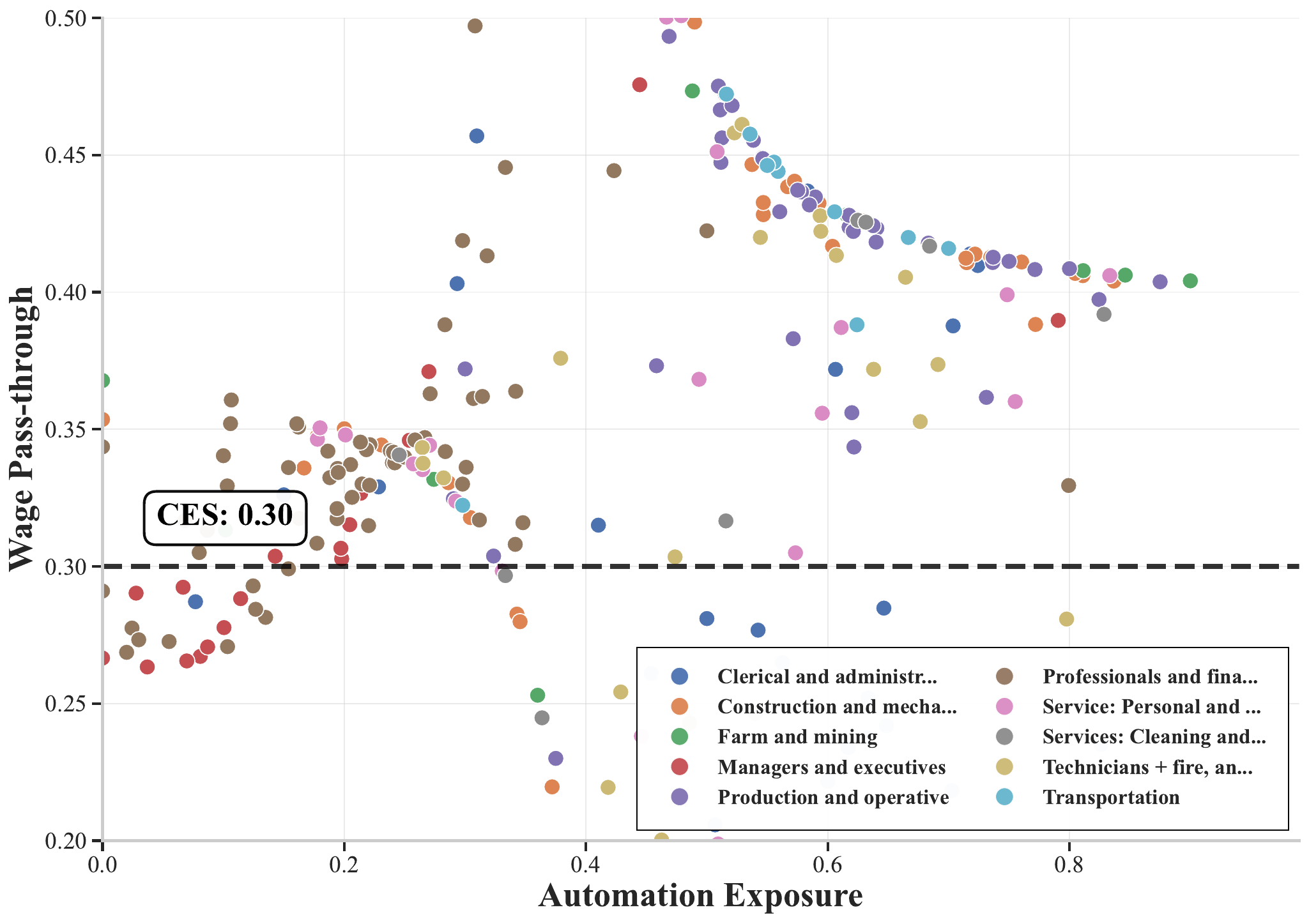}}\hfill
    \subcaptionbox{AI Pass-Through}{\includegraphics[scale=0.20]{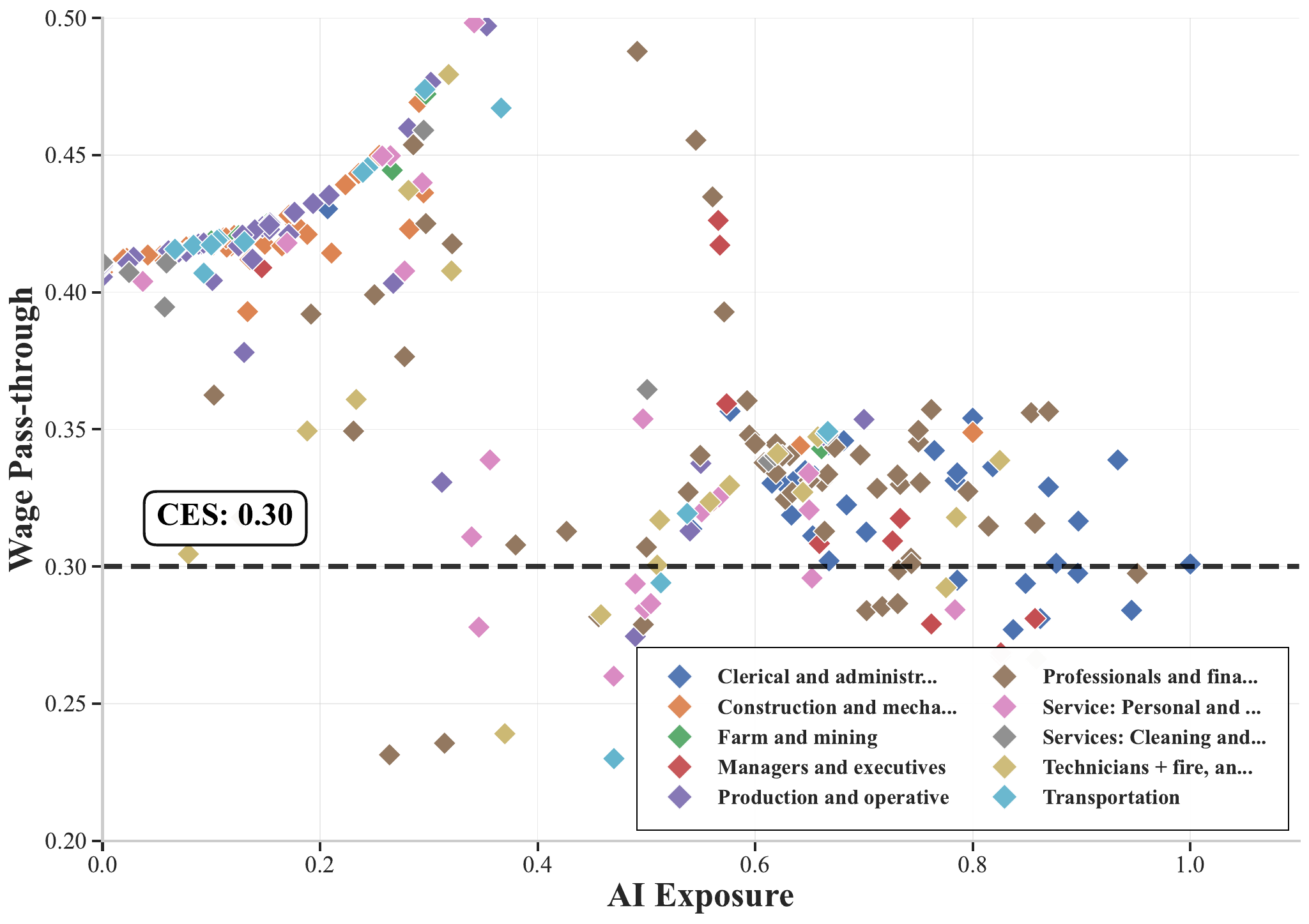}} \hfill
    \subcaptionbox{Distribution of Automation Pass-Through}{\includegraphics[scale=0.24]{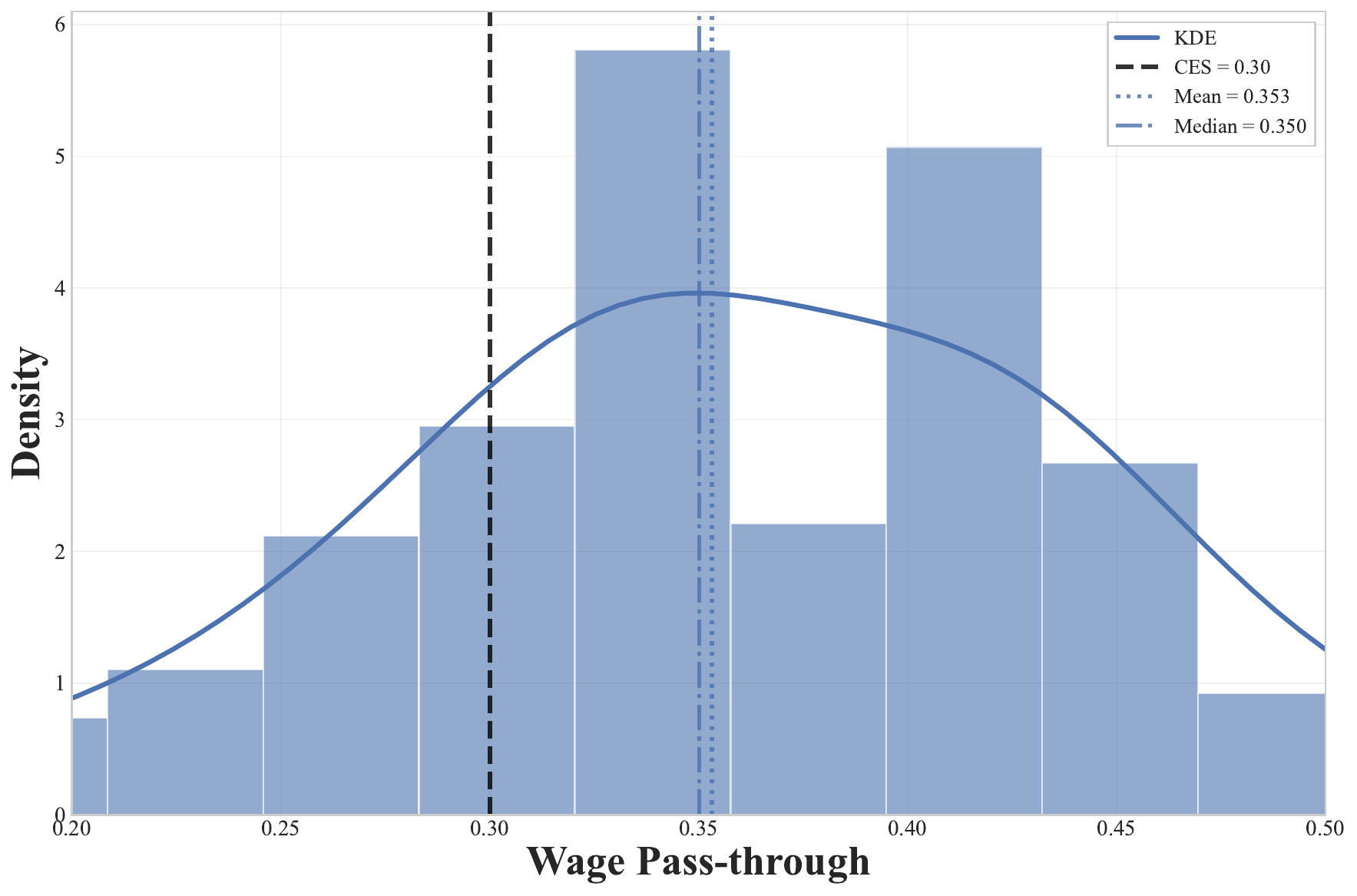}}\hfill
    \subcaptionbox{Distribution of AI Pass-Through}{\includegraphics[scale=0.24]{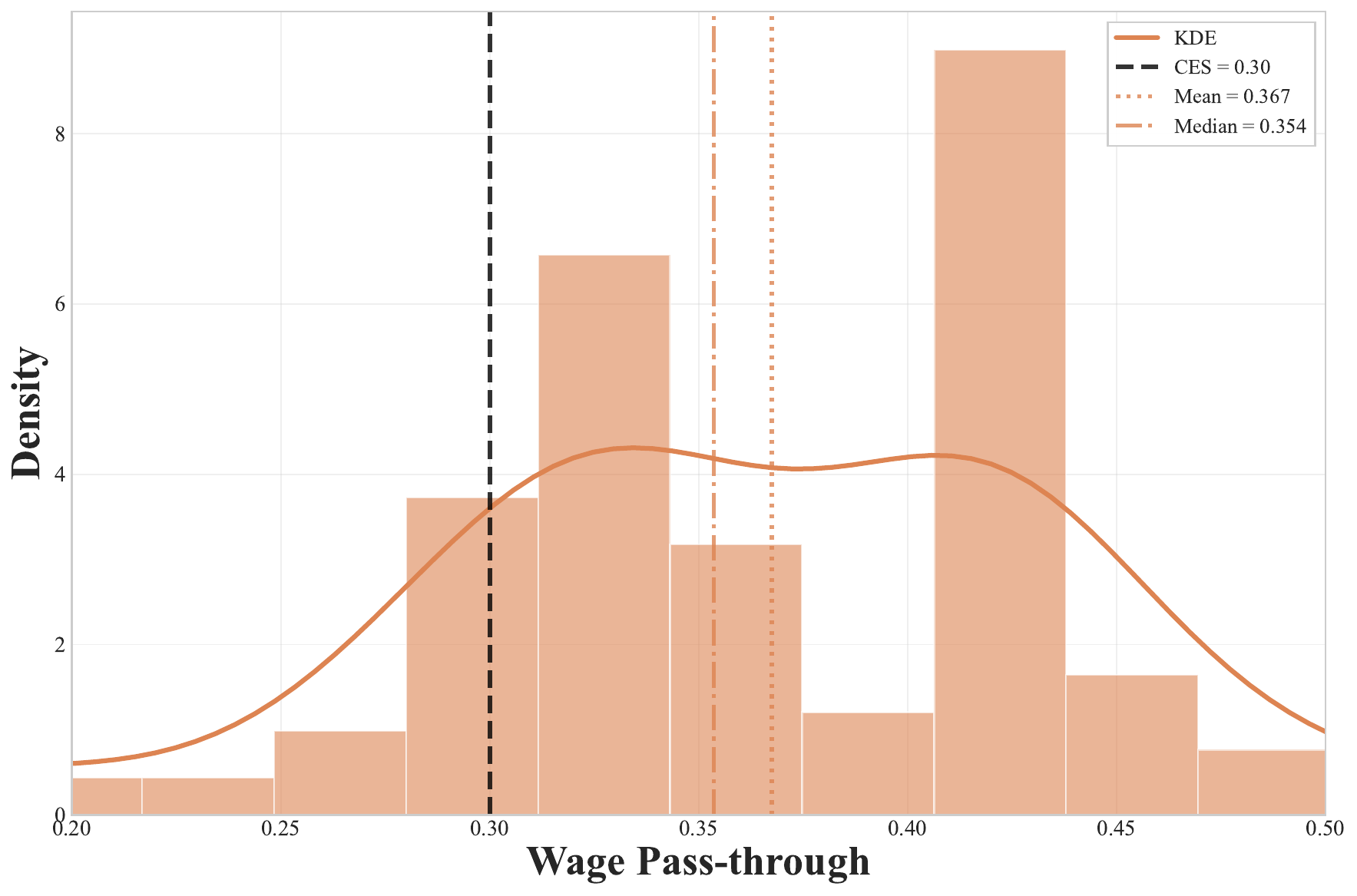}}
    \caption{Heterogeneous Wage Pass-Through of Technological Shocks}
    \label{f:incidence}
    \note{\textit{Notes:} Panels A-B: Pass-through rates versus exposure, with points colored by occupation category. Horizontal line marks CES benchmark (0.30). Panels C-D: Pass-through distributions with CES benchmark (dashed), mean (dotted), and median (dash-dotted) lines.}
\end{figure}

Figure~\ref{f:incidence} reveals systematic heterogeneity in how demand shocks split between employment and wage adjustments.\footnote{We invert the model using automation wage effects to recover demand changes. For AI, we normalize shocks to match automation's aggregate effect, enabling distributional comparison.} Panel A shows automation's pass-through ranging from 20\% to 50\%—far exceeding the 30\% CES benchmark for most exposed occupations. This variation directly reflects our theoretical prediction: when shocks cluster in skill space, affected workers cannot escape through reallocation, forcing adjustment through wages.

Production and transportation workers (purple/cyan points) exemplify this trap. Facing the largest negative shocks with 40–45\% pass-through, nearly half their demand destruction translates to wage losses—not the 30\% standard models predict. Automation clusters in manual-intensive occupations, eliminating natural transition pathways. The effective elasticity falls from 3.12 to 2.45, yielding an average pass-through of 0.353. This reveals that ignoring substitution structure and clustering overstates labor supply elasticity by 28\%, systematically understating wage inequality from technological change.

Panel B reveals AI's distinctive pattern compared to automation. While highly exposed occupations (right side) show modestly elevated pass-through around 35\%, larger pass-through rates emerge for AI-complementary occupations (left side), which exhibit pass-through rates of 40–45\%. This asymmetry reveals that AI generates larger wage gains for beneficiaries than wage losses for those displaced. Workers in AI-complementary occupations—those with minimal AI exposure—capture substantial wage increases because displaced cognitive workers cannot easily transition into these roles, which often require different skill combinations like manual or interpersonal expertise. This creates a more pronounced winner-take-all than automation: those who benefit from AI experience larger relative gains, amplifying inequality through a different mechanism than the symmetric displacement effects of automation.

Panels C and D translate occupation-specific effects into workforce distributions. Three features challenge conventional models. First, mean pass-through (0.353 for automation, 0.367 for AI) exceeds the CES benchmark by 20\%, implying that CES overstates average labor supply elasticity by 31\% and systematically understates wage effects. Second, the substantial range (0.2--0.5) reflects considerable heterogeneity in substitution structure. Third, these patterns have substantial consequences: for a 30\% demand shock, heavily exposed manual occupations experience 13.5\% wage declines (30\% × 0.45) versus 9\% under CES—a 50\% larger effect.

\subsection{Welfare Recovery Through Occupational Mobility}

While wage pass-through captures static losses, workers partially recover through occupational transitions. These mobility gains also depend crucially on the interaction between exposure patterns and worker mobility—revealing why average elasticities mislead.

\begin{figure}[ht]
    \centering
    \subcaptionbox{Automation: Mobility Gains}{\includegraphics[scale=0.2]{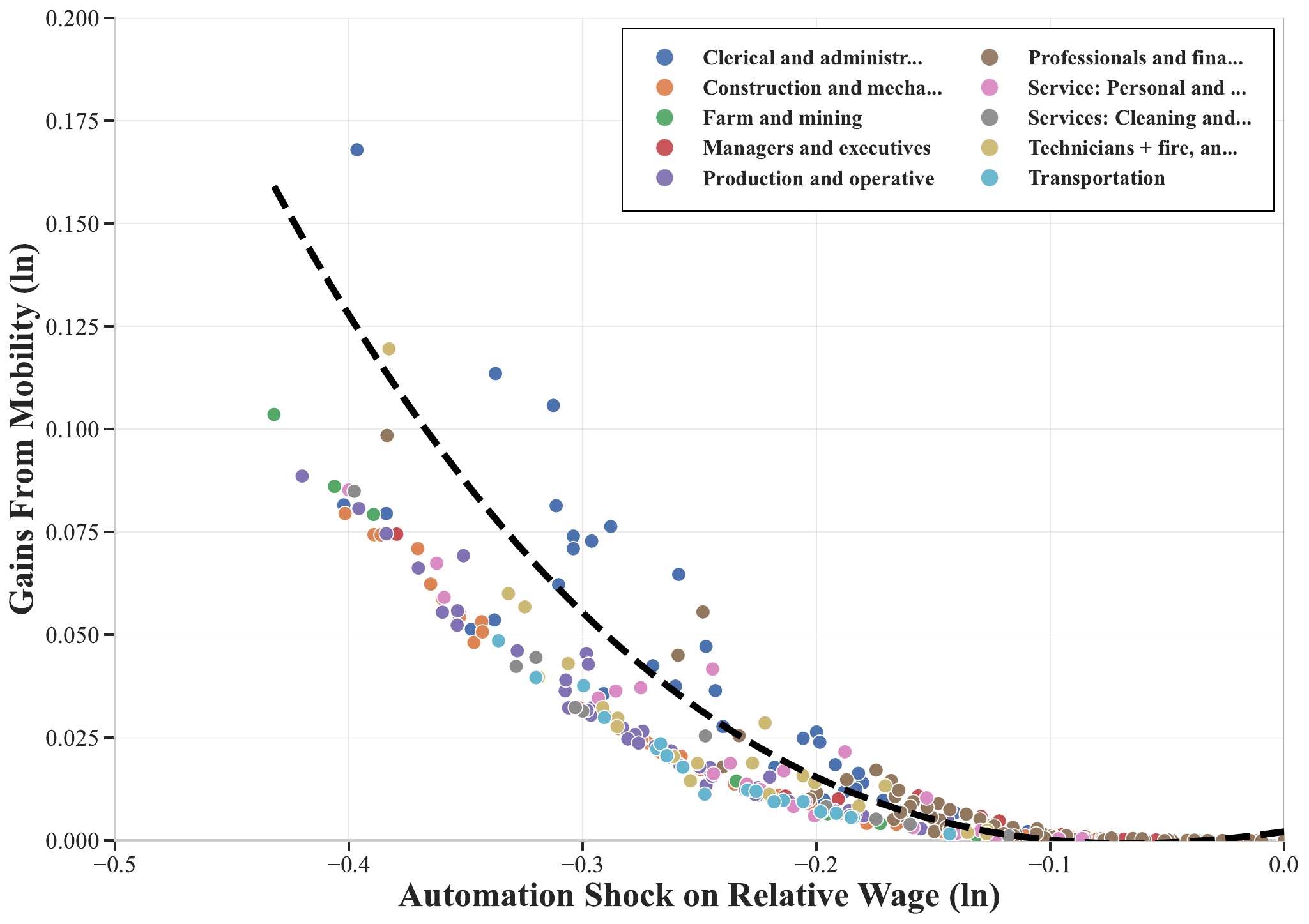}}\hfill
    \subcaptionbox{AI: Mobility Gains}{\includegraphics[scale=0.2]{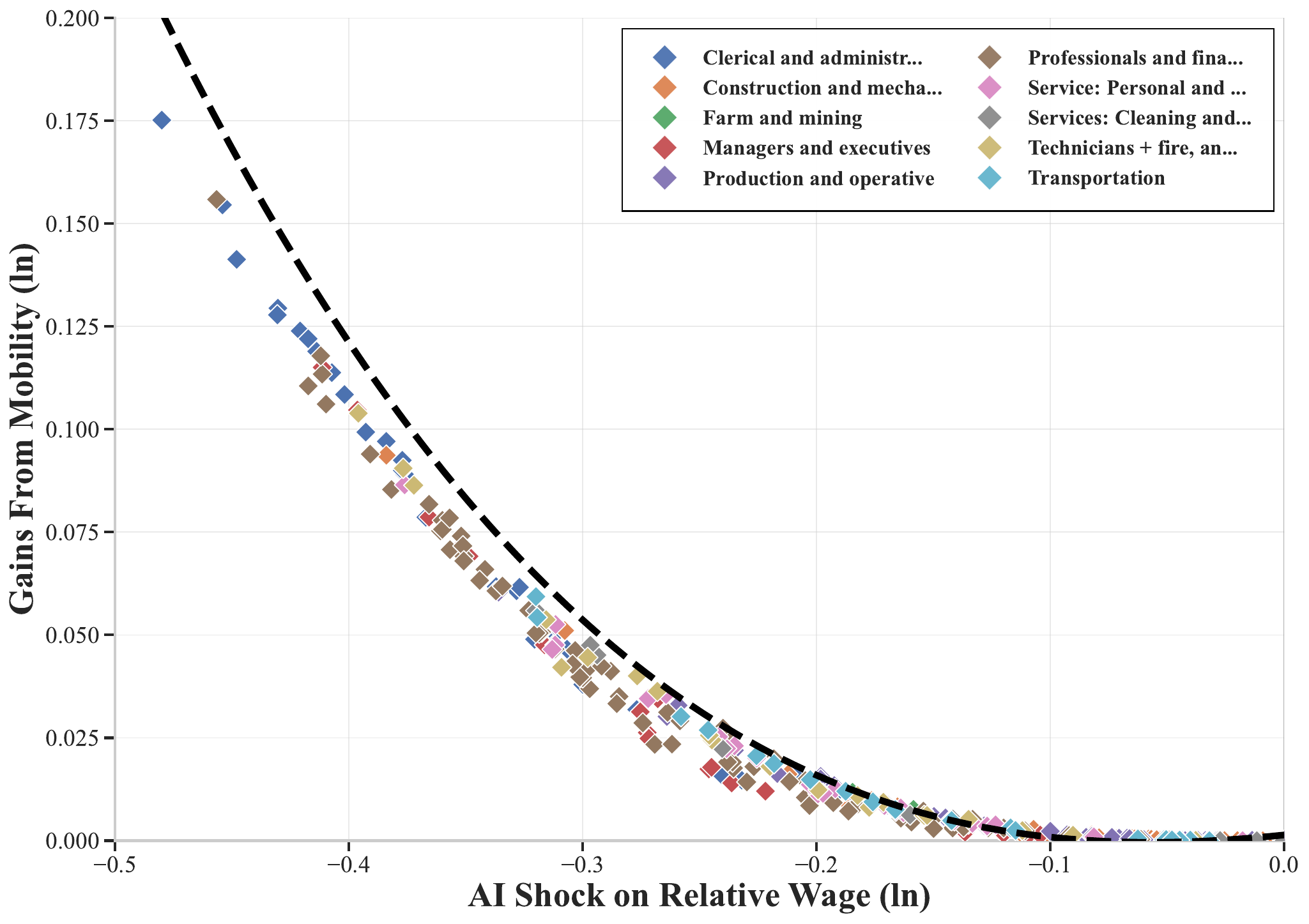}}
    \caption{Welfare Recovery Through Occupational Mobility}
    \label{f:gains_from_mobility}
    \note{\textit{Notes:} Welfare gains from transitions (equivalent variation) versus wage effects of technological shocks. Points represent occupations (DIDES model); black dashed curves show CES predictions.}
\end{figure}

Figure~\ref{f:gains_from_mobility} quantifies mobility gains as equivalent variation: $\mathrm{EV}_o = \sum_{o'} \mu_{oo'} \cdot \frac{\hat{w}_{o'}}{\hat{w}_{o}}$. The convex relationship confirms that larger shocks induce more transitions. Yet the systematic gap between our estimates and CES predictions reveals how clustering constrains recovery.

Panel A exposes automation's mobility trap: workers experiencing 40 log point wage declines recover only 20\% of losses through transitions, versus 30\% under CES—a 50\% overstatement. This gap emerges from a damaging interaction: automation concentrates in manual occupations where workers have the lowest skill transferability ($\rho_{\text{man}} = 0.48$). Production workers cannot escape to construction or transportation—their most productive alternatives—because these face similar threats. The clustering that drives wage losses simultaneously blocks escape routes. Standard models using average elasticity $\theta = 3.12$ miss this interaction, falsely implying that manual workers enjoy the same mobility as others.

Panel B reveals AI's contrasting pattern: cognitive workers experiencing 40 log point wage declines recover approximately 27\% of losses, still below CES predictions but notably higher than automation's impact. This improvement directly reflects our estimated $\rho_{\text{cog}} = 0.77$—cognitive skills transfer more readily across occupations. A threatened data analyst has more viable alternatives than a displaced welder, even when many cognitive occupations face AI exposure. Yet recovery remains limited: even with higher transferability, clustering ensures that workers' best alternatives are often similarly threatened.

\subsection{Summary: The Complete Incidence Picture}

Our three-pronged analysis reveals how technological clustering fundamentally reshapes labor market adjustment, generating more severe and persistent inequality than standard frameworks predict.

The spectral decomposition established that automation and AI concentrate on eigenshocks with the smallest eigenvalues (1.8--2.0), directing disruption through the labor market's most rigid adjustment channels. This concentration manifests in heterogeneous wage pass-through: our estimates reveal 20--50\% of demand shifts for both automation and AI translate to wages—substantially exceeding the 30\% CES benchmark. Heavily automation-exposed occupations face pass-through rates reaching 45\%, implying 50\% larger wage effects than standard models predict.

Mobility provides limited insurance against these losses. Workers heavily exposed to automation recover only 20\% of wage declines through occupational transitions, compared to 30\% under conventional assumptions. This constraint emerges from a damaging interaction: technological shocks cluster precisely where skill transferability is weakest. Automation targets manual occupations with $\rho_{\text{man}} = 0.48$, creating a double bind—large losses with minimal recovery options. AI focuses on cognitive occupations where $\rho_{\text{cog}} = 0.77$ offers better prospects, yet clustering still constrains escape routes.

The crucial insight connecting theory to empirics: incidence depends not only on average substitutability but on the interaction between shock distribution and substitution structure. In our framework, pass-through matrix $\Delta = (\mathbf{I} + \Theta/\sigma)^{-1}$ captures this interaction—when shocks align with low-eigenvalue eigenshocks (clustering), the labor market cannot dissipate shocks through reallocation, forcing adjustment through wages. This structural mechanism, absent from models assuming uniform elasticity, explains why technological change generates such pronounced and persistent distributional consequences.
\section{Dynamic Extension with Transition} \label{s:dynamic_model}
Our model retains the Roy structure, allowing it to be embedded within related frameworks while incorporating richer substitution patterns through DIDES. In this section, we extend the static model into a dynamic discrete choice framework \citep{Artuc2010-xh, Caliendo2019-oi} to capture gradual labor market transitions. This extension allows us to examine the dynamic labor market incidence of technological adoption in both the transition and the long run \citep{lehr2022optimal, Adao2024-ub}.

\subsection{Dynamic Discrete Choice with DIDES} \label{s:dynamic_dides}

Our focus is on studying the dynamic labor market incidence of technological adoption rather than modeling firms' endogenous technology adoption decisions. The production side remains identical to the static framework, while workers make rational, forward-looking occupational choices in response to automation and AI shocks. To model these choices, we adopt a structure similar to \cite{Caliendo2019-oi} with a correlated productivity distribution among jobs.

\paragraph{Workers' Dynamic Decision}

In each period, we denote the vector of occupational employment by $\boldsymbol{L}_t $. Workers are assumed to be hand-to-mouth, taking the wage path $\left\{\boldsymbol{w}_t\right\}_{t=0}^{\infty}$ as given, and derive utility from consumption and labor supply according to:
\begin{equation*}
    U\left(\left\{c_t(i), \ell_t(i)\right\}_{t=0}^{\infty}\right) = \sum_{t=0}^{\infty} \beta^t\left(\ln c_t(i)-\ln \ell_t(i)\right)
\end{equation*}

At the beginning of each period, workers draw labor productivity across all occupations from the same distribution as in the static model:\footnote{Unlike \cite{dvorkin2019occupation} and \cite{seo2023sectoral}, we abstract from persistent worker heterogeneity while allowing for a flexible substitution structure to focus on incidence. Incorporating persistent worker heterogeneity would allow for richer welfare analysis but would substantially increase data requirements, necessitating long panels with large numbers of individuals.}
\begin{equation*}
    \operatorname{Pr}\left[\epsilon_1(i) \leq \epsilon_1, \ldots, \epsilon_O(i) \leq \epsilon_O\right]=\exp \left[-F\left(A_1 \epsilon_1^{-\theta}, \ldots, A_O \epsilon_O^{-\theta}\right)\right]
\end{equation*}

After observing their labor productivity, workers choose an occupation, with consumption equal to occupational income, $c_t = w_{o,t}$, and labor supply given by:
\begin{equation*}
    \ln \left( \ell_t \left( i\right)\right) = - \kappa \ln \left(\epsilon_{o,t} \left( i\right) \right)
\end{equation*}

In contrast to the static model, productivity enters the labor supply function with a short-run disutility factor $\kappa$, which governs short-run labor supply elasticity. While we could allow workers to redraw labor productivity with some probability, doing so yields the same sufficient statistics for counterfactual welfare and similar dynamics.\footnote{Since we estimate the short-run elasticity from observed transition data, the key counterfactual objects depend only on these estimated elasticities, not on whether productivity redraws occur through a probability mechanism or our baseline specification. See \cite{dvorkin2019occupation} for a similar insight regarding sufficient statistics in dynamic discrete choice models.} Additionally, workers incur job transition costs $\tau_{oo^\prime}$ when switching occupations.


\begin{assumption}
    Job transition cost is constant over time $\tau_{oo^\prime}$ and measured in terms of utility.
\end{assumption}

Given this economic environment, we formulate workers' decisions recursively via the following Hamilton-Jacobi-Bellman (HJB) equation:
\begin{equation*}
    v_{o,t}\left(\boldsymbol{\epsilon}_{t}\right)=\max_{o^{\prime}}\left\{ \ln w_{o^{\prime},t}+\kappa\ln\epsilon_{o^{\prime},t}+\beta V_{o,t+1}-\tau_{oo^{\prime}}\right\}
\end{equation*}
where $V_{o,t+1}=\mathrm{E}_{\boldsymbol{\epsilon}}\left[v_{o,t+1}\left(\boldsymbol{\epsilon}\right)\right]$ and $v_{o,t}\left(\boldsymbol{\epsilon}_{t}\right)$ denote a worker's lifetime utility in occupation $o$ after observing their productivity. This utility comprises current-period benefits $\ln w_{o^{\prime},t}+\kappa\ln\epsilon_{o^{\prime},t}$ and discounted expected future utility $V_{o,t+1}$, net of job transition cost $\tau_{oo^{\prime}}$. Workers choose occupation $o^\prime$ to maximize lifetime utility.

Consequently, we can recursively express occupational expected utility as:
\begin{align*}
    V_{o,t} &=\ln\left(F\left(A_{1,t}Z_{o1,t}{}^{\frac{\theta}{\kappa}},\ldots,A_{O,t}Z_{oO,t}^{\frac{\theta}{\kappa}}\right)^{\frac{\kappa}{\theta}}\right)+\bar{\gamma}\frac{\kappa}{\theta} \\
    \text{where } Z_{oo^{\prime},t} &= \exp\left(\beta V_{o^{\prime},t+1}+\ln w_{o^{\prime},t}-\tau_{oo^{\prime}}\right)
\end{align*}

Additionally, job transition probability can be derived as shown in Appendix \ref{a:appendix:dynamic_worker}:
\begin{equation*}
    \mu_{oo^{\prime},t}=\frac{A_{o^{\prime},t}Z_{oo^{\prime},t}{}^{\frac{\theta}{\kappa}}\times F_{o^{\prime}}\left(A_{1,t}Z_{o1,t}{}^{\frac{\theta}{\kappa}},\ldots,A_{O,t}Z_{oO,t}^{\frac{\theta}{\kappa}}\right)}{F\left(A_{1,t}Z_{o1,t}{}^{\frac{\theta}{\kappa}},\ldots,A_{O,t}Z_{oO,t}^{\frac{\theta}{\kappa}}\right)}
\end{equation*}

The interplay of job transition costs and idiosyncratic productivity shocks generates slow labor market adjustments in our model. A key distinction of our approach is that it allows for rich substitution patterns between jobs, as embedded in the correlation function $F$. When $F$ is additive, our framework reduces to the standard model.

\paragraph{Dynamic Equilibrium}
The production side follows static firm optimization as discussed in the static model. The share of tasks performed by labor, denoted by $\left\{ \boldsymbol{\alpha}_t\right\}_{t=0}^{\infty}$, characterizes distributional effects of technological adoption, while aggregate productivity, $\left\{\mathcal{A}_t\right\}_{t=0}^{\infty}$, is Hicks-neutral (see Appendix \ref{app_ss:production}). All other occupational labor productivity is represented by $\left\{\boldsymbol{A}_t\right\}_{t=0}^{\infty}$.

Given time-varying fundamentals $\left\{\Psi_t \right\}_{t = 0}^{\infty}=\left\{ \boldsymbol{\alpha}_t, \mathcal{A}_t, \boldsymbol{A}_t\right\}_{t=0}^{\infty}$, we define a dynamic equilibrium under rational expectations. In this equilibrium, there exists a time path of wages $ \left\{ \boldsymbol{w}_{t}\right\}_{t=0}^{\infty}$, occupational allocations $\left\{ \boldsymbol{L}_{t}\right\}_{t=0}^{\infty}$, and job transition probabilities $ \left\{ \mu_{t}\right\}_{t=0}^{\infty}$ such that:
\begin{enumerate}
    \item Wage vector $\boldsymbol{w}_t = \boldsymbol{w}\left(\boldsymbol{\alpha}_t, \mathcal{A}_t, \boldsymbol{L}_t\right)$ solves the static production equilibrium.
    \item Workers' optimal occupational choices yield job transitions $\mu_{t}=\left\{ \mu_{oo^{\prime},t}\right\} _{o=1,o^{\prime}=1}^{O,O}$.
    \item Labor allocation evolves according to $L_{o, t}=\sum_{o^{\prime}} \mu_{o o^{\prime}, t} L_{o, t-1}$.\footnote{We construct job flows from retrospective responses in the March CPS; consequently, aggregate flows derived from these responses do not directly match observed occupational employment levels. To account for this discrepancy, we adjust the evolution of occupational employment as $L_{o, t}=\sum_{o^{\prime}} \mu_{o o^{\prime}, t} L_{o, t-1} + \Delta L_{o,t}$.}
\end{enumerate}
\subsection{Dynamic Hat Algebra with Correlation}

In this section, we extend the dynamic hat algebra solution method to incorporate correlated productivity distributions, enabling richer substitution patterns. The model addresses key counterfactual questions: What would have happened to the wage distribution if automation technologies had not been adopted? How much can the labor market absorb unequal demand shocks caused by AI if AI technologies are adopted to the same extent as automation, but at a much more rapid pace by 2030?

Formally, our counterfactual analysis studies how equilibrium allocations across occupations and over time change relative to a baseline economy when faced with an alternative sequence of fundamentals, denoted by $\left\{\Psi^{\prime}_t\right\}_{t=1}^{\infty}$. We examine how changes in these counterfactual fundamentals affect equilibrium outcomes of interest.

To facilitate characterization of the dynamic equilibrium, we introduce additional notation. For any scalar or vector $x$, we denote its proportional change between periods $t$ and $t + 1$ as $\dot{x}_{t + 1} = x_{t+1}/x_t$. Additionally, $x_t^\prime$ denotes the corresponding variable in the counterfactual economy. Finally, we define $\hat{x}_{t+1} = \dot{x}_{t+1}^{\prime}/\dot{x}_{t+1}$, which represents the ratio of the time change in the counterfactual equilibrium to that in the initial equilibrium.

Before characterizing counterfactual outcomes, we introduce the correlation-adjusted transition probability:
\begin{equation*}
    \tilde{\mu}_{oo^{\prime},t}=A_{o^{\prime},t}Z_{oo^{\prime},t}{}^{\frac{\theta}{\kappa}}/F\left(A_{1,t}Z_{o1,t}{}^{\frac{\theta}{\kappa}},\ldots,A_{O,t}Z_{oO,t}^{\frac{\theta}{\kappa}}\right)
\end{equation*}
which serves as a sufficient statistic. Note that when $F$ is additive, $\tilde{\mu}_t$ coincides with $\mu_t$ since $\tilde{\mu}_{oo^{\prime},t} = \mu_{oo^{\prime},t} / F_{o^\prime}$.

\begin{lemma}
    Given correlation function $F$, there exists a unique mapping between occupation transition probability $\mu_t$ and correlation-adjusted transition probability $\tilde{\mu}_t$:
    \begin{equation*}
        \Bigl\{ \mu_{oo^{\prime},t} = \tilde{\mu}_{oo^{\prime},t} \, F_{o^{\prime}}\Bigl(\tilde{\mu}_{o1,t},\ldots,\tilde{\mu}_{oO,t}\Bigr) \Bigr\}_{o^{\prime}=1}^{O},\quad \forall\, o
    \end{equation*}
\end{lemma}

With this correlation-adjusted transition probability, we introduce the dynamic hat algebra with correlation to analyze how economic outcomes change counterfactually. Specifically, we study how allocations and wages across occupations evolve over time in response to alternative sequences of fundamentals, denoted by $\left\{\hat{\Psi}_t\right\}_{t=1}^{\infty}$.

\begin{proposition}[Dynamic Hat Algebra with Correlation] \label{prop:6}
    For time-varying counterfactual changes in fundamentals $\left\{\hat{\Psi}_t\right\}_{t=1}^{\infty} = \left\{\hat{\boldsymbol{\alpha}}_t, \hat{\mathcal{A}}_t, \hat{\boldsymbol{A}}_t\right\}_{t=1}^{\infty}$ with $\lim _{t \rightarrow \infty} \hat{\Psi}_t=1$, observed allocations and transition probability $\left\{\boldsymbol{L}_t, \mu_{t}\right\}_{t=0}^{\infty}$ are sufficient to characterize counterfactual changes in allocations, wages, and expected utility ($u_{o,t} = \exp\left(V_{o,t} \right)$). Formally:
    \begin{itemize}
        \item Counterfactual changes in wages $\hat{\boldsymbol{w}}_t = \hat{\boldsymbol{w}}\left( \hat{\boldsymbol{\alpha}}_t, \hat{\mathcal{A}}_t, \hat{\boldsymbol{L}}_t \right)$ solve the static production equilibrium.
        \item Counterfactual correlation-adjusted transition probability is:
        \begin{equation*}
            \tilde{\mu}_{oo^{\prime},t}^{\prime} = \frac{\tilde{\mu}_{oo^{\prime},t-1}^{\prime}\dot{\tilde{\mu}}_{oo^{\prime},t}\hat{A}_{o^{\prime},t}\hat{u}_{o^{\prime},t+1}^{\beta\frac{\theta}{\kappa}}\hat{w}_{o^{\prime},t}^{\frac{\theta}{\kappa}}}{F\left(\left\{ \tilde{\mu}_{oo^{\prime\prime},t-1}^{\prime}\dot{\tilde{\mu}}_{oo^{\prime\prime},t}\hat{A}_{o^{\prime\prime},t}\hat{u}_{o^{\prime\prime},t+1}^{\beta\frac{\theta}{\kappa}}\hat{w}_{o^{\prime\prime},t}^{\frac{\theta}{\kappa}}\right\} _{o^{\prime\prime}=1}^{O}\right)}
        \end{equation*}
        \item Counterfactual change in utility is characterized by:
        \begin{equation*}
            \hat{u}_{o,t+1} = F\left(\left\{ \tilde{\mu}_{oo^{\prime\prime},t}^{\prime}\dot{\tilde{\mu}}_{oo^{\prime\prime},t+1}\hat{A}_{o^{\prime\prime},t+1}\hat{u}_{o^{\prime\prime},t+2}^{\beta\frac{\theta}{\kappa}}\hat{w}_{o^{\prime\prime},t+1}^{\frac{\theta}{\kappa}}\right\} _{o^{\prime\prime}=1}^{O}\right)^{\frac{\kappa}{\theta}}
        \end{equation*}
        with terminal condition $\lim_{t\rightarrow \infty}\hat{u}_{o,t} = 1$.
        \item Counterfactual occupational allocation evolves according to $L_{o^{\prime},t}^{\prime}=\sum_{o}\mu^{\prime}_{oo^{\prime},t}L_{o,t-1}^{\prime}$.
    \end{itemize}
\end{proposition}

\begin{proof}
    See Appendix \ref{a:appendix:dynamic_hat_algebra} for proof and for the different expression for time $0$ that accounts for unexpected changes in fundamentals.
\end{proof}

Proposition~\ref{prop:6} demonstrates the sufficient statistic property of the dynamic hat algebra: observed allocations and transition probabilities fully characterize counterfactual outcomes under a new sequence of fundamentals. Moreover, it underscores the critical role of the substitution structure—captured by function $F$—which governs counterfactual outcomes. While observed allocations serve as sufficient statistics, the substitution structure characterized by $F$ determines how changes in fundamentals translate into counterfactual wages and allocations.\footnote{When $F$ is additive, we return to the standard dynamic hat algebra approach with independent productivity distribution.} Our static results, showing that clustering of technological changes combined with DIDES leads to unequal labor market incidence, persist in the dynamic framework.

Finally, as derived in Appendix \ref{a:appendix:dynamic_welfare_metric}, welfare change resulting from a shift in fundamentals—measured in terms of consumption equivalent variation—can be expressed as:
\begin{equation*}
    \mathrm{EV}_{o,t} = \left(1 - \beta \right)\sum_{s=t}^{\infty}\beta^{s-t}\ln\left(\hat{w}_{o,s}/\hat{\tilde{\mu}}_{oo,s}^{\frac{\kappa}{\theta}}\right)
\end{equation*}
Moreover, changes in occupation-specific correlation-adjusted staying probabilities capture gains from mobility, echoing results in \cite{Arkolakis2012-tk}, once the substitution structure is taken into account.

\subsection{Data and Estimation}

\paragraph{The Euler-Equation Approach}

Following the Euler-equation approach of \cite{Artuc2010-xh}, we estimate the short-run elasticity while accounting for correlation in the productivity distribution. Specifically, we derive the following estimating equation:
\begin{equation*}
    \ln\frac{\tilde{\mu}_{oo^{\prime},t}}{\tilde{\mu}_{oo,t}} = \frac{\theta}{\kappa}\ln\frac{w_{o^{\prime},t}}{w_{o,t}}+\beta\ln\frac{\tilde{\mu}_{oo^{\prime},t+1}}{\tilde{\mu}_{o^{\prime}o^{\prime},t+1}}+\left(\beta-1\right)\tau_{oo^{\prime}} + \nu_t
\end{equation*}
where $\nu_t$ is an error term. This expression parallels the formulation in \cite{Artuc2010-xh} (hereafter ACM) but uses correlation-adjusted transition probabilities. Intuitively, correlation-adjusted transitions incorporate information on expected future wages and the option value of job mobility, with adjusted future transitions serving as sufficient statistics for these option values (see Appendix \ref{a:appendix:dynamic_welfare_metric} for details). The key insight is that, after conditioning on these adjusted future values, the coefficient $\frac{\theta}{\kappa}$ represents the elasticity of relative correlation-adjusted transitions with respect to relative wage changes.\footnote{We cannot separately identify $\theta$ and $\kappa$, nor is this necessary, as they enter equilibrium dynamics and welfare metrics jointly, as demonstrated in Proposition \ref{prop:6}.} As in ACM, the theory implies that lagged values of wages and adjusted transitions are valid instruments.\footnote{The exclusion condition requires that the error term $\nu_t$ is serially uncorrelated. See ACM for detailed discussion.}

\paragraph{Data and Estimation Results}
Our estimation strategy requires aggregate job flows across occupations and average wages—data readily available from standard sources. We construct these measures using individual-level data from the US Census Bureau's March Current Population Survey (CPS). While our approach requires only aggregate transitions and wages, the limited CPS sample size necessitates grouping occupations. We therefore cluster occupations into 15 groups based on their skill intensities using a k-means algorithm. This grouping is natural, as occupations with similar skill intensities cluster together in skill space. We then compute annualized job transition probabilities among these 15 clusters, $\mu_t$, for the period 1976–2019. Appendix \ref{b:appendix:CPS} provides detailed discussion of data construction.

\begin{table}[ht]
\centering
\def\sym#1{\ifmmode^{#1}\else\(^{#1}\)\fi}
\caption{Estimation of Short-Run Elasticity $\theta/\kappa$}
\label{tab:short_run_elas}
\resizebox{0.5\linewidth}{!}{
\begin{tabular}{l*{4}{c}}
\toprule
            &\multicolumn{1}{c}{(1)} & \multicolumn{1}{c}{(2)} & \multicolumn{1}{c}{(3)} & \multicolumn{1}{c}{(4)} \\
            & OLS & IV & IV + Dest. FE & IV + Origin FE \\
\midrule
$\theta/\kappa$ & 0.063  & 0.071 & 0.068  & 0.080 \\
                & (0.018)        & (0.018)        & (0.025)       & (0.025)       \\
\midrule
Observations    & 630 & 630 & 630 & 630 \\
Destination FE  & No  & No  & Yes & No  \\
Origin FE       & No  & No  & No  & Yes \\
IV              & No  & Yes & Yes & Yes \\
\bottomrule
\end{tabular}
}
\note{\textit{Notes:} This table reports estimates of the short-run elasticity $\theta/\kappa$ from the Euler equation specification. Column (1) presents the baseline OLS estimate. Column (2) employs IV estimation using lagged adjusted job transition probabilities and wages as instruments. Columns (3) and (4) add destination and origin fixed effects, respectively, both with IV. Standard errors in parentheses.}
\end{table}

We use $\beta = 0.96$ as the annual discount factor. Table \ref{tab:short_run_elas} reports estimation results for the short-run elasticity $\frac{\theta}{\kappa}$. Column (1) presents the OLS estimate, yielding a short-run elasticity of 0.063. Column (2) implements an IV approach using lagged adjusted transition probabilities and wages as instruments, resulting in an estimate of 0.071. Columns (3) and (4) incorporate destination and origin fixed effects, respectively, yielding estimates of 0.068 and 0.080. While these estimates are broadly consistent, they are substantially lower than those reported in ACM, primarily due to our use of correlation-adjusted transition probabilities. As discussed in the static model, this adjustment nets out within-skill substitutability—a major source of variation in transition responses to relative wage changes. Moreover, grouping occupations by similar skill intensities further reduces across-cluster transition responses, contributing to smaller elasticity estimates.

\subsection{The Dynamic Incidence of Automation and AI}

We now assess the effects of automation and AI within a slow-adjustment labor market framework. In our quantitative evaluation, we employ 15 occupation clusters with transition probabilities constructed from CPS data. For counterfactual applications, we use elasticities from our static estimation, augmented by the short-run labor supply elasticity $\theta/\kappa = 0.07$ estimated via the Euler-equation approach. To maintain clarity, we focus on the time path of average effects, as cross-sectional heterogeneity closely mirrors results from the static model.

For automation technologies, we obtain ex-post estimates of their dynamic wage effects, as shown in Panel A of Figure \ref{f:auto_wage_emp}. Occupations with higher automation exposure have experienced gradual relative wage decline since 1985, resulting in up to 50\% difference between occupations where all tasks are exposed and those where none are. To match this observed wage trend, we calibrate the occupation-specific demand shocks $\left\{\hat{\boldsymbol{\alpha}}_t^{\text{Automation}}\right\}$. We then implement the following counterfactual: what would have occurred if these automation shocks had not materialized since 1985?\footnote{Since we focus on the distributional effects of automation exposure, we omit discussion of aggregate gains.} Figure \ref{f:dynamic_incidence} displays the resulting time path of average effects across occupations.

\begin{figure}[ht]
    \centering
    \subcaptionbox{Employment Effects of Automation}{\includegraphics[scale=0.25]{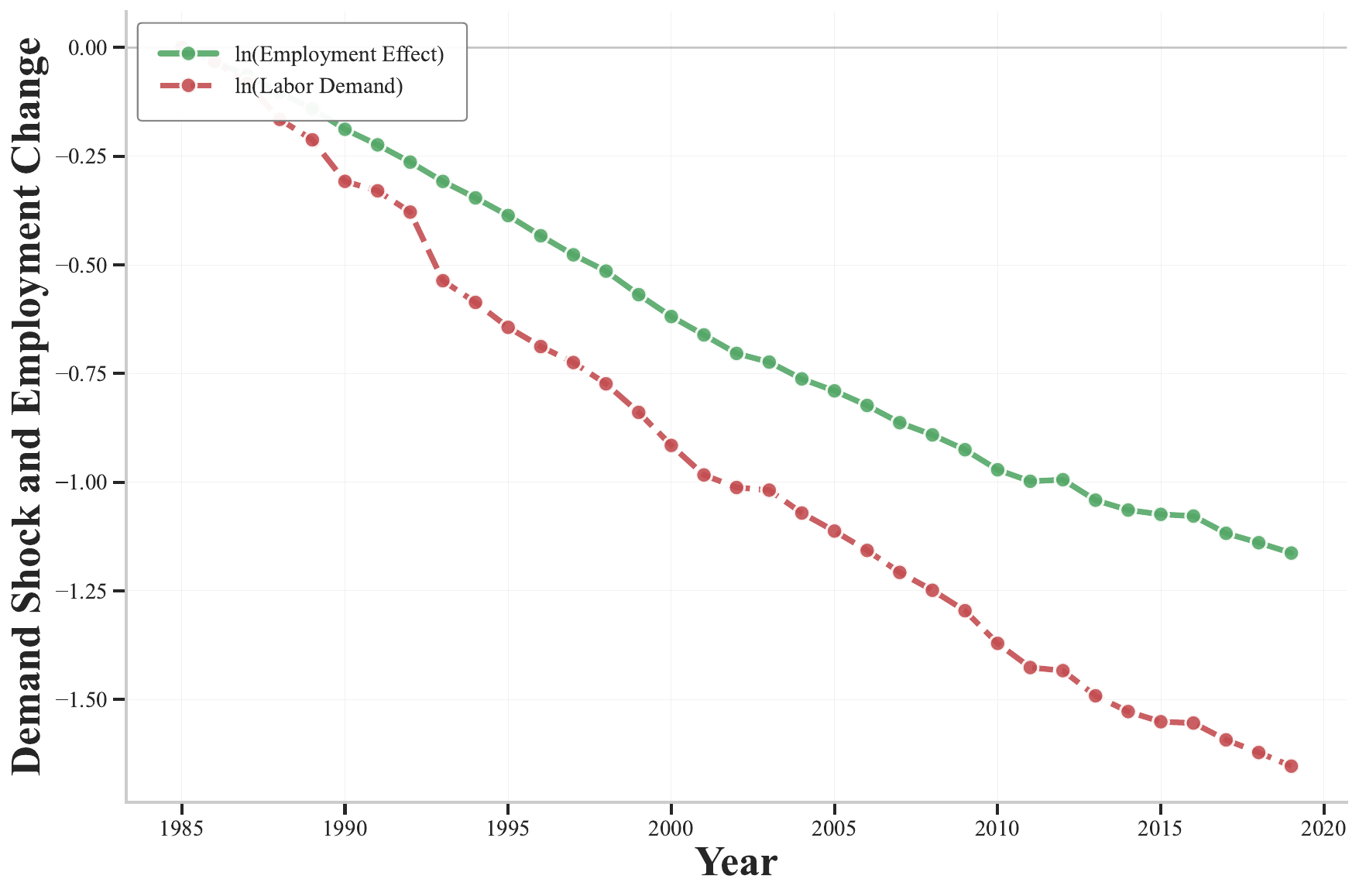}}\hfill
    \subcaptionbox{Wage Incidence of Automation}{\includegraphics[scale=0.25]{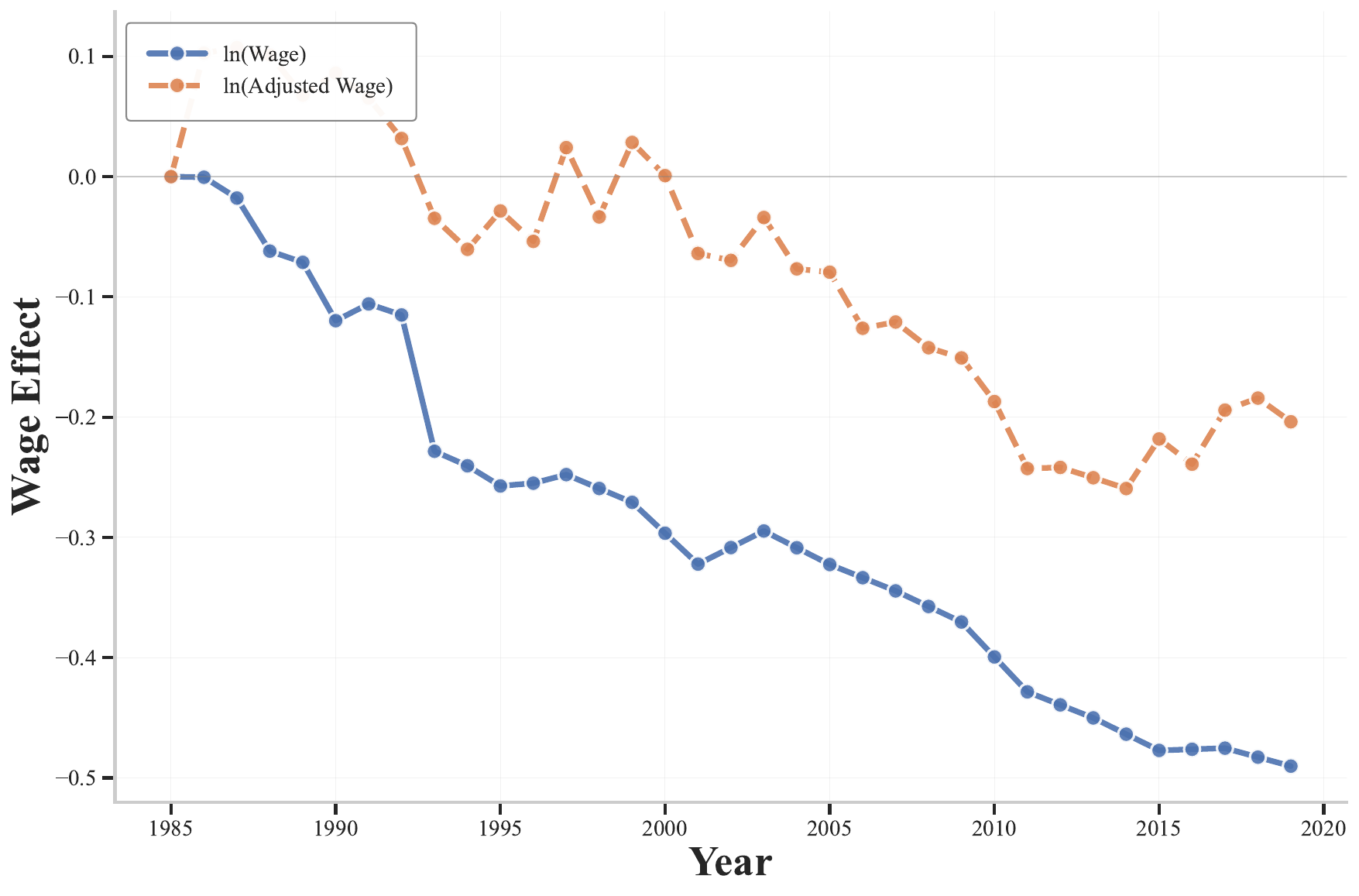}} 
    \subcaptionbox{Employment Effects of AI}{\includegraphics[scale=0.25]{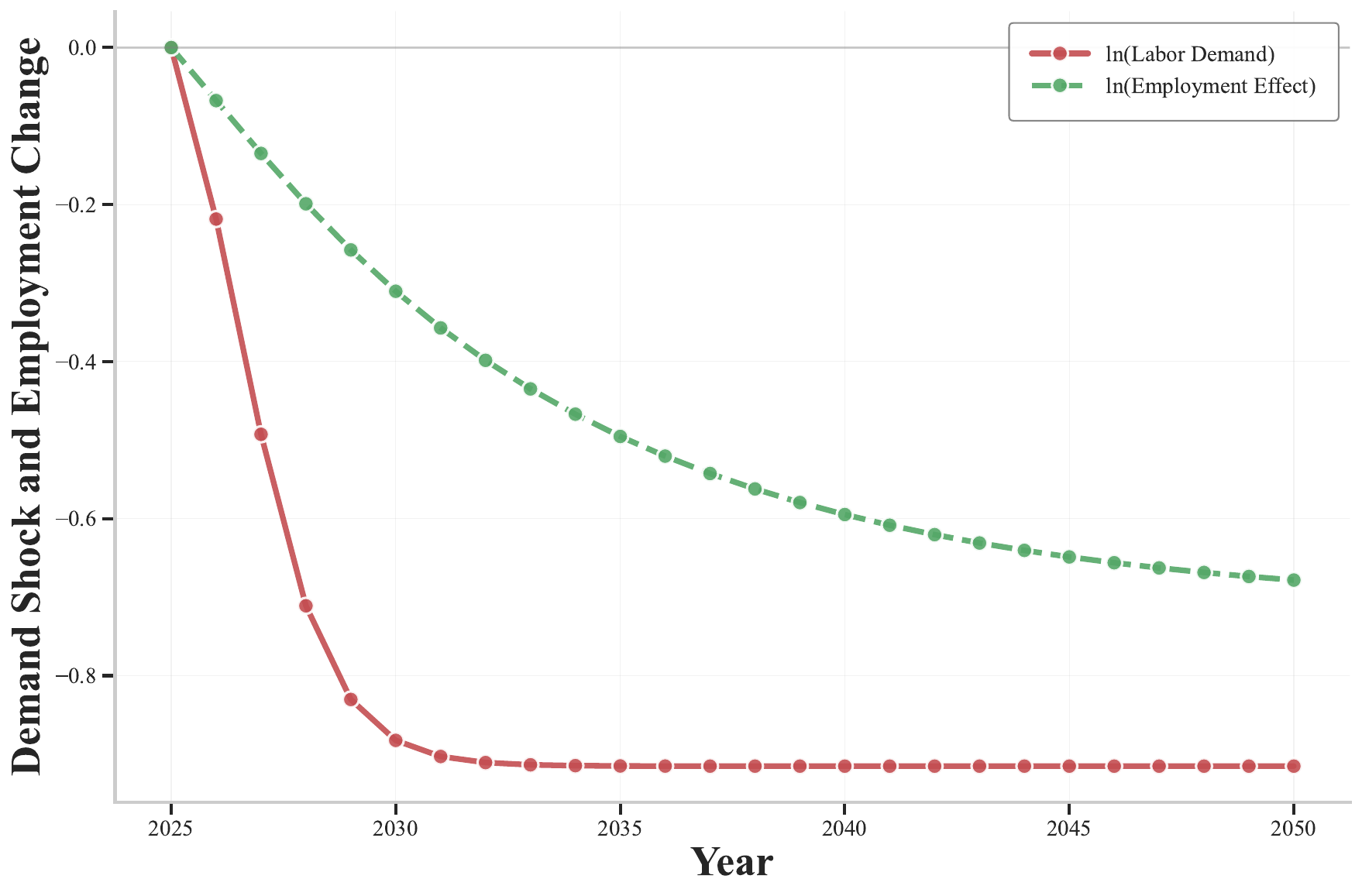}} \hfill
    \subcaptionbox{Wage Incidence of AI}{\includegraphics[scale=0.25]{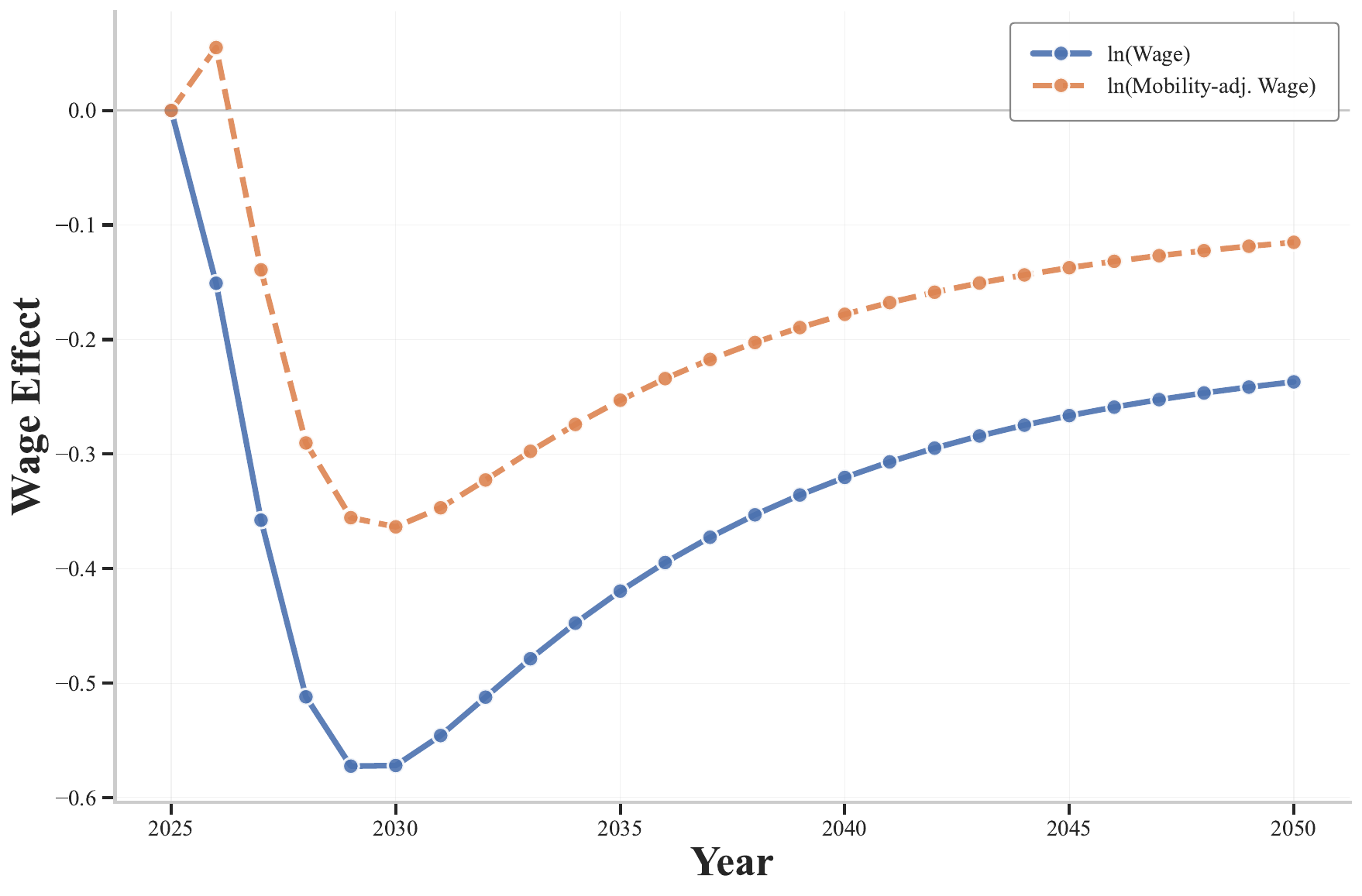}}
    \caption{The Dynamic Incidence of Automation and AI}
    \note{\textit{Notes:} Time path of average effects across occupations. Panels A and B show employment and wage effects of automation exposure, while Panels C and D depict projected effects of rapid AI adoption. Dashed lines represent changes in labor demand, green lines indicate employment shifts, and blue lines capture wage incidence. Orange lines in Panels B and D account for mobility-adjusted wage changes.}
    \label{f:dynamic_incidence}
\end{figure}

Panel A of Figure \ref{f:dynamic_incidence} illustrates the relative decline in occupational labor demand due to automation exposure (dashed line) alongside the relative demand changes absorbed by employment shifts (green line). Employment adjustments mitigate roughly two-thirds of relative demand changes, with the remaining one-third materializing as relative wage changes (blue line in Panel B). In Panel B, the orange line represents cumulative mobility-adjusted wage changes, given by $\sum_{s=1985}^{t} \ln\left(\hat{w}_{o,s}/\hat{\tilde{\mu}}_{oo,s}^{\kappa/\theta}\right)$, which accounts for worker mobility gains. These gains offset approximately half of wage losses. Compared to the static model, mobility gains are higher because we allow workers to redraw productivity.\footnote{Workers in current jobs typically have higher occupation-specific productivity due to selection; if productivity were permanent, they would face greater losses when transitioning.} Furthermore, because workers are forward-looking, mobility gains occur early in the adjustment process, as outside options improve immediately for negatively affected jobs, while wage effects accumulate gradually.

The gradual wage impact of automation suggests progressive adoption over the past four decades, allowing the labor market to absorb roughly two-thirds of associated labor demand shifts. This gradual adoption makes labor market adjustment in transition similar to that in the long run. However, if AI advances rapidly—as many practitioners advocate—the labor market may face greater adjustment challenges. To explore this scenario, we consider a counterfactual in which AI adoption reaches automation's scale by 2030,\footnote{Specifically, we set the AI demand shock proportional to automation's shock scaled by relative exposure: $\hat{\boldsymbol{\alpha}}^{\text{AI}} / \boldsymbol{z}^{\text{AI}} = \hat{\boldsymbol{\alpha}}^{\text{Automation}} / \boldsymbol{z}^{\text{Automation}}$. This implies occupations face demand changes proportional to their AI exposure, with the aggregate magnitude matching automation's cumulative impact.} allowing us to evaluate labor market responses to rapid technological transition.

Panels C and D of Figure \ref{f:dynamic_incidence} illustrate the dynamic incidence of accelerated AI adoption. Panel C shows that the labor market adjusts sluggishly, absorbing less than one-third of relative demand shifts initially, with another third absorbed over subsequent decades. In Panel D, occupations highly exposed to AI experience sharp wage declines as full adoption materializes by 2030, followed by gradual recovery. Mobility gains offset approximately one-third of relative wage losses during transition. These findings suggest that slow labor market adjustment severely limits its ability to absorb rapid AI advancement impacts.

These findings underscore a key insight extending beyond the static model: clustering of both automation and AI exposure constrains worker mobility, limiting the labor market's capacity to absorb shocks through occupational transitions in both the short and long run. For automation, gradual adoption since 1985 allowed the labor market to absorb roughly two-thirds of demand shifts, yet this adjustment generated persistent wage disparities reaching 50\% between high and low exposure occupations. For rapid AI expansion, mobility constraints operate more severely during transition: the labor market initially absorbs less than one-third of shocks, with mobility offsetting only one-third of wage losses, amplifying short-run inequality before gradual convergence.

\section{Additional Extensions} \label{s:add_extension}

This section extends our DIDES framework along two important dimensions: alternative specifications of worker efficiency and heterogeneous groups.
\subsection{Worker Heterogeneity and Effective Labor Supply} \label{ss:efficiency}

Our baseline model assumes that idiosyncratic worker productivity affects only the disutility of labor supply, operating through preferences rather than production. Here we examine a generalized specification where a fraction $\delta \in [0,1]$ of workers contribute their idiosyncratic productivity directly to production, while the remaining $(1-\delta)$ experience reduced effort costs proportional to their productivity \citep{Hsieh2019-np}.

This hybrid specification yields identical employment shares but distinct effective labor supplies. The effective labor in occupation $o$ becomes:
\begin{equation*}
    L_o^{\text{eff}} = \pi_o\bar{L}\left[\delta \cdot \mathbb{E}[\epsilon_o(i)|o^*(i)=o] + (1-\delta)\right]
\end{equation*}
where $\mathbb{E}[\epsilon_o(i)|o^*(i)=o] = \Gamma(1-1/\theta)W/w_o$ represents the conditional expectation of productivity, with $W$ being the wage index. For $\delta$ workers, higher wages attract workers with lower average productivity due to selection effects, partially offsetting employment increases.

The effective labor supply elasticity becomes (see Appendix \ref{app:efficiency} for derivation):
\begin{equation*}
    \Theta_{oo'}^{\text{eff}} = \begin{cases}
        \theta\left[\frac{x_oF_{oo}}{F_o}\right] + (\theta-s_o^{\delta})(1-\pi_o) & \text{if } o = o' \\
        \theta\left[\frac{x_{o'}F_{oo'}}{F_o}\right] - (\theta-s_{o}^{\delta})\pi_{o'} & \text{if } o \neq o'
    \end{cases}
\end{equation*}
where $s_o^{\delta}$ is the share of efficiency units in occupation $o$ contributed by productivity workers.\footnote{Specifically, $s_o^{\delta} = \delta \cdot \mathbb{E}[\epsilon_o(i)|o^*(i)=o]/[\delta \cdot \mathbb{E}[\epsilon_o(i)|o^*(i)=o] + (1-\delta)] = \delta \cdot \Gamma(1-1/\theta)W/w_o/[\delta \cdot \Gamma(1-1/\theta)W/w_o + (1-\delta)]$. This share increases with $\delta$ and decreases with the occupation's relative wage $w_o/W$.}

Two key insights emerge from this decomposition. First, the correlation term $\theta x_{o'}F_{oo'}/F_o$, capturing skill-based substitution patterns, remains unaffected by worker composition. Second, the independent substitution coefficient decreases from $\theta$ to $(\theta-s_o^{\delta})$ as the share of efficiency units from productivity workers rises. This occurs because for workers with productivity in production, cross-skill substitutability arises from the dispersion of productivity draws, which is inherently less transferable across occupations. The correlation structure, however, reflects the underlying skill intensities that determine occupational proximity, remaining invariant to how productivity manifests.

\paragraph{Implications for Incidence} 
Consider two extreme cases. Our baseline specification ($\delta = 0$) implies $s_o^{\delta} = 0$ for all occupations, yielding maximal labor supply elasticities with independent coefficient $\theta = 1.10$. At the other extreme, when $\delta = 1$ (all workers contribute productivity to production), $s_o^{\delta} = 1$ for all occupations by construction, reducing the independent coefficient to $(\theta-1) = 0.10$—essentially eliminating independent substitutability.

This reduction amplifies wage pass-through.\footnote{Mobility gains from reallocation depend on transition patterns rather than whether productivity enters production or preferences. For given wage changes, the welfare metric (equivalent variation) remains comparable across specifications with different $\delta$ values.} With $\Delta^{\text{eff}} = (\mathbf{I} + \Theta^{\text{eff}}/\sigma)^{-1}$, the pass-through matrix is inversely related to labor supply elasticities. Appendix \ref{app:passthrough_efficiency} quantifies this amplification: under $\delta = 1$, wage pass-through ranges from 0.20 to 0.70 across occupations (compared to 0.387 for CES), with automation-exposed occupations clustering at 0.60--0.65—nearly double the CES benchmark. Our baseline estimates thus provide lower bounds on wage inequality from technological clustering. The difference is substantial: moving from $\delta = 0$ to $\delta = 1$ would reduce independent labor market flexibility by over 90\%, leaving primarily correlation-driven substitution patterns to mediate technological shocks.
\subsection{Automation and Between-Group Inequality} \label{ss:distortions}

We extend our framework to examine how automation interacts with group-specific differences in occupational sorting to shape inequality across demographic groups \citep{Burstein2019-jl}. Following Section \ref{ss:hetero_workers}, we allow for group-specific productivity $A_o^g$ across occupations. The share of group $g$ workers choosing occupation $o$ is:
\begin{equation*}
    \pi_{o,t}^g = \frac{A_{o,t}^g w_{o,t}^{\theta} \cdot F_o(A_{1,t}^g w_{1,t}^\theta, \ldots, A_{O,t}^g w_{O,t}^\theta)}{F(A_{1,t}^g w_{1,t}^\theta, \ldots, A_{O,t}^g w_{O,t}^\theta)}
\end{equation*}

The correlation structure $F(\cdot)$ from our DIDES framework governs how groups sort across occupations. As detailed in Appendix \ref{app_ss:distortions}, the group-occupation productivity terms $\{A_{o,t}^g\}$ can be decomposed into labor productivity, pecuniary discrimination (wage penalties), and non-pecuniary barriers (amenity differences, cultural frictions). Observed changes in employment shares and wages allow us to separately identify these components over time, following the approach of \cite{Hsieh2019-np}.

\paragraph{Mobility Gains Across Demographic Groups}

The substitution structure determines how effectively different demographic groups can shield themselves from automation through occupational reallocation. We measure group-specific mobility gains as:
\begin{equation*}
    \text{Mobility Gain}^g = \Delta \ln W_{1980,\text{Automation}}^{g} - \sum_{o=1}^{O} \pi_{o,1980}^g \cdot \Delta \ln w_o
\end{equation*}
where $\Delta \ln W_{1980,\text{Automation}}^{g}$ is the change in the aggregate wage index for group $g$ under automation from 1980 to 2010 (computed using $\hat{w}_o^{\text{Automation}}$), and the second term represents the wage change that would occur absent any mobility. The difference quantifies welfare gains from workers' ability to transition toward less-affected occupations.

Table \ref{tab:mobility_gains_groups} reports mobility gains under two specifications of the correlation function: our estimated DIDES structure and the CES specification with uniform elasticity. The results reveal substantial heterogeneity in mobility capacity across demographic groups and substantial bias when imposing uniform substitutability.

\begin{table}[ht]
\centering
\caption{Mobility Gains from Automation by Demographic Group, 1980}
\label{tab:mobility_gains_groups}
\resizebox{\textwidth}{!}{%
\begin{tabular}{lcccccccc}
\hline\hline
                    & White    & White    & Black    & Black    & Hispanic & Hispanic & Other    & Other    \\
                    & Men      & Women    & Men      & Women    & Men      & Women    & Men      & Women    \\
\hline
No Mobility         & 0.00     & $-2.49$  & $-5.19$  & $-3.39$  & $-4.49$  & $-4.23$  & $-0.38$  & $-2.20$  \\
\addlinespace
DIDES               & 1.26     & 1.35     & 1.48     & 1.21     & 1.44     & 1.15     & 1.34     & 1.39     \\
\addlinespace
CES                 & 1.55     & 1.21     & 1.53     & 1.14     & 1.54     & 1.06     & 1.62     & 1.29     \\
\addlinespace
CES Bias            & 23.2     & $-10.4$  & 3.4      & $-5.8$   & 6.9      & $-7.8$   & 20.7     & $-7.1$   \\
\hline\hline
\end{tabular}%
}
\note{\textit{Notes:} Mobility gains measure the difference between aggregate wage index changes and employment-weighted occupation-specific wage changes. "No Mobility" shows relative wage changes holding employment shares fixed. "DIDES" uses our estimated distance-dependent correlation structure. "CES" imposes uniform elasticity across all occupations. "CES Bias" reports $(\text{CES} - \text{DIDES})/\text{DIDES} \times 100$. All values in log points (multiplied by 100) except CES Bias (in percent). The no-mobility row is normalized relative to white men due to the missing level intercept in our analysis.}
\end{table}

Three patterns emerge from Table \ref{tab:mobility_gains_groups}. First, mobility provides meaningful but limited compensation for automation-induced wage losses across all demographic groups: groups recover 1.2-1.5 log points through reallocation. This compensation is similar in magnitude across groups despite their different exposure to automation. Second, absent mobility, groups initially concentrated in manual-intensive occupations (Black men, Hispanic men and women) would face substantially larger relative wage losses—between 4.2 and 5.2 log points compared to white men—demonstrating their higher direct exposure to automation. However, their mobility gains are not larger, suggesting that occupational clustering constrains adjustment capacity even for heavily exposed groups. Third, the CES specification mismeasures mobility gains, with errors ranging from $-10.4\%$ to $+23.2\%$ across groups. The CES model overstates mobility gains for men (particularly white and other men) while understating gains for most women, generating misleading conclusions about which groups can best adapt to technological change.

These findings demonstrate that substitution patterns matter not just for aggregate incidence but for understanding distributional consequences across demographic groups. The heterogeneity in CES bias reflects how correlation structure interacts with group-specific employment patterns. Men's initial concentration in manual intensive occupations makes their mobility particularly insensitive to automation, while women's more dispersed occupational distribution across both manual and cognitive domains generates different mobility patterns. The CES specification, by imposing uniform substitutability, misses that workers can easily transition between skill-similar occupations but find limited shelter when shocks cluster. 

\section{Conclusion} \label{s:conclusion}

This paper establishes that the labor market incidence of technological change depends fundamentally on the interaction between shock distribution and the substitution structure. By developing and implementing a framework with distance-dependent elasticity of substitution (DIDES), we reveal how technological clustering in skill-adjacent occupations constrains employment reallocation and amplifies wage adjustment from automation and AI.

Our theoretical contribution embeds DIDES into a Roy model through correlated productivity draws that decline with skill distance. This achieves crucial dimensionality reduction, collapsing hundreds of thousands of bilateral elasticities into just four parameters governing a three-dimensional skill space. When technological shocks cluster in skill-adjacent occupations, they align with low-eigenvalue eigenvectors of the substitution matrix, forcing wage adjustment rather than employment reallocation.

Empirically, we map 306 occupations into cognitive, manual, and interpersonal skill dimensions and estimate that two-thirds of substitution occurs within skill clusters. Cognitive skills prove most transferable ($\rho_{\text{cog}} = 0.77$) while manual skills show limited portability ($\rho_{\text{man}} = 0.48$). These heterogeneous elasticities interact with technological clustering to generate striking patterns: on average, 36\% of demand shocks from both automation and AI translate to wages (versus 30\% under CES)—implying standard models overstate labor supply elasticity by 31\%. For most automation-exposed occupations, pass-through rates reach 45\%, generating wage effects 50\% larger than CES predictions. Workers recover only 20\% of wage losses through mobility, compared to 30\% predicted by standard models.

The dynamic analysis reveals persistent constraints. Gradual automation since 1985 generated wage gaps of up to 50\% between high and low exposure occupations. Rapid AI adoption shows starker patterns: less than one-third of shocks absorbed initially, with mobility offsetting only one-third of wage losses. This sluggish adjustment reflects clustered shocks eliminating transition pathways precisely where workers need them most.

Three key insights emerge. First, technological clustering is the fundamental driver of distributional consequences: when shocks concentrate in skill-adjacent occupations, labor market absorption through employment adjustment is constrained. Second, heterogeneous skill transferability creates asymmetric mobility gains: manual workers face high automation exposure combined with low transferability, experiencing severe losses, while cognitive workers threatened by AI benefit from higher transferability that enables more effective mobility. Third, conventional frameworks with uniform elasticities underestimate wage effects by 20\% on average and 50\% for heavily exposed occupations, obscuring the severity of technological incidence. As AI deployment accelerates with concentration in cognitive occupations, the labor market faces adjustment challenges that may exceed those from historical automation.

Future research should explore how rigid labor market adjustment shapes firms' incentives for technology innovation and adoption. The DIDES framework extends beyond technology to any distributional labor market shock: trade shocks, climate change, or demographic shifts. The central message is clear: technological progress need not generate severe inequality, but clustering combined with skill-based mobility constraints ensures that it does. Recognizing this mechanism, invisible to frameworks assuming uniform substitution, is essential for policies that protect workers while facilitating adjustment in an increasingly AI-automated economy.

\newpage
\bibliography{\bib}

\newpage
\appendix

\section{Proof of Results in Main Text}\label{a:appendix}

\subsection{Production and Labor Demand} \label{app_ss:production}

\subsubsection{Task Framework} \label{app_sss:task_framework}

Following \cite{Acemoglu2022-lv}, we begin with a task-based production framework where final output aggregates a continuum of tasks $\mathcal{T}$ through a constant elasticity of substitution technology:
\begin{equation*}
    y = \left(\int_{\mathcal{T}} y(x)^{\frac{\sigma-1}{\sigma}} dx\right)^{\frac{\sigma}{\sigma-1}}
\end{equation*}
where $y(x)$ denotes the input of task $x$ and $\sigma > 1$ is the elasticity of substitution between tasks.

The task space is partitioned across $O$ occupations, $\mathcal{O} = \{\mathcal{T}_1, \mathcal{T}_2, \ldots, \mathcal{T}_O\}$, where each task belongs exclusively to one occupation:
\begin{equation*}
    \mathcal{T} = \bigcup_{o=1}^{O} \mathcal{T}_o \quad \text{with} \quad \mathcal{T}_i \cap \mathcal{T}_j = \emptyset \text{ for } i \neq j
\end{equation*}

Each task can be produced using either labor or capital under perfect substitution:
\begin{equation*}
    y(x) = \ell_o(x) + a(x)k(x), \quad \forall x \in \mathcal{T}_o
\end{equation*}
where $\ell_o(x)$ is labor input from occupation $o$, $k(x)$ is capital input, and $a(x)$ represents task-specific capital productivity. Capital is produced from final output at unit cost.

\subsubsection{Labor Demand} \label{app_sss:labor_demand}

Given occupational wages $\{w_o\}_{o=1}^O$, cost minimization determines the optimal allocation of tasks between labor and capital. For each occupation $o$, tasks are assigned according to:
\begin{equation*}
    \mathcal{T}_o^{\ell} = \{x \in \mathcal{T}_o : w_o \leq 1/a(x)\} \quad \text{and} \quad \mathcal{T}_o^k = \{x \in \mathcal{T}_o : w_o > 1/a(x)\}
\end{equation*}
where $\mathcal{T}_o^{\ell}$ denotes tasks performed by labor and $\mathcal{T}_o^k$ denotes tasks performed by capital.

The equilibrium price of each task equals the unit cost of production:
\begin{equation*}
    p(x) = \begin{cases}
        1/a(x) & \text{if } x \in \mathcal{T}_o^k \\
        w_o & \text{if } x \in \mathcal{T}_o^{\ell}
    \end{cases}
\end{equation*}

Task demand follows from the CES structure: $y(x) = y \cdot p(x)^{-\sigma}$. Integrating over all tasks performed by occupation $o$ yields labor demand:
\begin{equation*}
    L_o = \int_{x \in \mathcal{T}_o^{\ell}} \ell_o(x) dx = y \cdot w_o^{-\sigma} \cdot M_{\mathcal{T}_o^{\ell}}
\end{equation*}
where $M_{\mathcal{T}_o^{\ell}} = \int_{\mathcal{T}_o^{\ell}} dx$ is the measure of tasks performed by occupation $o$.

\subsubsection{Reduced-Form Representation} \label{app_sss:reduced_form}

The zero-profit condition implies:
\begin{equation*}
    1 = \int_{\mathcal{T}} p(x)^{1-\sigma} dx = \int_{\cup_o \mathcal{T}_o^k} a(x)^{\sigma-1} dx + \sum_{o=1}^O w_o^{1-\sigma} \cdot M_{\mathcal{T}_o^{\ell}}
\end{equation*}

Define the share of tasks performed by labor in occupation $o$ as $s_o^{\ell} = M_{\mathcal{T}_o^{\ell}}/M_{\mathcal{T}}$, where $M_{\mathcal{T}}$ is the total measure of tasks. Similarly, let $s^k = 1 - \sum_o s_o^{\ell}$ denote the share of tasks performed by capital, with average capital productivity $a^k$ such that $s^k(a^k)^{\sigma-1} = \int_{\cup_o \mathcal{T}_o^k} a(x)^{\sigma-1} dx/M_{\mathcal{T}}$.

Solving for equilibrium output and substituting out capital yields the reduced-form production function:
\begin{equation*}
    y = \mathcal{A} \left(\sum_{o=1}^O \alpha_o^{\frac{1}{\sigma}} L_o^{\frac{\sigma-1}{\sigma}}\right)^{\frac{\sigma}{\sigma-1}}
\end{equation*}
where:
\begin{align*}
    \mathcal{A} &= \left[1 - s^k(a^k)^{\sigma-1}\right]^{-\frac{\sigma}{\sigma-1}} \quad \text{(aggregate productivity)} \\
    \alpha_o &= \frac{s_o^{\ell}}{1 - s^k(a^k)^{\sigma-1}} \quad \text{(effective labor share of occupation } o\text{)}
\end{align*}

This reduced form captures the essential features of the task model: $\alpha_o$ represents occupation $o$'s share of labor-performed tasks after accounting for automation. When technology advances increase $a(x)$ for tasks in $\mathcal{T}_o$, more tasks shift from labor to capital, reducing $s_o^{\ell}$ and hence $\alpha_o$. The occupational wage then follows:
\begin{equation} \label{app_eq:wage_equation}
    w_o = \frac{\partial y}{\partial L_o} = y^{\frac{1}{\sigma}} \alpha_o^{\frac{1}{\sigma}} L_o^{-\frac{1}{\sigma}}
\end{equation}

This parsimonious representation allows us to analyze the distributional effects of automation and AI through changes in task shares $\{\alpha_o\}$ without explicitly tracking individual task assignments.

\subsection{Workers and Labor Supply} \label{app_ss:labor_supply}

\subsubsection{Properties of the Correlation Function} \label{app_sss:correlation_properties}

The labor supply side of our model builds on a Roy framework with correlated productivity across occupations. Central to our analysis is the correlation function $F: \mathbb{R}_{+}^O \rightarrow \mathbb{R}_{+}$, which governs the substitution structure between occupations. This function satisfies three key properties:

\begin{enumerate}
    \item \textbf{Homogeneity of degree one:} $F(\lambda x_1, \ldots, \lambda x_O) = \lambda F(x_1, \ldots, x_O)$ for all $\lambda > 0$
    \item \textbf{Unboundedness:} $\lim_{x_o \to \infty} F(x_1, \ldots, x_O) = \infty$ for any $o$
    \item \textbf{Sign-switching property:} Mixed partial derivatives alternate in sign—the $n$-th order mixed partial is non-negative if $n$ is odd and non-positive if $n$ is even
\end{enumerate}

The sign-switching property ensures that occupations are gross substitutes from workers' perspective, a crucial feature for equilibrium uniqueness. Additionally, $C(u_1, \ldots, u_O) = \exp[-F(-\ln u_1, \ldots, -\ln u_O)]$ forms a max-stable copula, guaranteeing that workers' occupational choices aggregate consistently across the population.\footnote{Max-stability ensures that $C(u_1, \ldots, u_O) = C(u_1^{1/m}, \ldots, u_O^{1/m})^m$ for all $m > 0$ and $(u_1, \ldots, u_O) \in [0,1]^O$. This property is essential for the aggregation of individual choices to yield tractable labor supply functions.}

\subsubsection{Labor Supply} \label{app_sss:worker_optimization}

Workers are heterogeneous in their productivity across occupations. Each worker $i$ draws a productivity vector $\boldsymbol{\epsilon}(i) = \{\epsilon_o(i)\}_{o=1}^O$ from the joint distribution:
\begin{equation*}
    \Pr[\epsilon_1(i) \leq \epsilon_1, \ldots, \epsilon_O(i) \leq \epsilon_O] = \exp\left[-F\left(A_1\epsilon_1^{-\theta}, \ldots, A_O\epsilon_O^{-\theta}\right)\right]
\end{equation*}
where $A_o > 0$ captures average productivity in occupation $o$ and $\theta > 0$ governs the dispersion of productivity across workers. The marginal distributions are Fréchet: $\Pr[\epsilon_o(i) \leq \epsilon] = \exp(-A_o\epsilon^{-\theta})$.

Workers choose occupations to maximize utility. A worker with productivity vector $\boldsymbol{\epsilon}(i)$ receives utility $u_o(i) = w_o\epsilon_o(i)$ from working in occupation $o$, where the productivity term captures both output produced and the inverse of effort cost. The optimal choice is:
\begin{equation}
    o^*(i) = \arg\max_{o} \{w_o\epsilon_o(i)\}
\end{equation}

Given this optimization, the fraction of workers choosing occupation $o$ is:
\begin{equation} \label{app_eq:employment_share}
    \pi_o = \Pr[w_o\epsilon_o(i) = \max_{o'} w_{o'}\epsilon_{o'}(i)] = \frac{A_o w_o^{\theta} F_o(A_1w_1^{\theta}, \ldots, A_Ow_O^{\theta})}{F(A_1w_1^{\theta}, \ldots, A_Ow_O^{\theta})}
\end{equation}
where $F_o = \partial F/\partial x_o$ denotes the partial derivative with respect to the $o$-th argument.\footnote{The derivation of employment shares follows from the principle of maximum stability for multivariate extreme value distributions. See Section \ref{app_ss:emp_share_derivation} for the complete proof.}

Total labor supply to occupation $o$ is $L_o = \pi_o \bar{L}$, where $\bar{L}$ is the total workforce. The elasticity of labor supply with respect to wages determines how workers reallocate across occupations:
\begin{equation} \label{app_eq:elasticity_matrix}
    \Theta_{oo'} \equiv \frac{\partial \ln L_o}{\partial \ln w_{o'}} = 
    \begin{cases}
        \theta\left[\frac{x_{o'}F_{oo'}}{F_o}\bigg|_{x_j = A_jw_j^{\theta}} - \pi_{o'}\right] & \text{if } o \neq o' \\
        \theta\left[\frac{x_oF_{oo}}{F_o}\bigg|_{x_j = A_jw_j^{\theta}} + 1 - \pi_o\right] & \text{if } o = o'
    \end{cases}
\end{equation}
where $x_o = A_ow_o^{\theta}$ for notational convenience. The derivation of these elasticities from the employment share equation is provided in Section \ref{app_ss:elasticity_derivation}.

The cross-elasticities $\Theta_{oo'}$ for $o \neq o'$ are negative (reflecting substitution) while own-elasticities $\Theta_{oo}$ are positive. Importantly, $\sum_{o'} \Theta_{oo'} = 0$, confirming that proportional wage increases do not affect relative employment—only relative wage changes induce reallocation.\footnote{See Section \ref{app_ss:zero_row_sum} for the proof that row sums equal zero using the homogeneity property of $F$.}


\subsection{CNCES Microfoundation \ref{prop:cnces}} \label{app_ss:cnces_proof}

We derive the joint productivity distribution from the skill-specific distributions and the max operator.

\textbf{Step 1: Skill-specific productivity.}
For each skill $s$, productivity follows a correlated Fréchet distribution:
\begin{equation*}
    \Pr[\epsilon_1^s(i) \leq \epsilon_1^s, \ldots, \epsilon_O^s(i) \leq \epsilon_O^s] = \exp\left[-\left(\sum_{o=1}^O (\epsilon_o^s)^{\frac{-\theta}{1-\rho_s}}\right)^{1-\rho_s}\right]
\end{equation*}

\textbf{Step 2: Occupational productivity as maximum.}
Since $\epsilon_o(i) = \max_{s \in \mathcal{S}} A_o^s \cdot \epsilon_o^s(i)$, we have:
\begin{align*}
    \Pr[\epsilon_1(i) \leq \epsilon_1, \ldots, \epsilon_O(i) \leq \epsilon_O] &= \Pr\left[\epsilon_o^s(i) \leq \frac{\epsilon_o}{A_o^s}, \forall o, \forall s\right] \\
    &= \prod_{s \in \mathcal{S}} \Pr\left[\epsilon_o^s(i) \leq \frac{\epsilon_o}{A_o^s}, \forall o\right]
\end{align*}
where the product follows from independence across skills.

\textbf{Step 3: Substitute and simplify.}
Using the skill-specific distribution:
\begin{align*}
    &= \prod_{s \in \mathcal{S}} \exp\left[-\left(\sum_{o=1}^O (A_o^s)^{\frac{\theta}{1-\rho_s}} \epsilon_o^{\frac{-\theta}{1-\rho_s}}\right)^{1-\rho_s}\right] \\
    &= \exp\left[-\sum_{s \in \mathcal{S}} \left(\sum_{o=1}^O \left[(A_o^s)^{\theta} \epsilon_o^{-\theta}\right]^{\frac{1}{1-\rho_s}}\right)^{1-\rho_s}\right]
\end{align*}

\textbf{Step 4: Define aggregate parameters.}
Let $A_o = \sum_s (A_o^s)^{\theta}$ and $\omega_o^s = (A_o^s)^{\theta}/A_o$. Then:
\begin{equation*}
    (A_o^s)^{\theta} \epsilon_o^{-\theta} = \omega_o^s A_o \epsilon_o^{-\theta}
\end{equation*}

Substituting yields the correlation function:
\begin{equation*}
    F(x_1, \ldots, x_O) = \sum_{s \in \mathcal{S}} \left[\sum_{o=1}^O (\omega_o^s x_o)^{\frac{1}{1-\rho_s}}\right]^{1-\rho_s}
\end{equation*}
where $x_o = A_o\epsilon_o^{-\theta}$, completing the proof.
\subsection{CNCES Employment and Elasticities \ref{prop:cnces_elasticity}} \label{app_ss:cnces_elasticity_proof}

We derive the employment shares and correlated elasticities under the CNCES structure.

\subsubsection*{Part 1: Employment Shares}

From equation \eqref{app_eq:employment_share}, the employment share is:
\begin{equation*}
    \pi_o = \frac{A_ow_o^{\theta}F_o(x_1, \ldots, x_O)}{F(x_1, \ldots, x_O)}
\end{equation*}

For CNCES, $F(x_1, \ldots, x_O) = \sum_{s \in \mathcal{S}} G_s^{1-\rho_s}$ where $G_s = \sum_{o'} (\omega_{o'}^s x_{o'})^{\frac{1}{1-\rho_s}}$.

The partial derivative is:
\begin{equation*}
    F_o = \frac{\partial F}{\partial x_o} = \sum_{s \in \mathcal{S}} G_s^{-\rho_s} \omega_o^s (\omega_o^s x_o)^{\frac{\rho_s}{1-\rho_s}}
\end{equation*}

Therefore:
\begin{align*}
    \pi_o &= \frac{x_o F_o}{\sum_{o'} x_{o'} F_{o'}} = \sum_{s \in \mathcal{S}} \frac{(\omega_o^s x_o)^{\frac{1}{1-\rho_s}}}{G_s} \cdot \frac{G_s^{1-\rho_s}}{F} \\
    &= \sum_{s \in \mathcal{S}} \underbrace{\frac{(\omega_o^s A_o w_o^{\theta})^{\frac{1}{1-\rho_s}}}{\sum_{o'} (\omega_{o'}^s A_{o'} w_{o'}^{\theta})^{\frac{1}{1-\rho_s}}}}_{\pi_o^{s,W}} \cdot \underbrace{\frac{G_s^{1-\rho_s}}{\sum_{s'} G_{s'}^{1-\rho_{s'}}}}_{\pi^{s,B}}
\end{align*}

This establishes equation \eqref{eq:cnces_shares} with $\pi_o^s = \pi_o^{s,W} \cdot \pi^{s,B}$.

\subsubsection*{Part 2: Correlated Elasticities}

To derive equation \eqref{eq:cnces_elasticity}, we need the second derivative:
\begin{equation*}
    F_{oo'} = \frac{\partial^2 F}{\partial x_o \partial x_{o'}} = -\sum_{s \in \mathcal{S}} \frac{\rho_s}{1-\rho_s} G_s^{-\rho_s-1} \cdot \omega_o^s (\omega_o^s x_o)^{\frac{\rho_s}{1-\rho_s}} \cdot \omega_{o'}^s (\omega_{o'}^s x_{o'})^{\frac{\rho_s}{1-\rho_s}}
\end{equation*}

The ratio becomes:
\begin{equation*}
    \frac{x_{o'} F_{oo'}}{F_o} = -\sum_{s \in \mathcal{S}} \frac{\rho_s}{1-\rho_s} \cdot \underbrace{\frac{(\omega_{o'}^s x_{o'})^{\frac{1}{1-\rho_s}}}{G_s}}_{\mu_{o'}^s} \cdot \underbrace{\frac{G_s^{-\rho_s} \omega_o^s (\omega_o^s x_o)^{\frac{\rho_s}{1-\rho_s}}}{F_o}}_{\gamma_o^s}
\end{equation*}

Substituting $x_o = A_o w_o^{\theta}$ and noting that:
- $\mu_{o'}^s = \pi_{o'}^{s,W}$ (within-skill share)
- $\gamma_o^s = \pi_o^s/\pi_o$ (skill $s$'s contribution to occupation $o$)
- $\pi_o^s = \pi_o^{s,W} \cdot \pi^{s}$

We obtain:
\begin{equation*}
    \theta\frac{x_{o'}F_{oo'}}{F_o}\bigg|_{x_j = A_jw_j^{\theta}} = -\theta \sum_{s \in \mathcal{S}} \frac{\rho_s}{1-\rho_s} \cdot \pi_o^{s,W} \pi_{o'}^{s,W} \cdot \frac{\pi^{s}}{\pi_o}
\end{equation*}

This completes the proof of equation \eqref{eq:cnces_elasticity}.

\subsection{Proof of Proposition \ref{prop:hat_algebra} (Hat Algebra)} \label{app_ss:hat_algebra}

The proof demonstrates how observed employment shares serve as sufficient statistics for predicting counterfactual changes without requiring wage or productivity levels.

\textbf{Step 1: Express employment shares in terms of the correlation function.}

Given wages $\boldsymbol{w}_t$ and group-specific productivity $\{A_t^g\}_{g \in G}$, employment shares are:
\begin{equation*}
\pi^g_{o,t} = \frac{A^g_{o,t} w_{o,t}^{\theta} F_o(A^g_{1,t} w_{1,t}^\theta, \ldots, A^g_{O,t} w_{O,t}^\theta)}{F(A^g_{1,t} w_{1,t}^\theta, \ldots, A^g_{O,t} w_{O,t}^\theta)}
\end{equation*}

\textbf{Step 2: Define correlation-adjusted shares.}

Let the correlation-adjusted employment share be:
\begin{equation*}
\tilde{\pi}^g_{o,t} = \frac{A^g_{o,t} w_{o,t}^{\theta}}{F(A^g_{1,t} w_{1,t}^\theta, \ldots, A^g_{O,t} w_{O,t}^\theta)}
\end{equation*}

Since $F_o$ is homogeneous of degree zero, we obtain:
\begin{equation*}
\pi^g_{o,t} = \tilde{\pi}^g_{o,t} F_o(\tilde{\pi}^g_{1,t}, \ldots, \tilde{\pi}^g_{O,t})
\end{equation*}

This establishes a one-to-one mapping between observed shares $\{\pi^g_{o,t}\}$ and adjusted shares $\{\tilde{\pi}^g_{o,t}\}$.

\textbf{Step 3: Derive the evolution of adjusted shares.}

For wage changes from $t$ to $t+1$, the ratio of adjusted shares is:
\begin{equation*}
\frac{\tilde{\pi}^g_{o,t+1}}{\tilde{\pi}^g_{o,t}} = \frac{(w_{o,t+1}/w_{o,t})^{\theta}}{F(A^g_{1,t} w_{1,t+1}^\theta, \ldots, A^g_{O,t} w_{O,t+1}^\theta)/F(A^g_{1,t} w_{1,t}^\theta, \ldots, A^g_{O,t} w_{O,t}^\theta)}
\end{equation*}

Using the homogeneity property of $F$, the denominator simplifies to:
\begin{equation*}
\frac{F(A^g_{1,t} w_{1,t+1}^\theta, \ldots, A^g_{O,t} w_{O,t+1}^\theta)}{F(A^g_{1,t} w_{1,t}^\theta, \ldots, A^g_{O,t} w_{O,t}^\theta)} = F(\{\hat{w}_{o,t+1}^{\theta} \tilde{\pi}^g_{o,t}\}_{o \in O})
\end{equation*}

where $\hat{w}_{o,t+1} = w_{o,t+1}/w_{o,t}$ denotes the relative wage change.

\textbf{Step 4: Obtain the counterfactual algorithm.}

The adjusted shares evolve according to:
\begin{equation*}
\tilde{\pi}_{o,t+1}^g = \frac{ \hat{w}_{o,t+1}^\theta \tilde{\pi}^g_{o,t} }{F(\{\hat{w}_{o',t+1}^{\theta} \tilde{\pi}^g_{o',t}\}_{o' \in O})}
\end{equation*}

Finally, recover the counterfactual employment shares:
\begin{equation*}
\pi^g_{o,t+1} = \tilde{\pi}^g_{o,t+1} F_o(\tilde{\pi}^g_{1,t+1}, \ldots, \tilde{\pi}^g_{O,t+1})
\end{equation*}

\textbf{Step 5: Derive the change in the aggregate wage index.}

From the definition of correlation-adjusted shares in Step 2, we can express:
\begin{equation*}
F(A^g_{1,t} w_{1,t}^\theta, \ldots, A^g_{O,t} w_{O,t}^\theta) = \frac{A^g_{o,t} w_{o,t}^{\theta}}{\tilde{\pi}^g_{o,t}}
\end{equation*}

for any occupation $o$. Define the aggregate wage index as:
\begin{equation*}
W_t^g \equiv F(A^g_{1,t} w_{1,t}^\theta, \ldots, A^g_{O,t} w_{O,t}^\theta)^{1/\theta}
\end{equation*}

This represents the expected wage for group $g$ workers before observing idiosyncratic productivity draws. Taking the ratio between periods:
\begin{equation*}
\hat{W}_{t+1}^g = \frac{W_{t+1}^g}{W_t^g} = \left[\frac{F(A^g_{1,t} w_{1,t+1}^\theta, \ldots, A^g_{O,t} w_{O,t+1}^\theta)}{F(A^g_{1,t} w_{1,t}^\theta, \ldots, A^g_{O,t} w_{O,t}^\theta)}\right]^{1/\theta}
\end{equation*}

Using the result from Step 3:
\begin{equation*}
\hat{W}_{t+1}^g = F(\{\hat{w}_{o,t+1}^{\theta} \tilde{\pi}^g_{o,t}\}_{o \in O})^{1/\theta}
\end{equation*}

This shows that the counterfactual change in the aggregate wage index can be computed directly from observed employment shares $\{\tilde{\pi}^g_{o,t}\}$ and relative wage changes $\{\hat{w}_{o,t+1}\}$, without requiring knowledge of productivity levels.

This completes the proof and provides an algorithm to compute counterfactual employment shares and aggregate welfare changes using only observed shares and relative wage changes, without requiring knowledge of productivity levels or absolute wages. $\square$
\subsection{Dynamic Model with Forward-Looking Workers} \label{a:appendix:dynamic}

This appendix extends the static framework to incorporate forward-looking occupational choice with adjustment frictions. The dynamic model enables analysis of transition paths and the timing of labor market responses to technological shocks.

\subsubsection{Workers' Dynamic Problem} \label{a:appendix:dynamic_worker}

\textbf{Setup.} Consider a continuum of hand-to-mouth workers distributed across $O$ occupations. Workers maximize expected lifetime utility over consumption $c_t(i)$ and labor effort $\ell_t(i)$:
\begin{equation*}
U\left(\{c_t(i), \ell_t(i)\}_{t=0}^{\infty}\right) = \sum_{t=0}^{\infty} \beta^t[\ln c_t(i) - \ln \ell_t(i)]
\end{equation*}
where $\beta \in (0,1)$ is the discount factor.

\textbf{Productivity draws.} Each period, workers draw productivity vectors 
\[\boldsymbol{\epsilon}_t(i) = (\epsilon_{1,t}(i), \ldots, \epsilon_{O,t}(i))\]
from the same multivariate Fréchet distribution as in the static model:
\begin{equation*}
\Pr[\epsilon_{1,t}(i) \leq \epsilon_1, \ldots, \epsilon_{O,t}(i) \leq \epsilon_O] = \exp\left[-F(A_{1,t}\epsilon_1^{-\theta}, \ldots, A_{O,t}\epsilon_O^{-\theta})\right]
\end{equation*}
where the correlation function $F$ embeds the CNCES structure:
\begin{equation*}
F(x_1, \ldots, x_O) = \sum_{s=1}^S \left[\sum_{o=1}^O (\omega_{so} x_o)^{\frac{1}{1-\rho_s}}\right]^{1-\rho_s}
\end{equation*}

\textbf{Occupational choice with transition costs.} After observing $\boldsymbol{\epsilon}_t(i)$, workers choose occupations subject to transition costs $\tau_{oo'} \geq 0$ (measured in utility units). The instantaneous utility from occupation $o'$ is:
\begin{equation*}
u_t(i) = \ln w_{o',t} + \kappa \ln \epsilon_{o',t}(i)
\end{equation*}
where $\kappa > 0$ governs the short-run labor supply elasticity, capturing sluggish adjustment relative to the static model's long-run elasticity $\theta$.

\textbf{Value function.} The Bellman equation for a worker in occupation $o$ with productivity $\boldsymbol{\epsilon}_t$ is:
\begin{equation*}
v_{o,t}(\boldsymbol{\epsilon}_t) = \max_{o'} \left\{\ln w_{o',t} + \kappa \ln \epsilon_{o',t} + \beta V_{o',t+1} - \tau_{oo'}\right\}
\end{equation*}
where $V_{o',t+1} = \mathbb{E}_{\boldsymbol{\epsilon}}[v_{o',t+1}(\boldsymbol{\epsilon})]$ is the expected continuation value.

\textbf{Aggregation.} Define the inclusive value:
\begin{equation*}
Z_{oo',t} = \exp(\beta V_{o',t+1} + \ln w_{o',t} - \tau_{oo'})
\end{equation*}

Given the Fréchet structure, the expected value simplifies to:
\begin{equation*}
V_{o,t} = \ln\left[F(A_{1,t}Z_{o1,t}^{\theta/\kappa}, \ldots, A_{O,t}Z_{oO,t}^{\theta/\kappa})^{\kappa/\theta}\right] + \bar{\gamma}\frac{\kappa}{\theta}
\end{equation*}
where $\bar{\gamma}$ is the Euler-Mascheroni constant.

This formulation nests the static model when $\kappa = \theta$ (no adjustment frictions) and generates gradual transitions when $\kappa < \theta$ (costly adjustment). The correlation structure $F$ preserves the DIDES property: workers transition more easily between skill-similar occupations, but adjustment slows when technological shocks cluster within skill domains.

\subsubsection{Occupation Switching Probabilities} \label{a:appendix:occ_switch}

This section derives the transition probabilities between occupations, showing how the correlation structure generates realistic mobility patterns.

\textbf{Switching probability.} The probability that a worker in occupation $o$ switches to $o'$ at time $t$ is:
\begin{equation*}
\mu_{oo',t} = \Pr\left[Z_{oo',t} \epsilon_{o',t}^\kappa \geq \max_{o''} Z_{oo'',t} \epsilon_{o'',t}^\kappa\right]
\end{equation*}

Using the properties of the multivariate Fréchet distribution (see Section \ref{app_ss:emp_share_derivation}), this probability becomes:
\begin{equation*}
\mu_{oo',t} = \frac{A_{o',t} Z_{oo',t}^{\theta/\kappa} F_{o'}(A_{1,t}Z_{o1,t}^{\theta/\kappa}, \ldots, A_{O,t}Z_{oO,t}^{\theta/\kappa})}{F(A_{1,t}Z_{o1,t}^{\theta/\kappa}, \ldots, A_{O,t}Z_{oO,t}^{\theta/\kappa})}
\end{equation*}
where $F_{o'} = \partial F/\partial x_{o'}$ denotes the partial derivative.

\textbf{Correlation-adjusted transition rates.} Define the correlation-adjusted transition probability:
\begin{equation*}
\tilde{\mu}_{oo',t} = \frac{A_{o',t} Z_{oo',t}^{\theta/\kappa}}{F(A_{1,t}Z_{o1,t}^{\theta/\kappa}, \ldots, A_{O,t}Z_{oO,t}^{\theta/\kappa})}
\end{equation*}

This adjustment isolates the role of correlation from the baseline substitution effect. The observed and adjusted probabilities are related by:
\begin{equation*}
\mu_{oo',t} = \tilde{\mu}_{oo',t} F_{o'}(\tilde{\mu}_{o1,t}, \ldots, \tilde{\mu}_{oO,t})
\end{equation*}

This establishes a one-to-one mapping between observed transitions $\{\mu_{oo',t}\}$ and adjusted rates $\{\tilde{\mu}_{oo',t}\}$.

\textbf{Euler equation for mobility.} The evolution of adjusted transition rates satisfies:
\begin{equation*}
\ln\frac{\tilde{\mu}_{oo',t}}{\tilde{\mu}_{oo,t}} = \frac{\theta}{\kappa}\ln\frac{w_{o',t}}{w_{o,t}} + \beta\ln\frac{\tilde{\mu}_{oo',t+1}}{\tilde{\mu}_{o'o',t+1}} + (\beta-1)\tau_{oo'}
\end{equation*}

This Euler equation shows that relative transition rates depend on three factors:
\begin{itemize}
    \item Current wage differentials (scaled by $\theta/\kappa$, the short-run elasticity)
    \item Future option values (captured by next period's staying probabilities)
    \item Transition costs (discounted by $\beta-1 < 0$)
\end{itemize}

The correlation-adjusted formulation enables estimation of $\theta/\kappa$ from observed transitions while accounting for the skill-based clustering that constrains mobility between distant occupations.

\subsubsection{Static Production Equilibrium} \label{a:appendix:dynamic_production}

This section characterizes the production side of the economy, which remains static within each period while labor allocations adjust dynamically across periods.

\textbf{Production technology.} The production side follows the reduced-form representation derived in Appendix \ref{app_sss:reduced_form}. Output in period $t$ is given by:
\begin{equation*}
Y_t = \mathcal{A}_t \left(\sum_{o=1}^O \alpha_{o,t}^{\frac{1}{\sigma}} L_{o,t}^{\frac{\sigma-1}{\sigma}}\right)^{\frac{\sigma}{\sigma-1}}
\end{equation*}
where $\mathcal{A}_t$ captures aggregate productivity (incorporating the contribution of capital to production) and $\alpha_{o,t}$ represents occupation $o$'s effective labor share after accounting for task automation.

\textbf{Wage determination.} Competitive labor markets equate wages to marginal products:
\begin{equation*}
w_{o,t} = \frac{\partial Y_t}{\partial L_{o,t}} = Y_t^{\frac{1}{\sigma}}\alpha_{o,t}^{\frac{1}{\sigma}}L_{o,t}^{-\frac{1}{\sigma}}
\end{equation*}

The task shares $\{\alpha_{o,t}\}$ capture the distributional effects of technology: when automation or AI displaces labor from tasks in occupation $o$, the corresponding $\alpha_{o,t}$ declines, reducing wages even as aggregate productivity $\mathcal{A}_t$ may rise through lower production costs. See Appendix \ref{app_sss:reduced_form} for the microfoundation from the task-based framework.

\subsubsection{Dynamic Equilibrium} \label{a:appendix:dynamic_equilibrium}

This section defines the dynamic equilibrium, accounting for data limitations and characterizing the conditions for market clearing across time.

\textbf{Measurement reconciliation.} The retrospective design of the March CPS creates a discrepancy between measured job flows and observed employment levels. We account for this by augmenting the employment evolution equation:
\begin{equation*}
L_{o,t} = \sum_{o'=1}^O \mu_{o'o,t} L_{o',t-1} + \Delta L_{o,t}
\end{equation*}
where $\Delta L_{o,t}$ represents exogenous net inflows/outflows satisfying $\sum_o L_{o,t} = 1$ (normalization) and $\sum_o \Delta L_{o,t} = 0$ (no aggregate employment change).

\textbf{Model primitives.} The economy is characterized by:
\begin{itemize}
    \item Time-varying fundamentals: $\boldsymbol{A}_t = \{A_{o,t}\}$ (productivity), $\boldsymbol{\alpha}_t = \{\alpha_{o,t}\}$ (task shares), $\mathcal{A}_t$ (aggregate productivity)
    \item Structural parameters: $\tau_{oo'}$ (transition costs), $\omega_{os}$ (skill weights), $\sigma$ (demand elasticity), $\theta$ (dispersion), $\rho_s$ (skill correlation), $\kappa$ (short-run elasticity), $\beta$ (discount factor)
\end{itemize}

\textbf{Definition (Dynamic Equilibrium).} A dynamic equilibrium is a sequence $\{\boldsymbol{L}_t, \boldsymbol{w}_t, \mu_t, V_t\}_{t=0}^{\infty}$ satisfying:

\begin{enumerate}
    \item \textit{Production equilibrium:} Wages equal marginal products and output clears markets:
    \begin{align*}
        w_{o,t} &= Y_t^{\frac{1}{\sigma}}\alpha_{o,t}^{\frac{1}{\sigma}} L_{o,t}^{-\frac{1}{\sigma}} \\
        Y_t &= \mathcal{A}_t \left(\sum_o \alpha_{o,t}^{\frac{1}{\sigma}} L_{o,t}^{\frac{\sigma-1}{\sigma}}\right)^{\frac{\sigma}{\sigma-1}}
    \end{align*}
    
    \item \textit{Optimal expectations:} Workers correctly anticipate future values:
    \begin{equation*}
        V_{o,t} = \ln\left[F(A_{1,t}Z_{o1,t}^{\theta/\kappa}, \ldots, A_{O,t}Z_{oO,t}^{\theta/\kappa})^{\kappa/\theta}\right] + \bar{\gamma}\frac{\kappa}{\theta}
    \end{equation*}
    
    \item \textit{Optimal mobility:} Transition probabilities satisfy workers' optimization:
    \begin{equation*}
        \mu_{oo',t} = \frac{A_{o',t} Z_{oo',t}^{\theta/\kappa} F_{o'}(A_{1,t}Z_{o1,t}^{\theta/\kappa}, \ldots, A_{O,t}Z_{oO,t}^{\theta/\kappa})}{F(A_{1,t}Z_{o1,t}^{\theta/\kappa}, \ldots, A_{O,t}Z_{oO,t}^{\theta/\kappa})}
    \end{equation*}
    
    \item \textit{Labor market clearing:} Employment evolves according to transitions:
    \begin{equation*}
        L_{o,t} = \sum_{o'} \mu_{o'o,t} L_{o',t-1} + \Delta L_{o,t}
    \end{equation*}
\end{enumerate}

This equilibrium preserves the DIDES structure: technological shocks that cluster in skill space generate limited mobility (through $\mu$) and force adjustment through wages, creating persistent inequality during transitions.

\subsubsection{System in Changes} \label{a:appendix:sys_changes_dynamic}

This section expresses the dynamic equilibrium in growth rates, facilitating the analysis of transition paths and steady-state convergence.

\textbf{Notation.} Define the growth factor $\dot{x}_{t+1} = x_{t+1}/x_t$ for any variable $x$. For utility, define $u_{o,t} = \exp(V_{o,t})$ to work with levels rather than logs.

\textbf{Production in changes.} Log-differentiating the wage equation yields:
\begin{equation*}
\sigma \ln \dot{w}_{o,t+1} + \ln \dot{L}_{o,t+1} = \ln \dot{Y}_{t+1} + \ln \dot{\alpha}_{o,t+1}
\end{equation*}

This links wage growth to changes in aggregate output, task shares, and employment.

\textbf{Dynamic system in growth rates.} The evolution of correlation-adjusted transition probabilities and expected utilities can be expressed as:
\begin{align}
\frac{\tilde{\mu}_{oo',t}}{\tilde{\mu}_{oo',t-1}} &= \frac{\dot{A}_{o',t}\dot{u}_{o',t+1}^{\beta\theta/\kappa}\dot{w}_{o',t}^{\theta/\kappa}}{F(\{\tilde{\mu}_{oo'',t-1}\dot{A}_{o'',t}\dot{u}_{o'',t+1}^{\beta\theta/\kappa}\dot{w}_{o'',t}^{\theta/\kappa}\}_{o''=1}^{O})} \label{eq:mu_growth}\\
\dot{u}_{o,t+1} &= F(\{\tilde{\mu}_{oo'',t}\dot{A}_{o'',t+1}\dot{u}_{o'',t+2}^{\beta\theta/\kappa}\dot{w}_{o'',t+1}^{\theta/\kappa}\}_{o''=1}^{O})^{\kappa/\theta} \label{eq:u_growth}
\end{align}

The observed transition probabilities follow:
\begin{equation*}
\mu_{oo',t} = \tilde{\mu}_{oo',t} F_{o'}(\tilde{\mu}_{o1,t}, \ldots, \tilde{\mu}_{oO,t})
\end{equation*}

with employment evolving according to the transition matrix and exogenous flows.

\textbf{Interpretation.} Equations \eqref{eq:mu_growth}--\eqref{eq:u_growth} form a forward-looking system where current mobility depends on future expected utilities. The correlation function $F$ preserves the DIDES structure: when technological shocks cluster (affecting the growth rates $\dot{A}_{o,t}$ and $\dot{\alpha}_{o,t}$ in skill-similar occupations), the denominator in \eqref{eq:mu_growth} limits relative mobility adjustments, forcing wage changes to absorb the shock.

See Appendix \ref{c:appendix:add_math_dynamic} for detailed derivations.
\subsubsection{Dynamic Hat Algebra} \label{a:appendix:dynamic_hat_algebra}

This section extends the hat algebra to dynamic settings, enabling counterfactual analysis of transition paths under alternative technological scenarios.

\textbf{Counterfactual notation.} For counterfactual fundamentals $\{\hat{\boldsymbol{A}}_t, \hat{\boldsymbol{\alpha}}_t, \hat{\mathcal{A}}_t\}$, define:
\begin{itemize}
    \item $\hat{x}_t = \dot{x}_t'/\dot{x}_t$: ratio of counterfactual to baseline growth rates
    \item $\dot{x}_t' = x_t'/x_{t-1}'$: counterfactual growth rate
\end{itemize}

\textbf{Counterfactual equilibrium.} The wage response follows from production equilibrium:
\begin{equation*}
\hat{w}_{o,t+1} = \left(\frac{\hat{Y}_{t+1} \hat{\alpha}_{o,t+1}}{\hat{L}_{o,t+1}}\right)^{\frac{1}{\sigma}}
\end{equation*}

Counterfactual transition probabilities evolve recursively:
\begin{equation*}
\tilde{\mu}_{oo',t}' = \frac{\tilde{\mu}_{oo',t-1}' \dot{\tilde{\mu}}_{oo',t} \hat{A}_{o',t} \hat{u}_{o',t+1}^{\beta\theta/\kappa} \hat{w}_{o',t}^{\theta/\kappa}}{F(\{\tilde{\mu}_{oo'',t-1}' \dot{\tilde{\mu}}_{oo'',t} \hat{A}_{o'',t} \hat{u}_{o'',t+1}^{\beta\theta/\kappa} \hat{w}_{o'',t}^{\theta/\kappa}\}_{o''=1}^{O})}
\end{equation*}

Expected utilities adjust according to:
\begin{equation*}
\hat{u}_{o,t+1} = F(\{\tilde{\mu}_{oo'',t}' \dot{\tilde{\mu}}_{oo'',t+1} \hat{A}_{o'',t+1} \hat{u}_{o'',t+2}^{\beta\theta/\kappa} \hat{w}_{o'',t+1}^{\theta/\kappa}\}_{o''=1}^{O})^{\kappa/\theta}
\end{equation*}

with observed transitions $\mu_{oo',t}' = \tilde{\mu}_{oo',t}' F_{o'}(\tilde{\mu}_{o1,t}', \ldots, \tilde{\mu}_{oO,t}')$ and employment evolution:
\begin{equation*}
L_{o,t}' = \sum_{o'} \mu_{o'o,t}' L_{o',t-1}' + \Delta L_{o,t}
\end{equation*}

\subsubsection{Initial Conditions}

For unexpected shocks at $t=1$ (with baseline conditions at $t=0$: $\hat{u}_{o,0} = 1$, $\mu_{oo',0}' = \mu_{oo',0}$, $L_{o,0}' = L_{o,0}$):
\begin{align*}
\tilde{\mu}_{oo',1}' &= \frac{\vartheta_{oo',1} \hat{A}_{o',1} \hat{w}_{o',1}^{\theta/\kappa} \hat{u}_{o',2}^{\beta\theta/\kappa}}{F(\{\vartheta_{oo'',1} \hat{A}_{o'',1} \hat{w}_{o'',1}^{\theta/\kappa} \hat{u}_{o'',2}^{\beta\theta/\kappa}\}_{o''=1}^{O})} \\
\hat{u}_{o,1} &= F(\{\vartheta_{oo',1} \hat{A}_{o',1} \hat{w}_{o',1}^{\theta/\kappa} \hat{u}_{o',2}^{\beta\theta/\kappa}\}_{o'=1}^{O})^{\kappa/\theta}
\end{align*}
where $\vartheta_{oo',1} = \tilde{\mu}_{oo',1} \hat{u}_{o',1}^{\beta\theta/\kappa}$ captures the initial adjustment.

See Appendix \ref{c:appendix:add_math_dynamic} for detailed derivations.
\subsubsection{Welfare Metrics} \label{a:appendix:dynamic_welfare_metric}

This section derives welfare measures that account for both wage changes and mobility gains through the lens of staying probabilities.

\textbf{Value function decomposition.} The recursive value function can be rewritten to highlight the role of staying probabilities:
\begin{align*}
V_{o,t} &= \ln w_{o,t} + \beta V_{o,t+1} + \frac{\kappa}{\theta} \ln\left(\frac{F(A_{1,t}Z_{o1,t}^{\theta/\kappa}, \ldots, A_{O,t}Z_{oO,t}^{\theta/\kappa})}{\exp(\beta V_{o,t+1} + \ln w_{o,t})^{\theta/\kappa}}\right) + \bar{\gamma}\frac{\kappa}{\theta} \\
&= \ln(A_{o,t}^{\kappa/\theta} w_{o,t}) + \beta V_{o,t+1} + \frac{\kappa}{\theta} \ln\left(\frac{1}{\tilde{\mu}_{oo,t}}\right) + \bar{\gamma}\frac{\kappa}{\theta}
\end{align*}

Iterating forward yields:
\begin{equation*}
V_{o,t} = \sum_{s=t}^{\infty} \beta^{s-t} \ln\left[\left(\frac{A_{o,s}}{\tilde{\mu}_{oo,s}}\right)^{\kappa/\theta} w_{o,s}\right] + \frac{\bar{\gamma}\kappa}{\theta(1-\beta)}
\end{equation*}

\textbf{Equivalent variation.} The welfare change from baseline to counterfactual, measured as equivalent variation $\delta_{o,t}$, satisfies:
\begin{equation*}
V_{o,t}' = V_{o,t} + \frac{\ln \delta_{o,t}}{1-\beta}
\end{equation*}

This yields:
\begin{equation*}
\delta_{o,t} = (1-\beta)\sum_{s=t}^{\infty} \beta^{s-t} \ln\left[\frac{w_{o,s}'}{w_{o,s}} \left(\frac{A_{o,s}'/\tilde{\mu}_{oo,s}'}{A_{o,s}/\tilde{\mu}_{oo,s}}\right)^{\kappa/\theta}\right]
\end{equation*}

\textbf{Hat algebra representation.} Expressed in terms of counterfactual changes:
\begin{equation*}
\delta_{o,t} = (1-\beta)\sum_{s=t}^{\infty} \beta^{s-t} \ln\left(\hat{w}_{o,s} \cdot \hat{\tilde{\mu}}_{oo,s}^{-\kappa/\theta}\right)
\end{equation*}

The term $\hat{\tilde{\mu}}_{oo,s}^{-\kappa/\theta}$ captures mobility gains: when technological shocks reduce staying probabilities (workers transition more), this provides partial welfare compensation for wage losses. 

\subsection{Alternative Specification with Productivity in Production} \label{app:efficiency}

Consider an alternative specification where a fraction $\delta \in [0,1]$ of workers contribute their idiosyncratic productivity directly to production, while the remaining $(1-\delta)$ experience reduced effort costs proportional to their productivity as in the baseline model.

\paragraph{Workers with productivity in production ($\delta$ fraction)} 
For these workers, occupation $o$ provides wage $w_o\epsilon_o(i)$ per unit of labor supplied, where $\epsilon_o(i)$ enters production as efficiency units. Their utility from occupation $o$ is:
\begin{equation*}
    u_o(i) = \ln(w_o\epsilon_o(i)) = \ln w_o + \ln\epsilon_o(i)
\end{equation*}

\paragraph{Workers with productivity in preferences ($(1-\delta)$ fraction)} 
As in the baseline model, these workers receive wage $w_o$ and supply effort $\ell_o(i) = 1/\epsilon_o(i)$, yielding utility:
\begin{equation*}
    u_o(i) = \ln w_o - \ln\ell_o(i) = \ln w_o + \ln\epsilon_o(i)
\end{equation*}

Both types make identical occupational choices since utility functions are equivalent. The key difference emerges in aggregation and wage responses.

\paragraph{Effective Labor Supply}
The effective labor supplied to occupation $o$ combines both worker types:
\begin{equation}\label{app:eq:eff_labor}
    L_o^{\text{eff}} = \delta \cdot \pi_o\bar{L} \cdot \mathbb{E}[\epsilon_o(i)|o^*(i)=o] + (1-\delta) \cdot \pi_o\bar{L}
\end{equation}

where the conditional expectation of productivity is:
\begin{equation*}
    \mathbb{E}[\epsilon_o(i)|o^*(i)=o] = \Gamma\left(1-\frac{1}{\theta}\right)\frac{W}{w_o}
\end{equation*}

with wage index $W = F(A_1w_1^{\theta}, \ldots, A_Ow_O^{\theta})^{1/\theta}$. This yields:
\begin{equation}\label{app:eq:eff_labor_final}
    L_o^{\text{eff}} = \pi_o\bar{L}\left[1 + \delta\left(\Gamma\left(1-\frac{1}{\theta}\right)\frac{W}{w_o} - 1\right)\right]
\end{equation}

Define $s_o^{\delta}$ as the share of efficiency units in occupation $o$ contributed by productivity workers:
\begin{equation*}
    s_o^{\delta} = \frac{\delta \cdot \mathbb{E}[\epsilon_o(i)|o^*(i)=o]}{\delta \cdot \mathbb{E}[\epsilon_o(i)|o^*(i)=o] + (1-\delta)} = \frac{\delta \cdot \Gamma(1-\frac{1}{\theta})\frac{W}{w_o}}{\delta \cdot \Gamma(1-\frac{1}{\theta})\frac{W}{w_o} + (1-\delta)}
\end{equation*}

Note that for $\delta$ workers, higher wages in occupation $o$ attract workers with lower average productivity due to selection effects, partially offsetting employment increases.

\paragraph{Labor Supply Elasticities}
The effective elasticity of labor supply combines both worker types (see Appendix \ref{app:eff_elasticity_derivation} for detailed derivation):
\begin{equation}\label{app:eq:eff_elasticity}
    \Theta_{oo'}^{\text{eff}} = \begin{cases}
        \theta\left[\frac{x_oF_{oo}}{F_o}\right] + (\theta-s_o^{\delta})(1-\pi_o) & \text{if } o = o' \\
        \theta\left[\frac{x_{o'}F_{oo'}}{F_o}\right] - (\theta-s_{o'}^{\delta})\pi_{o'} & \text{if } o \neq o'
    \end{cases}
\end{equation}

This decomposition reveals two key components:
\begin{itemize}
    \item \textbf{Correlation term} $\theta\cdot\frac{x_{o'}F_{oo'}}{F_o}$: Captures skill-based substitution patterns, unaffected by worker composition
    \item \textbf{Baseline term} with coefficient $(\theta-s_o^{\delta})$: Represents average substitutability, decreasing in the share of efficiency units from productivity workers
\end{itemize}

As $s_o^{\delta}$ increases (either through higher $\delta$ or lower relative wages in occupation $o$), baseline substitutability falls. The coefficient ranges from $\theta$ when $s_o^{\delta} = 0$ (no productivity workers or very high wages) to $(\theta-1)$ when $s_o^{\delta} \to 1$ (dominated by productivity workers). Since $\theta > 1$ in our estimates, the elasticity remains positive but diminished. This implies that when more efficiency units come from workers contributing productivity to production, the labor market becomes less responsive to wage changes, with adjustment increasingly dominated by skill-based substitution patterns rather than average mobility.

\paragraph{Implications for Incidence}
Higher values of $s_o^{\delta}$ have three main effects:
\begin{enumerate}
    \item \textbf{Reduced overall elasticities:} The norm $\|\Theta^{\text{eff}}\|$ decreases, implying less employment adjustment
    \item \textbf{Greater wage pass-through:} With $\Delta^{\text{eff}} = (\mathbf{I} + \Theta^{\text{eff}}/\sigma)^{-1}$, smaller elasticities yield larger diagonal elements
    \item \textbf{Correlation dominance:} Skill-based substitution patterns become the primary adjustment mechanism
\end{enumerate}

Our baseline specification ($\delta = 0$, hence $s_o^{\delta} = 0$ for all $o$) assumes all workers experience productivity through preferences, maximizing labor supply elasticities. Any positive $\delta$ would reduce these elasticities through $s_o^{\delta} > 0$, implying our estimates provide lower bounds on wage inequality from technological clustering.

\subsection{Technological Change and Group-Specific Labor Market Frictions} \label{app_ss:distortions}

Following \cite{Burstein2019-jl}, \cite{Hsieh2019-np}, and \cite{Hurst2024-jw}, we incorporate pre-existing labor market discrimination to examine how technological change interacts with group-specific barriers to generate heterogeneous distributional effects.

\paragraph{Setup}

We normalize labor productivity across groups: $A_{o,t}^g = A_{o,t}$ for all $g$.\footnote{This normalization follows standard practice in the discrimination literature. The key assumption is that relative labor productivity across groups is constant, implying changes in occupational distributions and wages relative to white men are driven by changes in discrimination or preferences rather than productivity differences.}

Worker $i$ from group $g$ receives utility from occupation $o$ at time $t$:
\begin{equation*}
    u_{o,t}^g(i) = \ln[(1 - \tau_{o,t}^g)w_{o,t}] + \ln z_{o,t}^g + \ln\epsilon_o(i)
\end{equation*}
where $\tau_{o,t}^g \in [0,1)$ represents pecuniary discrimination (wage "tax"), $z_{o,t}^g > 0$ captures non-pecuniary barriers, and $\epsilon_o(i)$ is idiosyncratic productivity. A fraction $\delta$ of workers contribute productivity to production while $(1-\delta)$ affect effort costs.

\paragraph{Equilibrium Employment Shares}

The share of group $g$ workers choosing occupation $o$ is:
\begin{equation*}
    \pi_{o,t}^g = \frac{x_o^g \cdot F_o(\mathbf{x}^g)}{F(\mathbf{x}^g)}
\end{equation*}
where $x_o^g \equiv A_{o,t}[(1-\tau_{o,t}^g)w_{o,t}z_{o,t}^g]^{\theta}$ and $\mathbf{x}^g = (x_1^g, \ldots, x_O^g)$.

Define the correlation-adjusted share as $\tilde{\pi}_o^g \equiv x_o^g/F(\mathbf{x}^g)$, which satisfies:
\begin{equation*}
    \pi_o^g = \tilde{\pi}_o^g \cdot F_o(\tilde{\boldsymbol{\pi}}^g)
\end{equation*}
This establishes a one-to-one mapping between observed shares $\{\pi_o^g\}$ and adjusted shares $\{\tilde{\pi}_o^g\}$.

\paragraph{Welfare Index and Average Productivity}

The group welfare index is:
\begin{equation*}
    W_t^g \equiv F(\mathbf{x}^g)^{1/\theta}
\end{equation*}

Average productivity for workers selecting occupation $o$ in group $g$:
\begin{equation*}
    \mathbb{E}[\epsilon_o(i)|o^*(i)=o, g] = \Gamma\left(1-\frac{1}{\theta}\right) \cdot A_{o,t}^{1/\theta} \cdot (\tilde{\pi}_o^g)^{-1/\theta}
\end{equation*}

\paragraph{Average Wages: Geometric Mean}

Following \cite{Hsieh2019-np}, we define the geometric average wage for group $g$ workers in occupation $o$:
\begin{equation*}
    \overline{\text{wage}}_{o,t}^g \equiv (1-\tau_{o,t}^g)w_{o,t} \cdot e^{\mathbb{E}[\ln\epsilon_o(i)|o^*(i)=o,g]}
\end{equation*}

With fraction $\delta$ of workers contributing productivity to production and $(1-\delta)$ affecting effort costs, the weighted geometric average becomes:
\begin{equation*}
    \overline{\text{wage}}_{o,t}^g = (1-\tau_{o,t}^g)w_{o,t} \cdot e^{\delta \cdot \mathbb{E}[\ln\epsilon_o(i)|o^*(i)=o,g]}
\end{equation*}

From Fréchet distribution properties:
\begin{equation*}
    \mathbb{E}[\ln\epsilon_o(i)|o^*(i)=o,g] = -\frac{\gamma}{\theta} + \ln\mathbb{E}[\epsilon_o(i)|o^*(i)=o,g]
\end{equation*}
where $\gamma$ is the Euler-Mascheroni constant.

Since $\mathbb{E}[\epsilon_o(i)|o^*(i)=o,g] = \Gamma(1-1/\theta) \cdot A_{o,t}^{1/\theta} \cdot (\tilde{\pi}_o^g)^{-1/\theta}$:
\begin{equation*}
    \overline{\text{wage}}_{o,t}^g = (1-\tau_{o,t}^g)w_{o,t} \cdot \left[\Gamma(1-1/\theta) \cdot A_{o,t}^{1/\theta} \cdot (\tilde{\pi}_o^g)^{-1/\theta} \cdot e^{-\gamma/\theta}\right]^{\delta}
\end{equation*}

The crucial advantage of the geometric average is its multiplicative separability: pecuniary discrimination $(1-\tau_{o,t}^g)$ enters only through the base wage term, while non-pecuniary discrimination $z_{o,t}^g$ affects only productivity selection through $\tilde{\pi}_o^g$.

\paragraph{Identification of Discrimination Changes}

Using geometric average wages, we identify relative discrimination changes across occupations. Since welfare indices are unobservable, we can only identify occupation-specific discrimination relative to an economy-wide average.

From geometric average wages:
\begin{equation*}
    \Delta\ln\overline{\text{wage}}_{o,t}^g - \Delta\ln\overline{\text{wage}}_{o,t}^w = \Delta\ln\left[\frac{1-\tau_{o,t}^g}{1-\tau_{o,t}^w}\right] - \frac{\delta}{\theta}\left[(\Delta\ln\tilde{\pi}_{o,t}^g - \Delta\ln\tilde{\pi}_{o,t}^w)\right]
\end{equation*}

From adjusted employment shares:
\begin{equation*}
    \Delta\ln\tilde{\pi}_{o,t}^g - \Delta\ln\tilde{\pi}_{o,t}^w = \theta \cdot \Delta\ln\left[\frac{(1-\tau_{o,t}^g)z_{o,t}^g}{(1-\tau_{o,t}^w)z_{o,t}^w}\right]
\end{equation*}

Combining these equations yields the composite discrimination effect:
\begin{equation*}
    \Delta\ln\left[\frac{(1-\tau_{o,t}^g)z_{o,t}^g}{(1-\tau_{o,t}^w)z_{o,t}^w}\right] = \frac{1}{\theta}(\Delta\ln\tilde{\pi}_{o,t}^g - \Delta\ln\tilde{\pi}_{o,t}^w)
\end{equation*}

Pecuniary discrimination:
\begin{equation*}
    \Delta\ln\left[\frac{1-\tau_{o,t}^g}{1-\tau_{o,t}^w}\right] = (\Delta\ln\overline{\text{wage}}_{o,t}^g - \Delta\ln\overline{\text{wage}}_{o,t}^w) + \frac{\delta}{\theta}(\Delta\ln\tilde{\pi}_{o,t}^g - \Delta\ln\tilde{\pi}_{o,t}^w)
\end{equation*}

Non-pecuniary discrimination:
\begin{equation*}
    \Delta\ln\left[\frac{z_{o,t}^g}{z_{o,t}^w}\right] = \frac{1-\delta}{\theta}(\Delta\ln\tilde{\pi}_{o,t}^g - \Delta\ln\tilde{\pi}_{o,t}^w) - (\Delta\ln\overline{\text{wage}}_{o,t}^g - \Delta\ln\overline{\text{wage}}_{o,t}^w)
\end{equation*}

\textbf{Special cases:}
\begin{itemize}
    \item When $\delta = 0$: Pecuniary discrimination is identified directly from wage gaps, and non-pecuniary from employment shifts
    \item When $\delta = 1$: Non-pecuniary discrimination vanishes from the wage equation, allowing clean separation
    \item With $\delta = 0.5$ (our estimate from Section \ref{ss:efficiency}): Both discrimination types contribute equally to employment gaps while pecuniary discrimination has larger wage effects
\end{itemize}

\paragraph{Counterfactual Analysis of Between-Group Inequality}

This framework enables analysis of how technological change affects between-group inequality through three counterfactual scenarios:

\textbf{Counterfactual 1: No technological change.} Holding technology constant ($\Delta\ln A_o = \Delta\ln w^{\text{Automation}}_o = 0$) while allowing discrimination to evolve reveals baseline convergence trends. This isolates whether groups were converging absent technological disruption, crucial for understanding whether technology accelerates or reverses existing trends.

\textbf{Counterfactual 2: No pecuniary discrimination changes.} Fixing wage penalties ($\Delta\ln(1-\tau_o^g) = 0$) while allowing technology and non-pecuniary barriers to evolve shows how pecuniary discrimination mediates technological impact. Comparing to the full model reveals whether wage discrimination changes amplify or dampen technological disruption.

\textbf{Counterfactual 3: No non-pecuniary discrimination changes.} Holding amenity and cultural barriers constant ($\Delta\ln z_o^g = 0$) isolates how non-pecuniary factors shape adjustment to technological shocks. The gap between this and Counterfactual 2 identifies how non-pecuniary barriers prevent optimal reallocation following technological change.

\paragraph{Welfare and Wage Implications}

These counterfactuals reveal distinct effects on group welfare and average wages. Group $g$'s welfare relative to white men is:
\begin{equation*}
    \frac{W^g}{W^w} = \left[\frac{F(\mathbf{x}^g)}{F(\mathbf{x}^w)}\right]^{1/\theta}
\end{equation*}
where $\mathbf{x}^g$ embeds technology, pecuniary, and non-pecuniary factors through $x_o^g = A_{o,t}[(1-\tau_{o,t}^g)w_{o,t}z_{o,t}^g]^{\theta}$.

The relative average wage between groups is:
\begin{equation*}
    \frac{\overline{\text{wage}}^g}{\overline{\text{wage}}^w} = \frac{\sum_o \pi_o^g \cdot \overline{\text{wage}}_o^g}{\sum_o \pi_o^w \cdot \overline{\text{wage}}_o^w}
\end{equation*}
where $\overline{\text{wage}}_o^g$ incorporates both direct wage effects and selection through the geometric average.

By comparing these objects across our three counterfactual scenarios, we can decompose between-group inequality changes into different sources: the direct effect of technological change through occupational wages, the role of changing pecuniary discrimination, the contribution of non-pecuniary barriers, and their interactions through the substitution structure. This decomposition reveals how technological clustering interacts with pre-existing inequalities.
\newpage

\section{Additional Empirical Results}\label{b:apppendix}

\setcounter{theorem}{0}
\setcounter{proposition}{0} 
\setcounter{lemma}{0}
\setcounter{corollary}{0}
\setcounter{definition}{0}
\setcounter{assumption}{0}
\setcounter{remark}{0}
\setcounter{table}{0}
\setcounter{figure}{0}
\setcounter{equation}{0} 
%
\renewcommand{\thetheorem}{B\arabic{theorem}}
\renewcommand{\theproposition}{B\arabic{proposition}}
\renewcommand{\thelemma}{B\arabic{lemma}}
\renewcommand{\thecorollary}{B\arabic{corollary}}
\renewcommand{\thedefinition}{B\arabic{definition}}
\renewcommand{\theassumption}{B\arabic{assumption}}
\renewcommand{\theremark}{B\arabic{remark}}
\renewcommand{\thetable}{B\arabic{table}}
\renewcommand{\thefigure}{B\arabic{figure}}
\renewcommand{\theequation}{B\arabic{equation}}


\subsection{Occupation Classification}\label{b:appendix:occ_xwalk}

Our analyses require constructing occupation-level panels for the period 1980–2018. To this end, following \cite{Autor2024-ay}, we use a consistent occupation coding scheme (occ1990dd), originally developed by Dorn (2009) and updated through 2018, which yields a balanced panel of 306 consistent, 3-digit occupations. This detailed classification preserves crucial occupational variation and accurately captures the structure of the labor market over the period.

\subsection{Occupational Skill Intensities}\label{b:appendix:skills}

This appendix details the construction of occupational skill intensities $\{\omega_o^s\}$ from O*NET data, following the methodology of \cite{Lise2020-hm}.

\paragraph{Data Source and Theoretical Correspondence} 

O*NET version 28.2 provides comprehensive occupational information for 873 occupations. The database contains 277 descriptors organized into nine categories, with ratings derived from two sources: (i) worker surveys for occupation-specific assessments, and (ii) occupational analyst surveys for standardized evaluations. We retain 218 descriptors from five categories—skills, abilities, knowledge, work activities, and work context—as these directly correspond to the theoretical concept of skill intensities. The remaining categories (job interests, work values, work styles, and experience/education requirements) are excluded as they reflect preferences or credentials rather than skill utilization in production.

The O*NET skill descriptors have cardinal meaning: evaluators assess both the importance and level of each skill intensity for performing an occupation's tasks on quantitative scales. These cardinal measures capture how intensively occupations utilize different skills in production, directly corresponding to our theoretical object $\omega_o^s = (A_o^s)^\theta/A_o$—the share of occupational productivity attributable to each skill dimension. We treat O*NET measures as direct observations of this skill productivity structure under the maintained assumption that O*NET's measurement protocol reflects the same productive skill intensities embedded in our model. 

\paragraph{Dimension Reduction} 

Following \cite{Lise2020-hm}, we apply Principal Component Analysis (PCA) with exclusion restrictions to extract three interpretable skill dimensions:
\begin{enumerate}
\item \textbf{Cognitive skills}: Identified through the mathematics knowledge descriptor
\item \textbf{Manual skills}: Identified through the mechanical knowledge descriptor  
\item \textbf{Interpersonal skills}: Identified through the social perceptiveness descriptor
\end{enumerate}

These exclusion restrictions ensure that each principal component has a clear economic interpretation while maintaining orthogonality—a property that aligns with the model's assumption of independent skill-specific productivity draws. The first three components explain 58\% of total variation (cognitive: 35.6\%, manual: 15.2\%, interpersonal: 6.9\%), with the dominance of cognitive and manual dimensions reflecting their primary role in occupational differentiation and, consequently, in determining substitution patterns.

\paragraph{Construction of Skill Intensities} 

The raw principal component loadings contain negative values, violating the theoretical requirement that $\omega_o^s \geq 0$. We address this through a two-step procedure:

\textbf{Step 1: Rescaling.} Apply linear transformations to map each occupation's loading on principal component $s$ to the unit interval:
\begin{equation*}
r_o^s = \frac{PC_o^s - \min_{o'} PC_{o'}^s}{\max_{o'} PC_{o'}^s - \min_{o'} PC_{o'}^s}
\end{equation*}

Linear transformations are crucial as they preserve relative distances between occupations in each skill dimension—a key feature for distance-dependent elasticity of substitution. Converting to ranks would impose uniform spacing between adjacent occupations, eliminating meaningful variation in skill proximity that drives substitution patterns.

\textbf{Step 2: Variance weighting.} Convert rescaled loadings to variance-weighted shares to obtain the final skill intensities:
\begin{equation*}
\omega_o^s = \frac{r_o^s \times \text{Var}_s}{\sum_{s' \in \mathcal{S}} r_o^{s'} \times \text{Var}_{s'}}
\end{equation*}
where $\text{Var}_s$ is the proportion of variance explained by component $s$. This weighting ensures that skills contributing more to occupational variation receive proportionally higher weight in the correlation function $F$, consistent with their greater role in determining substitution patterns. The formulation guarantees $\sum_s \omega_o^s = 1$ for each occupation, as required by the theoretical restriction that skill intensities sum to unity.

\paragraph{Mapping to Occupation Codes} 

The final step maps O*NET occupation codes to the consistent occ1990dd classification used throughout the analysis, enabling linkage with employment and wage data from Census and CPS. The crosswalk covers 306 three-digit occupations, preserving granular variation while maintaining temporal consistency from 1980-2018.
\subsection{Measures of Automation and AI Exposure}\label{b:appendix:expos}

\paragraph{Existing Measures and Task Evaluation}

The literature has developed several measures of occupational exposure to automation, each capturing different aspects of technological vulnerability. These include occupational routine task intensity \citep{Autor2013-zm}, the decline in labor share due to the adoption of industrial robots, machines, and software \citep{Acemoglu2022-lv}, and occupational exposure to automation patents \citep{Autor2024-ay}. These measures share a common theoretical foundation rooted in Polanyi's Paradox \citep{Autor2015-ff}: jobs codifiable into well-defined rules or algorithms are more susceptible to automation and are typically classified as routine. Consistent with this framework, numerous studies document that occupations with higher automation exposure have experienced slower wage growth over the past four decades.

In contrast to automation, measuring occupational exposure to artificial intelligence presents unique challenges, as its full economic impact remains unrealized. To address this challenge, recent research has leveraged large language models (LLMs) as predictive tools for assessing economic outcomes. \cite{Eloundou2024-cu} pioneered this approach by evaluating occupational exposure to LLMs through a dual methodology: human annotators and GPT-4 classified O*NET tasks using an exposure rubric to determine whether LLMs can perform or assist with specific tasks. Their findings highlight the potential of LLMs as general-purpose technologies. 

Subsequent validation studies strengthen confidence in this approach. \cite{Bick2024-eu} and \cite{Tomlinson2025-va} demonstrate high correlations between LLM task evaluations and ex-post real-world generative AI adoption patterns. Most compelling, \cite{brynjolfsson2025canaries} provide causal evidence that LLM exposure measures predict actual labor market outcomes: using high-frequency administrative payroll data, they document that early-career workers (ages 22-25) in the most AI-exposed occupations have experienced a 13\% relative decline in employment since widespread AI adoption, with effects concentrated in occupations where AI automates rather than augments human labor.

\paragraph{Our Methodology: ChatGPT Task Evaluation}

Building on this validated literature, we adopt a streamlined yet comparable approach leveraging ChatGPT to directly estimate AI and automation exposure. Our methodology employs the O*NET database, which provides detailed descriptions of 19,200 tasks across 862 occupations. Each task undergoes two distinct assessments:

\begin{itemize}
    \item \textbf{AI Exposure}: We query ChatGPT: ``Can generative AI (e.g., large language models like ChatGPT) potentially perform this task without human intervention?'' This assessment captures the extent to which occupations are exposed to AI-driven technologies.
    \item \textbf{Automation Exposure}: We query ChatGPT: ``Can industrial robots, machines, and computers (no AI capability) perform this task without human intervention?'' This distinguishes tasks automatable using conventional, rule-based systems from those requiring advanced AI capabilities.
\end{itemize}

Based on these evaluations, ChatGPT estimates that approximately 6,000 tasks (roughly one-third of the total) can be performed by AI without human intervention—a scale comparable to traditional automation technologies. This classification provides a granular perspective on the differential impacts of AI versus traditional automation across occupations. We then calculate the share of automatable or AI-exposed tasks within each occupation to construct our exposure measures.

\begin{figure}[ht]
    \centering
    \includegraphics[width=0.8\linewidth]{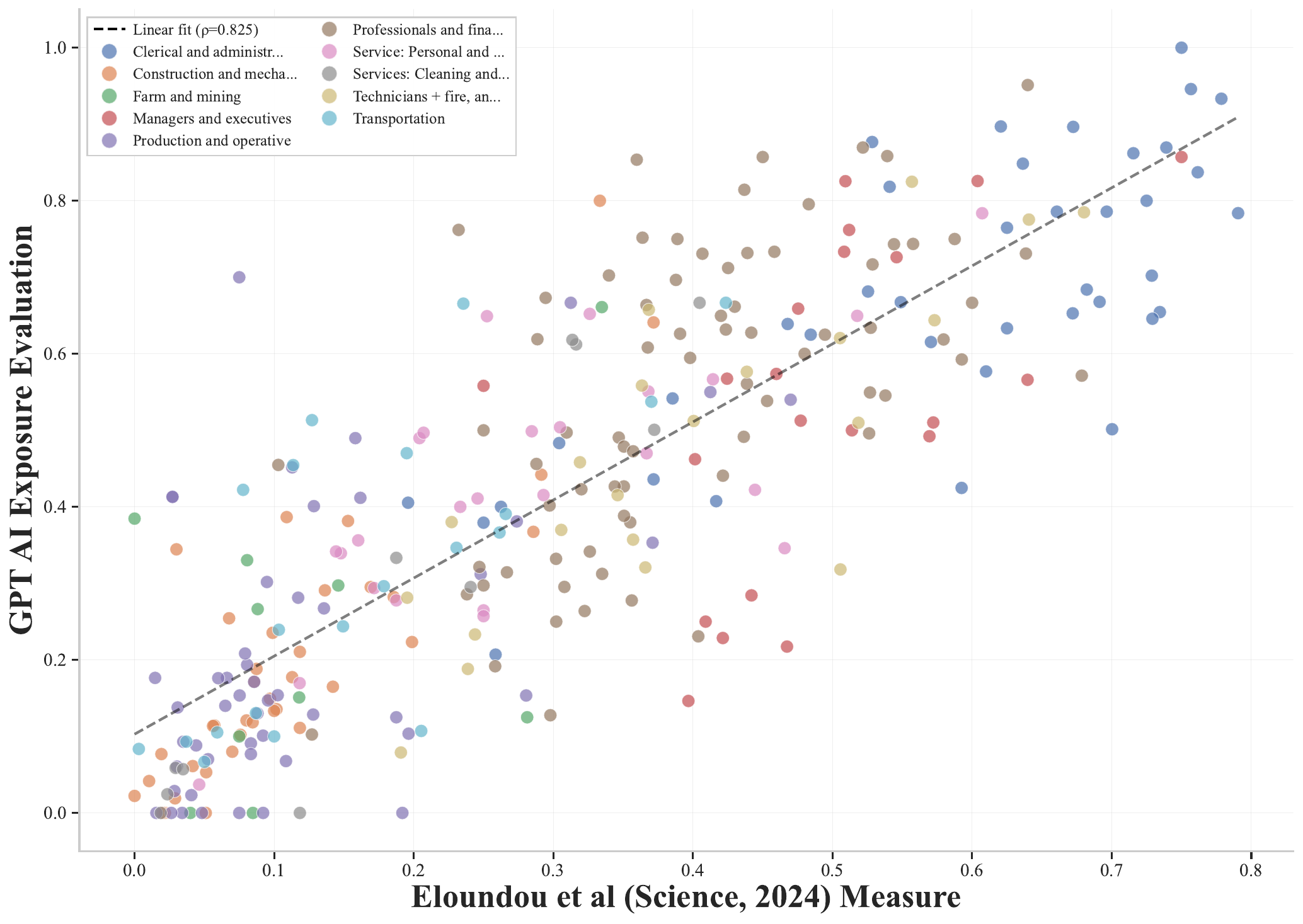}
    \caption{Comparison of AI Exposure Measures}
    \note{\textit{Notes:}This figure compares our ChatGPT AI evaluation scores (y-axis) with the science-based measure from \cite{Eloundou2024-cu} (x-axis). Points represent occupations colored by broad occupational categories. The dashed line shows the linear fit with correlation coefficient $\rho = 0.825$. The strong positive relationship validates our ChatGPT evaluation methodology against established measures in the literature.}
    \label{f:ai_science_comparison}
\end{figure}

\paragraph{Validation Against Existing Measures}

To validate our approach, we compare our ChatGPT-based measures with established metrics in the literature. Figure \ref{f:ai_science_comparison} demonstrates that our occupational exposure to generative AI correlates strongly with \cite{Eloundou2024-cu}'s measure, yielding a correlation coefficient of 0.825. This high correlation validates our streamlined methodology while confirming the robustness of LLM-based evaluation approaches.

\begin{figure}[ht]
    \centering
    \subcaptionbox{Automation Exposure and Wage, AR 2022}{\includegraphics[scale=0.39]{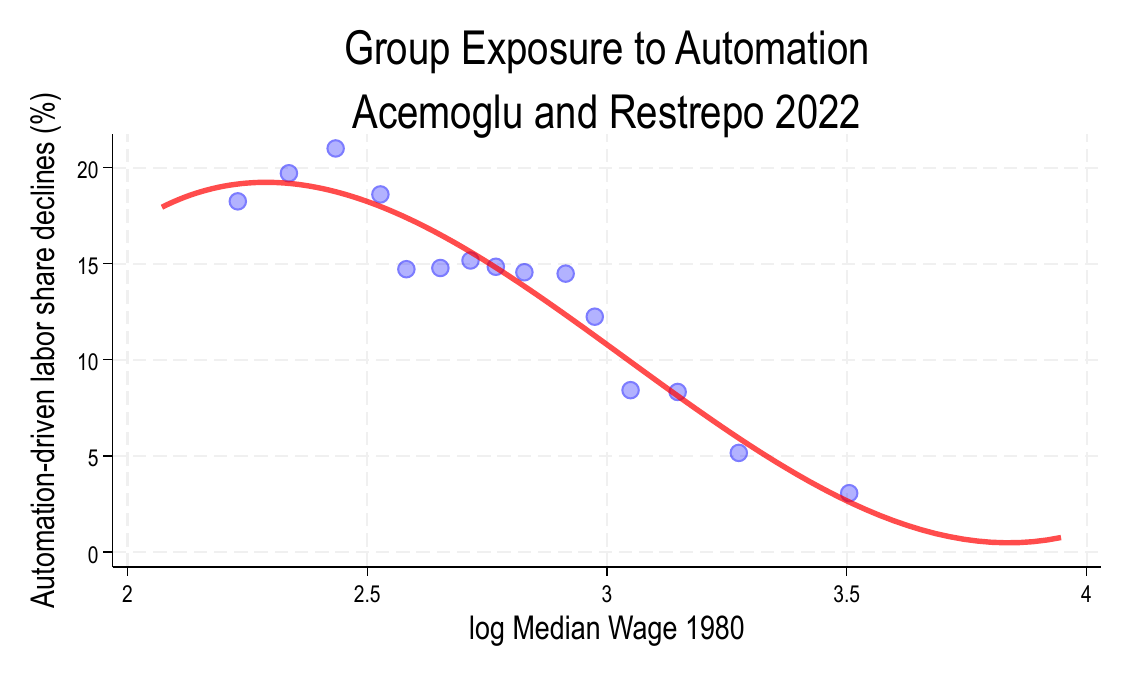}}\hfill
    \subcaptionbox{Automation Exposure and Wage, ChatGPT}{\includegraphics[scale=0.39]{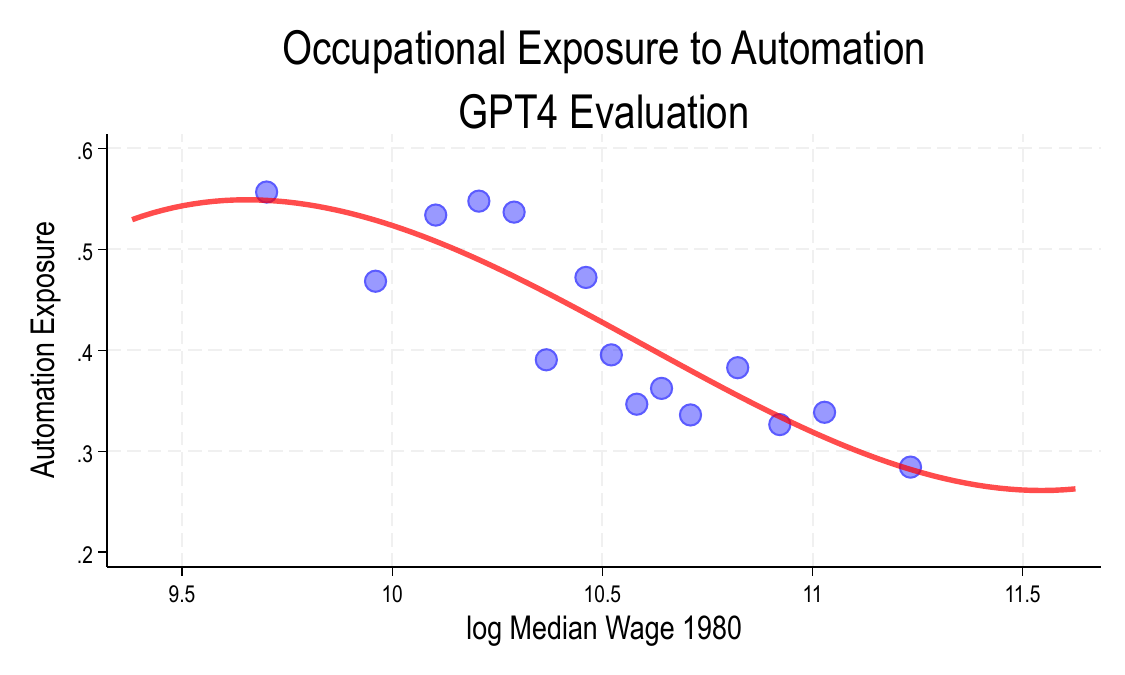}}\\
    \subcaptionbox{Routine Task Intensity}{\includegraphics[scale=0.39]{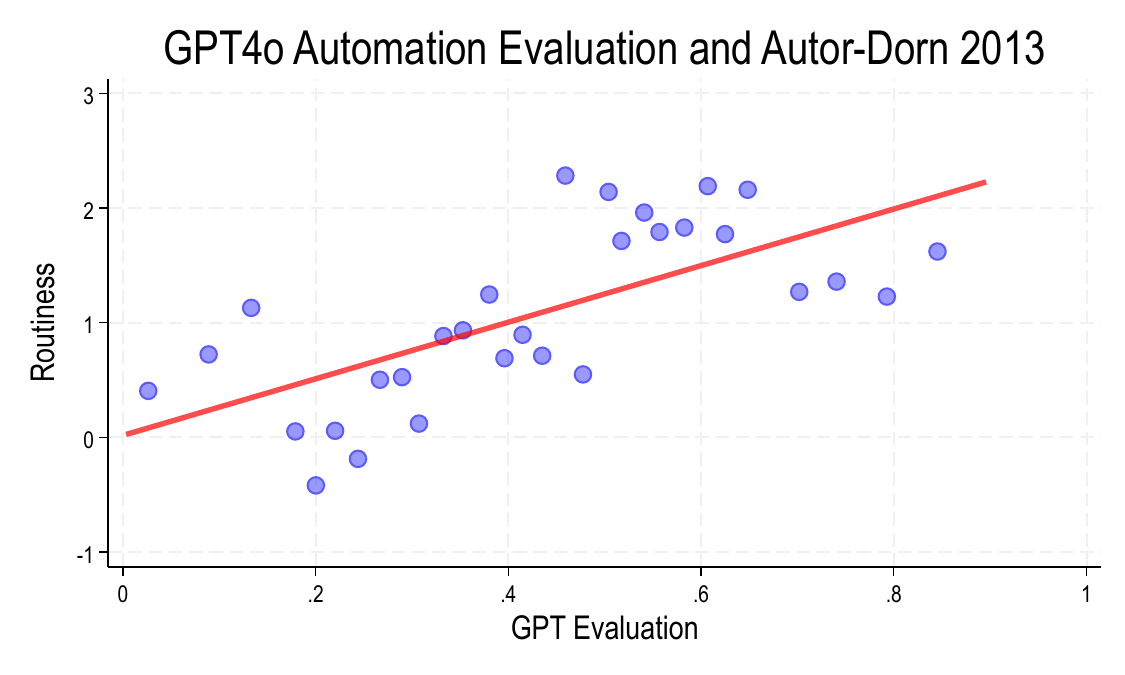}}\hfill
    \subcaptionbox{Automation Patent Exposure}{\includegraphics[scale=0.39]{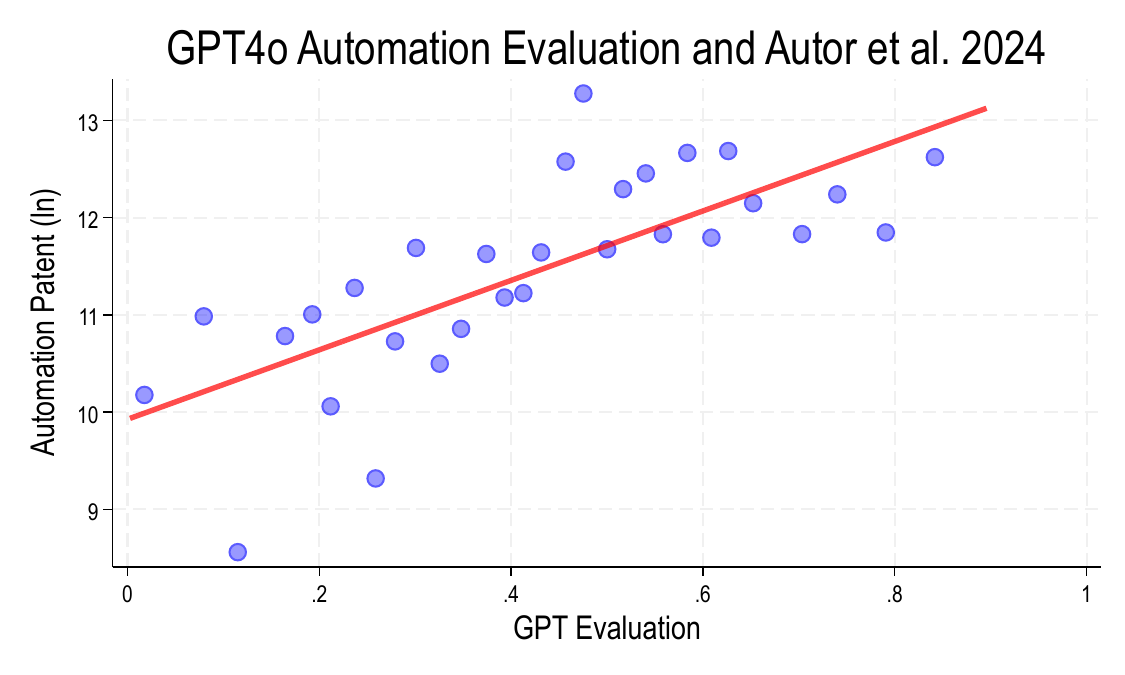}}
    \caption{Validation of Automation Exposure Measures}
    \note{\textit{Notes:}This figure compares automation exposure as evaluated by ChatGPT with existing measures using binscatter plots. Panel (a) shows the decline in labor share due to automation from \cite{Acemoglu2022-lv}, Panel (b) presents our ChatGPT estimates, Panel (c) compares with routine task intensity from \cite{Autor2013-zm}, and Panel (d) validates against automation patent exposure from \cite{Autor2024-ay}. The consistent patterns across all measures validate our ChatGPT evaluation approach.}
    \label{f:auto_expos}
\end{figure}

Figure \ref{f:auto_expos} further validates our automation exposure measure by comparing ChatGPT's estimates with existing metrics. Panels (a) and (b) compare our estimates with the automation exposure measure from \cite{Acemoglu2022-lv}. Since their measure operates at the demographic-age-education group level rather than the occupational level, we plot exposure against log median wage in 1980. The striking similarity of the two distributions across income levels confirms the validity of our approach. Panel (c) demonstrates a strong correlation between our measure and occupational routine task intensity, while Panel (d) reveals consistent patterns with exposure to automation patents from \cite{Autor2024-ay}.

These validation exercises demonstrate that our ChatGPT-based methodology produces measures highly consistent with established approaches while offering the advantage of direct, task-level evaluation for both automation and AI exposure. This validation is crucial for our subsequent analysis of how technological clustering in skill space shapes labor market incidence.
\subsection{Technological Exposure across Inter-personal Dimension} \label{b:appendix:expos_int}
Figure \ref{f:expos_int} illustrates how occupational exposure to automation and AI varies with interpersonal skill intensities. Panel (a) shows that occupations requiring greater interpersonal skills tend to be less exposed to automation, aligning with the intuition that social and emotional intelligence—often critical in managerial, negotiation, and caregiving roles—are difficult to codify into rule-based processes. In contrast, Panel (b) reveals that occupations with higher interpersonal skill intensities tend to be more exposed to AI, though with greater variance. This noisier relationship suggests that while AI can assist or complement interpersonal tasks (e.g., customer support or education), full automation remains limited by the complexity of human interaction.

These findings reinforce the distinct nature of AI and automation risks: whereas automation displaces predictable, rule-based tasks, AI is more likely to augment or replace cognitive tasks, including those requiring some degree of human interaction. However, interpersonal-intensive occupations—such as psychologists, teachers, and business executives—still rely on empathy, persuasion, and social nuance, which remain challenging for AI to fully replicate.

\begin{figure}[ht]
    \centering
    \subcaptionbox{Interpersonal vs. Automation Exposure}{\includegraphics[scale=0.39]{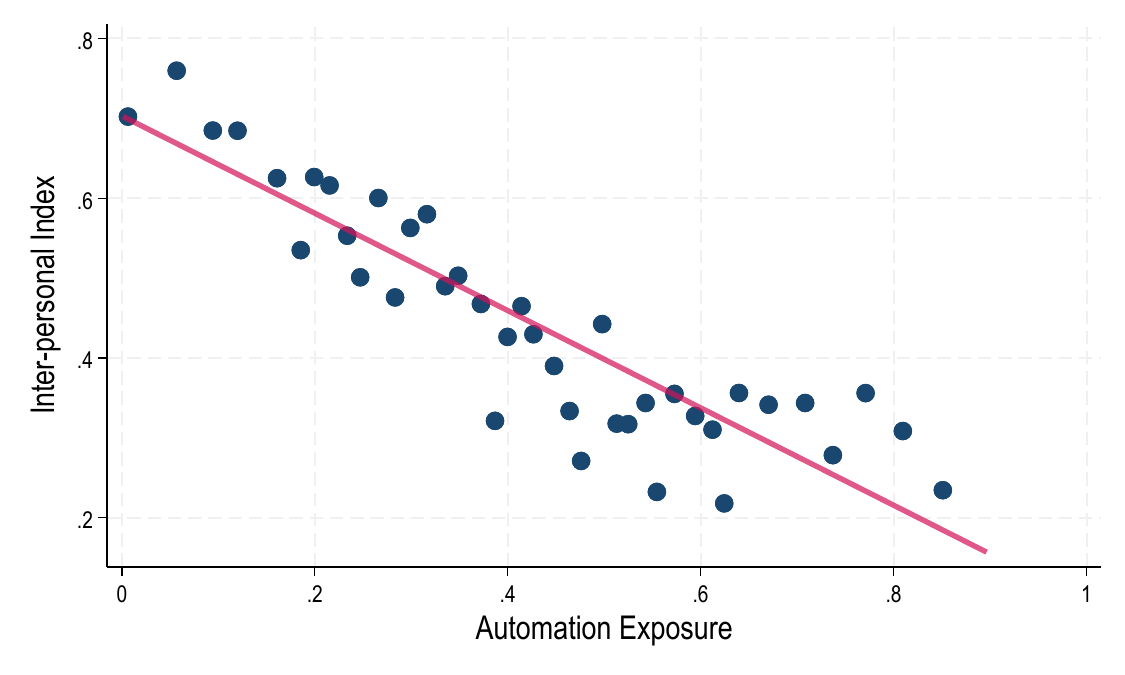}}\hfill
    \subcaptionbox{Interpersonal vs. AI Exposure}{\includegraphics[scale=0.39]{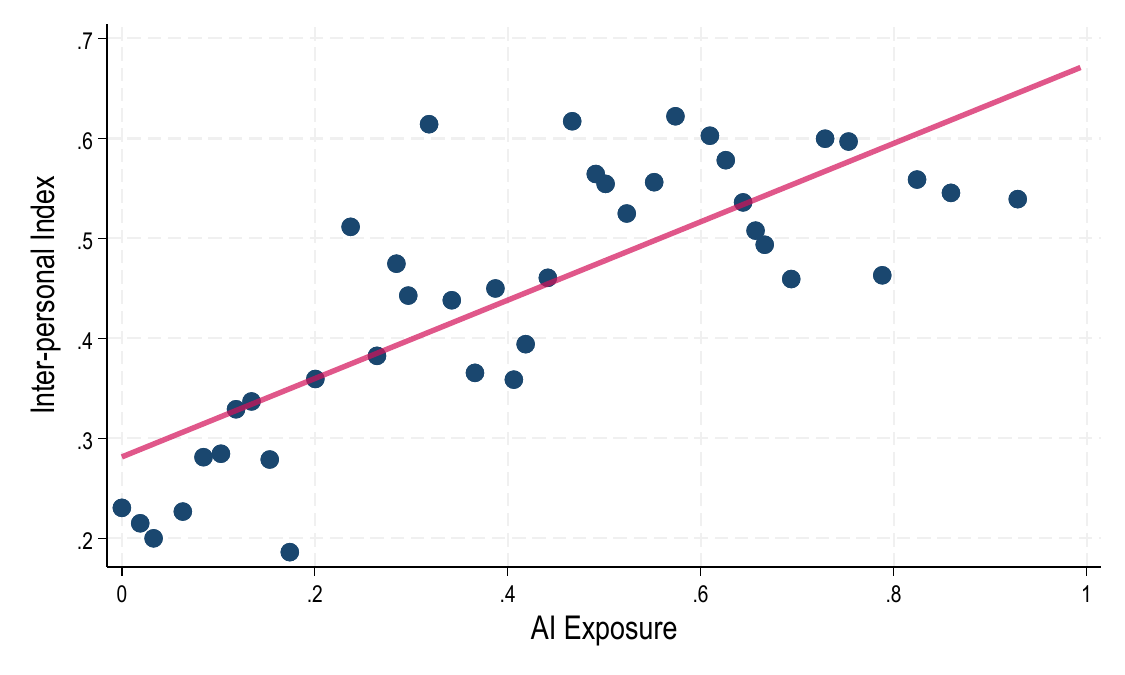}}
    \caption{Technological Exposures across Interpersonal Skills}
    \note{\textit{Notes:}This figure illustrates the relationship between occupational interpersonal skill intensities and exposure to automation (Panel (a)) and AI (Panel (b)).}
    \label{f:expos_int}
\end{figure}
\subsection{Spatial Visualization of Technological Clustering} \label{app:spatial_visualization}

Figure~\ref{f:spatial_structure} visualizes the spatial mechanism through which technological clustering constrains labor market adjustment. Panel A maps the eigenshock with the smallest eigenvalue across cognitive-manual skill space. The pattern bifurcates: high-manual/low-cognitive occupations (lower right) and high-cognitive/low-manual occupations (upper left) load strongly but oppositely. Panel B reveals that AI exposure concentrates precisely in the high-cognitive region identified by the eigenshock.

\begin{figure}[ht]
    \centering
    \subcaptionbox{Eigenshock with Smallest Eigenvalue}{\includegraphics[scale = 0.25]{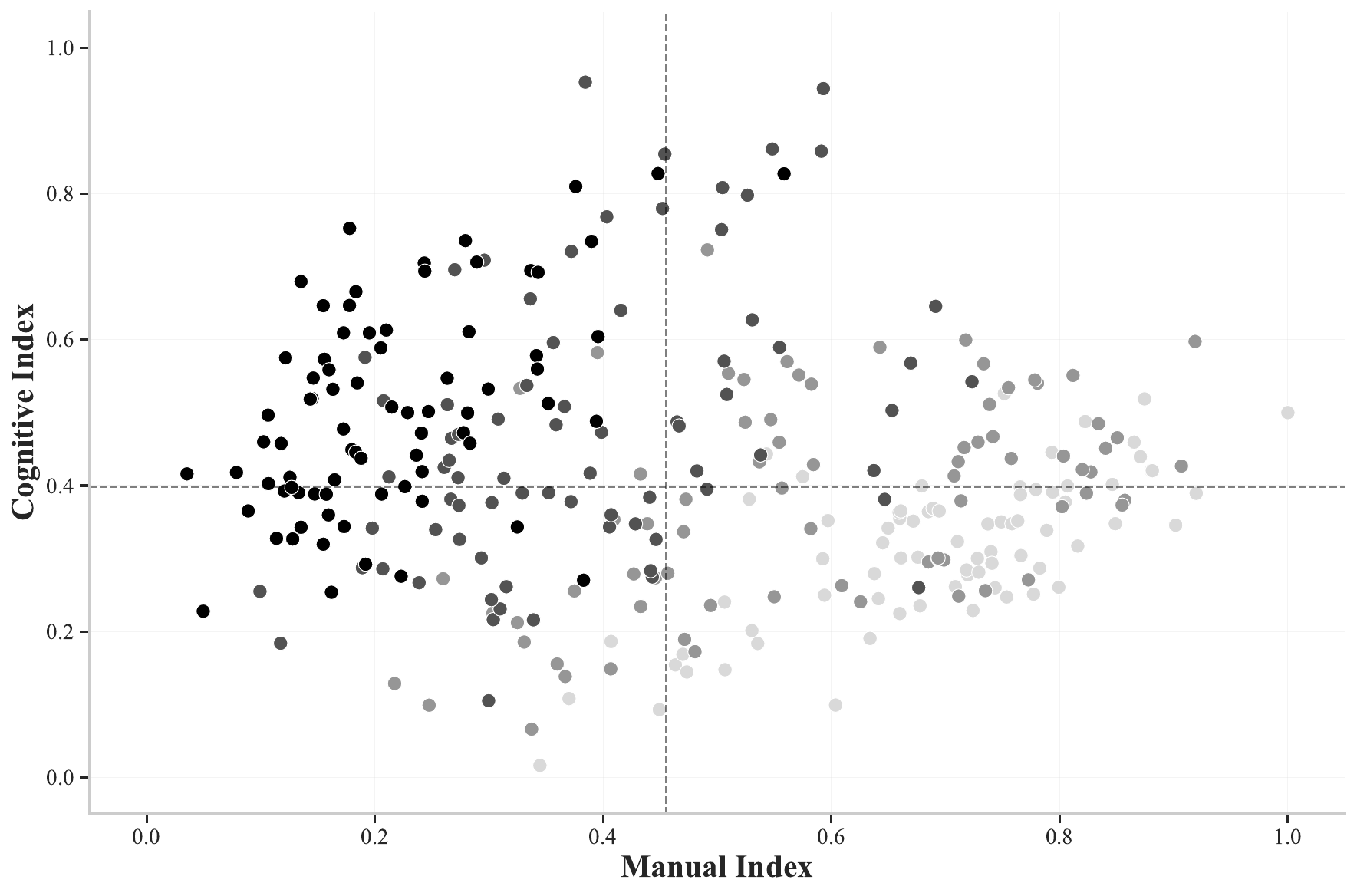}}\hfill
    \subcaptionbox{AI Exposure in Skill Space}{\includegraphics[scale = 0.25]{figures/skill_space_gpt4_ai_eval.pdf}}
    \caption{Spatial Structure of Technological Constraints}
    \label{f:spatial_structure}
    \note{\textit{Notes:}Panel A: Eigenshock with smallest eigenvalue (2018) in cognitive-manual skill space. Panel B: AI exposure distribution. Darker shading indicates higher loading/exposure. Dashed lines mark median skill intensities. The correspondence shows AI aligns with the shock pattern least absorbable through reallocation.}
\end{figure}

The grayscale gradient exposes the mobility trap: occupations with highest eigenshock loading (darkest quartile) form tight clusters. A financial analyst facing AI exposure cannot escape to data analysis or market research—these skill-similar alternatives face comparable threats. The correlation enabling natural transitions becomes the mechanism preventing escape.\footnote{Dashed median lines divide the space into quadrants. Occupations in the upper-left quadrant face double jeopardy: direct AI exposure plus being surrounded by similarly threatened occupations.} This spatial concentration in skill space—visible in the alignment between Panels A and B—explains why technological shocks generate more severe wage effects than alternative demand shifts that disperse across skill dimensions.

\subsection{Alternative Demand Shocks: Trade and Demographics} \label{app:alternative_shocks}

To contextualize the distributional impact of automation and AI, we compare their spectral properties with two alternative sources of occupational demand shifts obtained from \cite{Autor2024-ay}: Chinese import competition (the ``China shock'') and demographic changes from population aging. These shocks provide counterfactual benchmarks for understanding how technological clustering differs from other labor market disruptions.

\paragraph{China Shock}

The China shock represents occupation-level demand changes induced by increased Chinese import competition in U.S. manufacturing during 1991–2014. Following \cite{Autor2013-et}, industry $i$'s exposure is measured as the change in imports from China to other developed countries ($\Delta M_{i,t}^{OC}$) scaled by the industry's initial U.S. market size (domestic output plus imports minus exports in 1988). An occupation's exposure is then constructed as:
\begin{equation*}
    \text{ChinaExposure}_{j,t} = 100 \times \sum_i \frac{E_{ij,t-10}}{E_{j,t-10}} \times \frac{\Delta M_{i,t}^{OC}}{Y_{i,88} + M_{i,88} - X_{i,88}}
\end{equation*}
where $E_{ij,t-10}/E_{j,t-10}$ is occupation $j$'s employment share in industry $i$ at the beginning of the period. This measure captures how occupation $j$'s employment distribution across industries exposes it to differential trade shocks. By construction, non-manufacturing occupations have zero China exposure, while manufacturing occupations vary based on their specific industry composition and those industries' import exposure.

The China shock exhibits fundamentally different spatial properties than technological shocks. While it concentrates in manufacturing, affected occupations span diverse skill intensities—from production workers to engineers to managers—generating dispersion across the skill space rather than clustering within it. This dispersion creates adjustment pathways: displaced manufacturing workers can transition to service occupations with similar skill intensities but different industry exposure.

\paragraph{Demographic Changes}

Demographic demand shifts capture how population aging alters consumption patterns across industries, following \cite{DellaVigna2007-ff}. The Baby Boom generation's progression through the age distribution—from prime working age (1980–2000) to middle and late adulthood (2000–2018)—systematically shifted consumption toward healthcare, leisure, and age-related services while reducing demand for child-related products and durables.

Occupation-level demographic exposure is constructed by combining age-specific consumption profiles from the Bureau of Labor Statistics' Consumer Expenditure Survey with Census population data. For each product category $k$, age-consumption profiles are estimated and consumption changes are predicted based on demographic shifts, then crosswalked to consistent industries. An occupation's demographic exposure is calculated as:
\begin{equation*}
    \text{DemographicExposure}_{j,t} = 100 \times \sum_i \frac{E_{ij,t-1}}{E_{j,t-1}} \times \widetilde{\Delta \ln \text{demand}}_{i,t}
\end{equation*}
where $\widetilde{\Delta \ln \text{demand}}_{i,t}$ is the predicted log change in demand for industry $i$'s output due to demographic composition changes.

Demographic shocks generate winners and losers across the occupational distribution—healthcare occupations expand while those in youth-oriented industries contract—but affected occupations are dispersed throughout skill space. A childcare worker can transition to eldercare; a toy designer to medical device development. This dispersion across skill dimensions, rather than concentration within them, fundamentally distinguishes demographic shocks from technological clustering.

\subsection{Wage and Employment Effects of Automation} \label{b:appendix:wage_emp_auto}

This section provides additional details on the wage and employment effects of automation exposure. The Panel Study of Income Dynamics (PSID) is a widely used longitudinal dataset that has tracked nearly 9,200 U.S. families since 1968. We leverage its panel structure to estimate relative wage trends by occupation while controlling for selection effects.

Since the main specifications have already been discussed, we now present additional results in Figure~\ref{f:psid_additional}, which examines wage effects by gender and under different control specifications. Panel A reports the wage effects of automation separately for men and women, showing that the results are nearly identical, with no statistically significant differences. Panel B introduces additional controls, with the blue line accounting for age and $\text{age}^2$ and allows for changing return to education for the green line, while the green line further allows for a changing return to education. The results suggest that changes in the return to education explain about a quarter of the wage effects attributed to automation.

However, when estimating elasticities, we prefer the main specification without controlling for changes in the return to education. From a long-run perspective, new workers may adjust their educational and occupational choices in response to shifts in the skill premium.
\begin{figure}[ht]
    \centering
    \subcaptionbox{Wage Effects by Gender}{\includegraphics[scale=0.24]{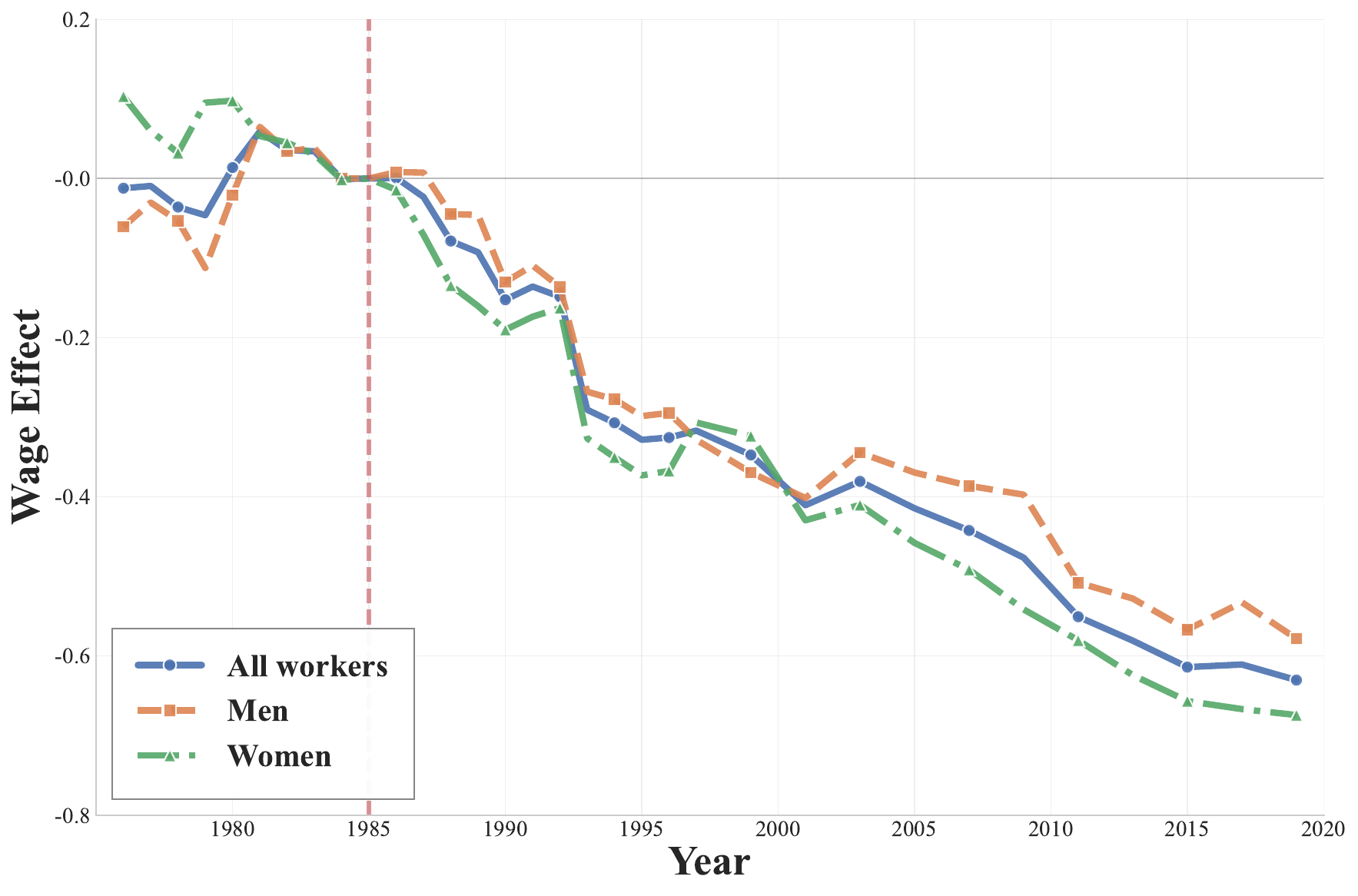}}\hfill
    \subcaptionbox{Additional Controls}{\includegraphics[scale=0.24]{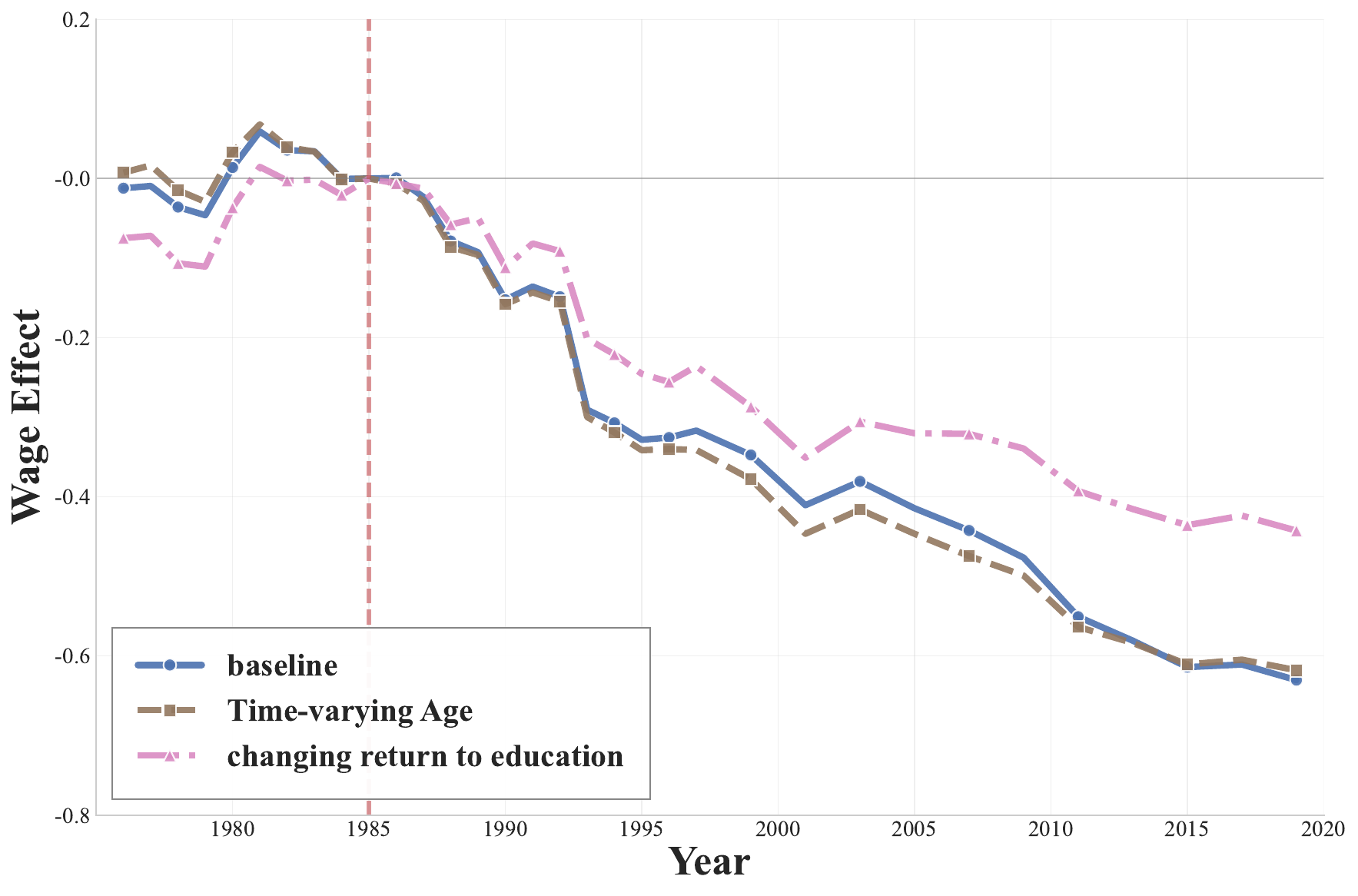}}
    \caption{Effects of Automation on Wages}
    \note{\textit{Notes:} Panel A presents the wage effects of automation separately for men and women, showing no statistically significant differences. Panel B introduces additional controls, where the blue line includes age and \( \text{age}^2 \), and the green line further accounts for a changing return to education. The latter explains approximately 25\% of the wage effects attributed to automation.}
    \label{f:psid_additional}
\end{figure}

We now present additional results on the heterogeneous employment effects of automation across demographic groups, which are used to estimate correlation structures. Panel A of Figure~\ref{f:hetero_emp_demo} displays the average change in log employment shares between 1980 and 2010 by gender for white workers, while Panel B presents the corresponding employment effects for Black workers. The results indicate that white men are the least responsive to automation. Based on the data, this group was predominantly employed in occupations requiring more manual skills, which, as shown in our estimation results, are less portable across occupations. This pattern is reflected in our estimation procedure, which captures the variation in occupational transitions. As a result, the estimated correlation parameter for manual skills, $\rho_{\text{Man}}$, is relatively small, indicating lower substitutability of manual-intensive jobs.

\begin{figure}[ht]
    \centering
    \subcaptionbox{White Workers}{\includegraphics[scale=0.25]{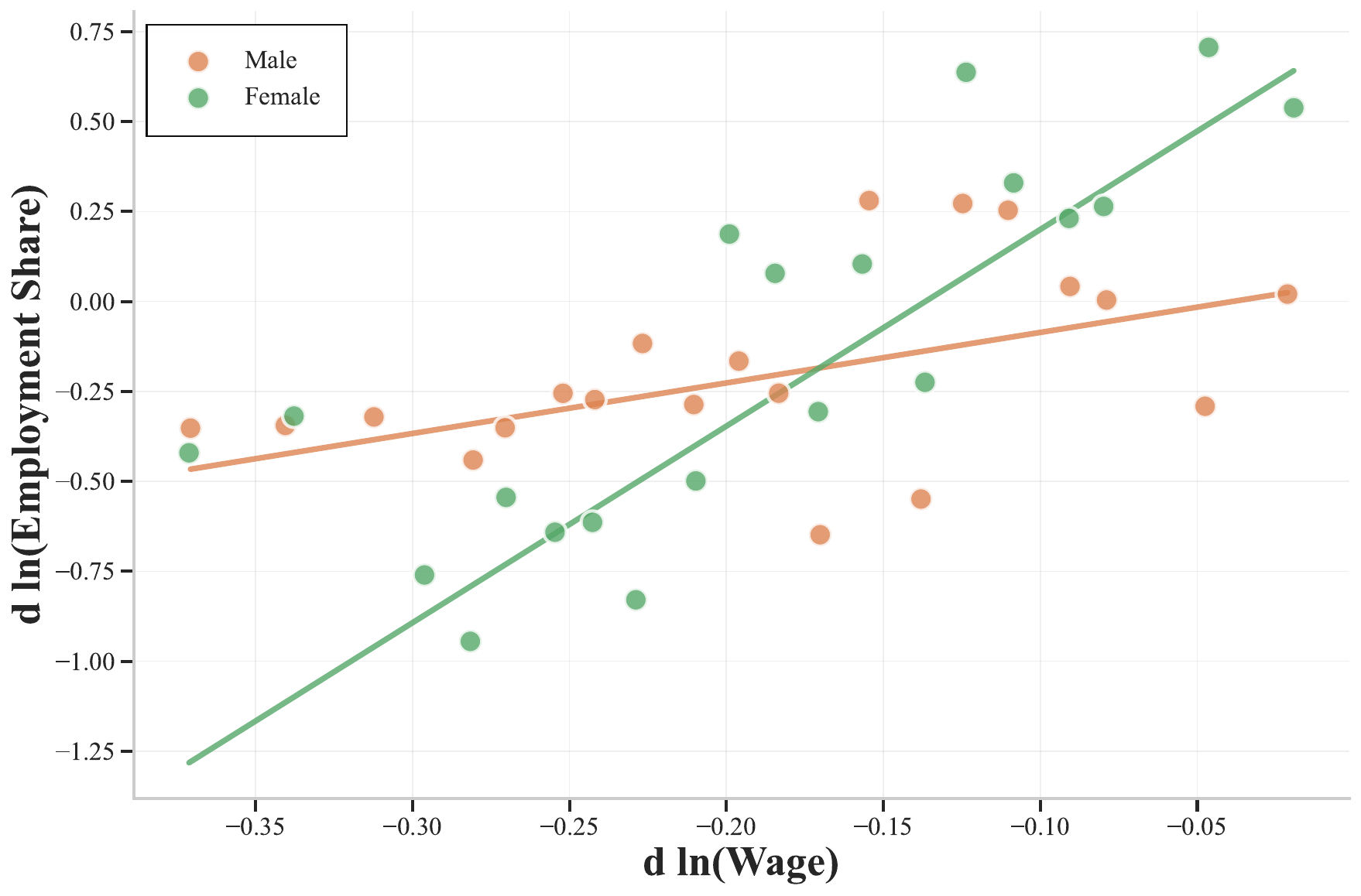}}\hfill
    \subcaptionbox{Black Workers}{\includegraphics[scale=0.25]{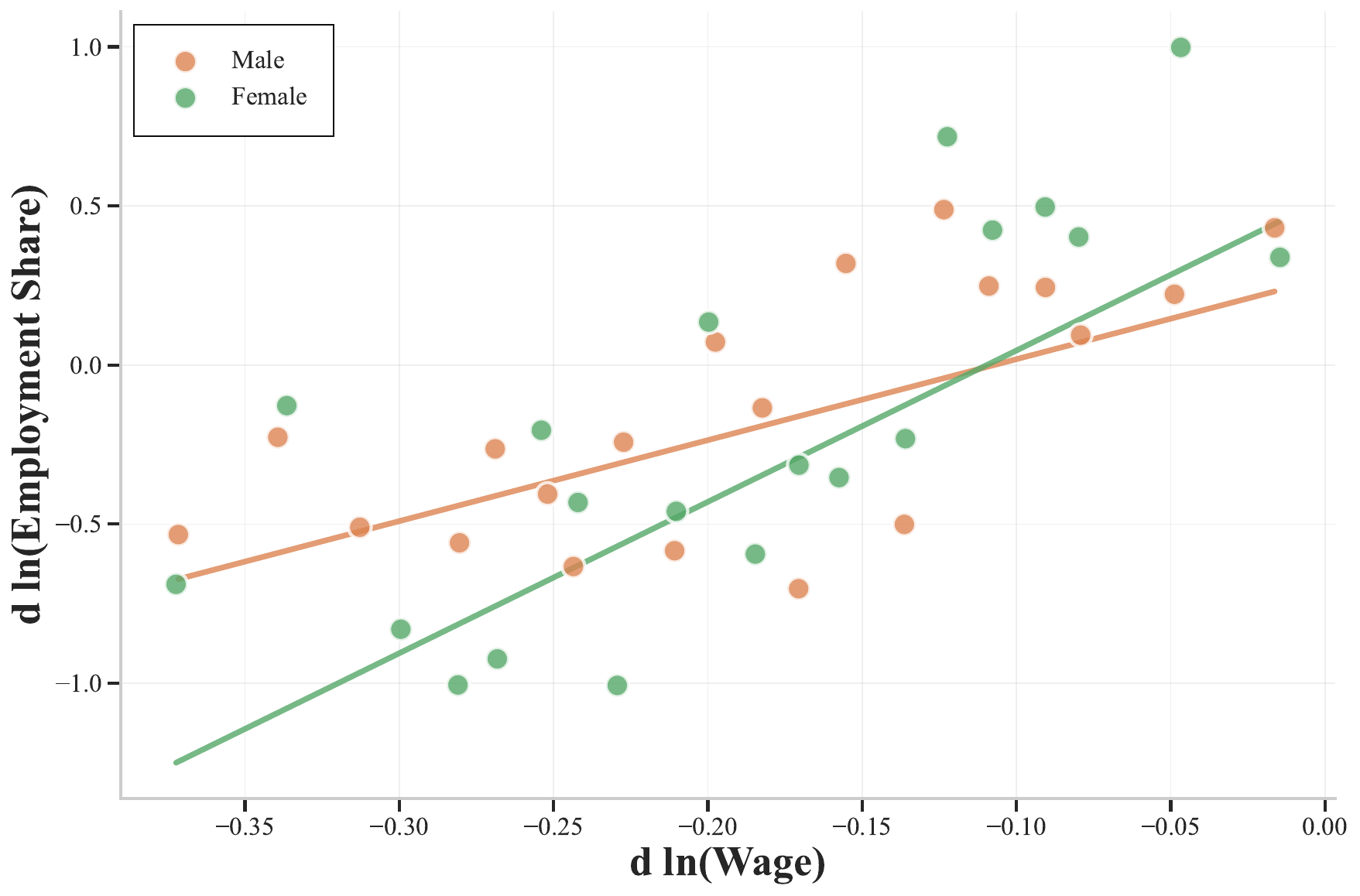}}
    \caption{Heterogeneous Employment Effects by Demographic Groups}
    \note{\textit{Notes:} Panel A shows the employment effects of automation for white workers by gender, while Panel B presents the results for Black workers.}
    \label{f:hetero_emp_demo}
\end{figure}

The PPML estimator jointly incorporates changes in the employment distribution, naturally weighting employment shares in the estimation process.

\subsection{Alternative Nested CES Specifications and Model Comparison} \label{b:appendix:nested_ces}

This section demonstrates why standard Nested CES specifications fail to capture realistic substitution patterns, validating our CNCES framework's flexible skill-based approach.

\paragraph{Nested CES Framework}
Standard Nested CES partitions occupations into mutually exclusive nests $\mathcal{O} = \bigcup_{n=1}^N \mathcal{N}_n$ with $\mathcal{N}_i \cap \mathcal{N}_j = \emptyset$. This generates within-nest elasticity $\theta/(1-\rho_n)$ and cross-nest elasticity $\theta$, assuming each occupation belongs to exactly one nest.\footnote{The productivity distribution is $\Pr[\boldsymbol{\epsilon}(i) \leq \boldsymbol{\epsilon}] = \exp\left[-\sum_{n=1}^N\left(\sum_{o\in \mathcal{N}_n}\epsilon_{o}^{\frac{-\theta}{1-\rho_n}}\right)^{1-\rho_n}\right]$.}

\paragraph{Estimation Results}
Table \ref{tab:nested_ces_results} presents PPML estimates for two standard nesting structures using our 1980--2000 automation data.

\begin{table}[ht]
\centering
\caption{Nested CES Estimation Results}
\label{tab:nested_ces_results}
\footnotesize
\renewcommand{\arraystretch}{1.2}
\begin{tabular}{l|ccc|ccc}
\toprule
 & \multicolumn{3}{c|}{\textbf{Occupation Categories}} & \multicolumn{3}{c}{\textbf{Skill Intensity}} \\
\midrule
$\theta$ (Cross-nest) & \multicolumn{3}{c|}{2.67 (0.31)} & \multicolumn{3}{c}{2.05 (0.28)} \\
\midrule
Nest & Low-skill & High-skill & Manuf. & Cognitive & Manual & Interpers. \\
$\rho_n$ (Within-nest) & 0.26 (0.18) & 0.00 (--) & 0.00 (--) & 0.30 (0.15) & 0.00 (--) & 0.45 (0.20) \\
Within-nest elasticity & 3.61 & 2.67 & 2.67 & 2.93 & 2.05 & 3.73 \\
\bottomrule
\end{tabular}
\note{\textit{Notes:}Standard errors in parentheses. Dashes indicate parameters constrained to zero.}
\end{table}

Three patterns emerge. First, most within-nest correlations $\rho_n$ are statistically zero: occupations within predefined categories are no closer substitutes than the cross-nest average. Second, cross-nest elasticities (2.05--2.67) approach our CES benchmark of 3.12, indicating rigid nesting provides minimal improvement over independence. Third, these estimates contrast sharply with our CNCES results: $\theta = 1.10$ with substantial correlations ($\rho_{\text{Cog}} = 0.77$, $\rho_{\text{Man}} = 0.48$), revealing two-thirds of substitution occurs within skill dimensions.

\paragraph{Why Nested CES Fails}
The failure stems from imposing discrete boundaries on continuous skill intensities. Nested CES assumes $\omega_o^s \in \{0,1\}$—each occupation uses exactly one skill. Our data reveals occupations draw from 2.3 skills on average with continuous intensities $\omega_o^s \in (0,1)$. A financial analyst primarily uses cognitive skills ($\omega_o^{\text{cog}} = 0.75$) but also requires interpersonal abilities ($\omega_o^{\text{int}} = 0.20$). Forcing such occupations into single nests destroys the natural substitution structure.\footnote{Mathematically, CNCES nests standard Nested CES when $\omega_o^s \in \{0,1\}$, reducing to $F(x_1, \ldots, x_O) = \sum_{s} \left[\sum_{o: \omega_o^s = 1} x_o^{\frac{1}{1-\rho_s}}\right]^{1-\rho_s}$.}

These results confirm that flexible skill-based distances—not arbitrary categorical boundaries—determine occupational substitutability and shape technological incidence.

\subsection{The Network Topology of AI Exposure} \label{appendix:ai_topology}

To complement our analysis of automation exposure in the main text, this section examines how AI exposure maps onto the occupational substitution network using 2018 data, when AI capabilities had become more clearly defined.

\begin{figure}[ht]
    \centering
    \includegraphics[width = \linewidth]{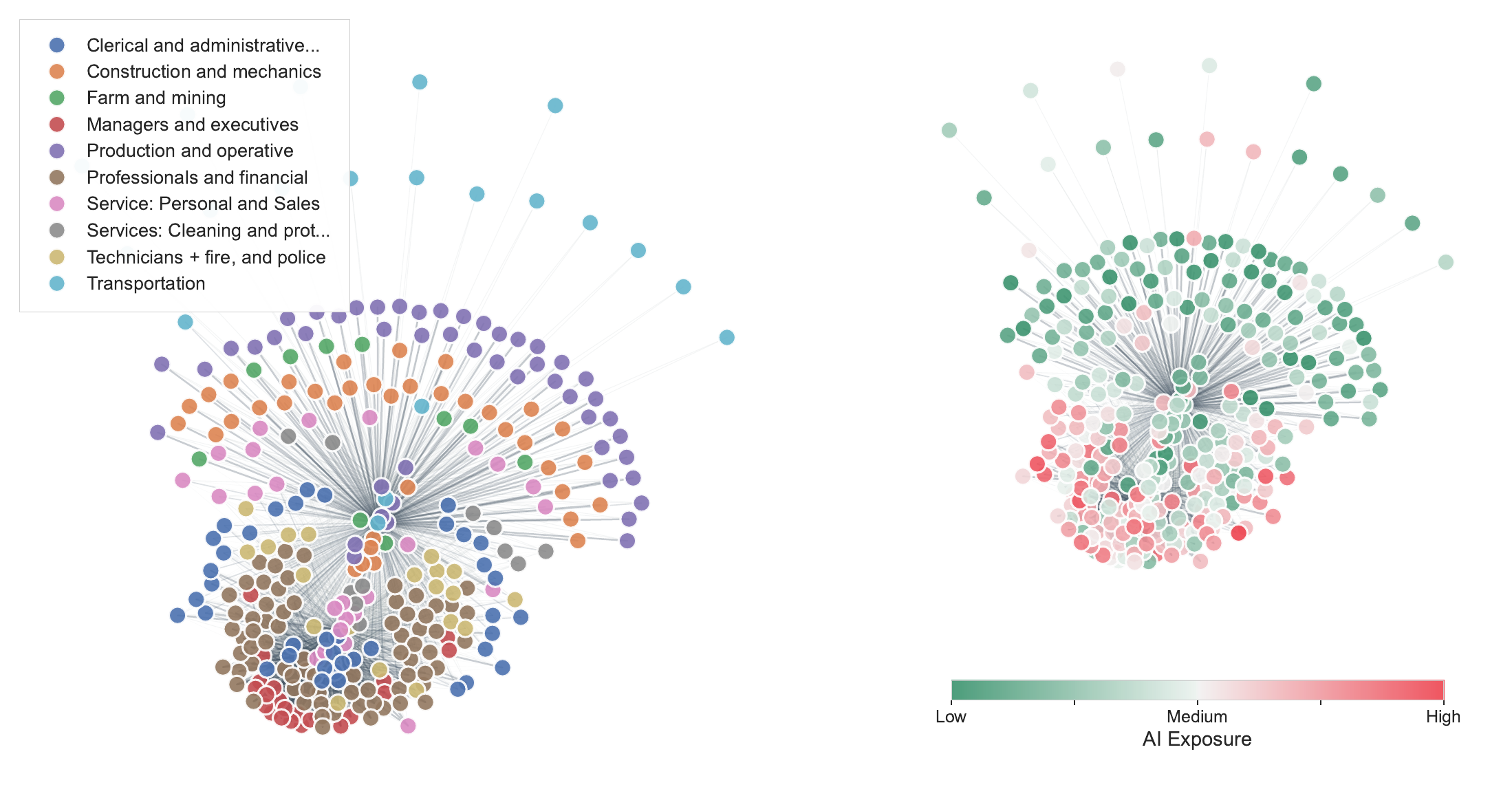}
    \caption{The Network Structure of Occupational Substitutability and AI Exposure, 2018}
    \label{fig:occupation_topology_2018}
    \note{\textit{Notes:}This figure presents the occupational substitution network for 2018, following the same methodology as Figure~\ref{fig:occupation_topology_1980}. Edges represent substitutability between occupation pairs based on estimated cross-wage elasticities, with darker lines indicating stronger substitution relationships. The left panel shows occupational clustering by broad categories, while the right panel maps AI exposure using a green gradient (darker green indicates higher AI exposure). The concentration of AI exposure in professional and cognitive-intensive occupations contrasts with the automation pattern, yet exhibits similar clustering within skill-adjacent occupations.}
\end{figure}

Figure~\ref{fig:occupation_topology_2018} demonstrates that our parsimonious skill-based framework remains robust over time, continuing to generate natural occupational clusters that align with economic intuition. Despite nearly four decades of technological change and labor market evolution between 1980 and 2018, the fundamental structure persists: occupations group according to their cognitive, manual, and interpersonal skill intensities. The professional cluster remains cohesive, production occupations maintain their tight interconnections, and service occupations continue to form distinct sub-clusters based on their specific skill combinations.

The stability of this topology structure validates our modeling choice to characterize occupations by their location in a three-dimensional skill space. The 2018 topology shows some evolution—certain connections have strengthened while others have weakened—but the overall topology remains remarkably consistent. This persistence suggests that the skill-based organization of work represents a fundamental feature of the labor market rather than a temporary configuration. Our CNCES framework, by incorporating these skill dimensions through the correlation structure, captures this enduring architecture.

The right panel reveals how AI exposure maps onto the occupational network, providing a striking contrast to automation. AI exposure (shown in darker green) concentrates in the professional and financial cluster, with particularly high intensity among occupations requiring advanced cognitive skills such as financial analysts, market researchers, and technical writers. This concentration extends to adjacent clerical and administrative occupations that share cognitive skill intensities. The spatial pattern of AI exposure—clustering in cognitive-intensive occupations at the opposite end of the network from automation's manual-intensive targets—emerges naturally from AI's capacity to perform tasks involving pattern recognition, language processing, and analytical reasoning.

Notably, both automation and AI exhibit the clustering pattern: technological shocks concentrate within skill-adjacent occupations rather than dispersing randomly across the network. This parallel structure, despite affecting different segments of the labor market, underscores a fundamental insight of our model. When technologies target specific skills, they necessarily affect clusters of related occupations. The topology visualization makes this abstract concept concrete, showing how our three-skill parameterization successfully captures the complex substitution patterns that govern labor market adjustment to technological change.
\subsection{Gender-Specific Spectral Decomposition} \label{app:gender_spectral}

The spectral decomposition reveals striking gender differences in how technological shocks interact with occupational structure, reflecting distinct employment distributions across skill space.

\begin{figure}[ht]
    \centering
    \subcaptionbox{Male Workers}{\includegraphics[width=0.48\linewidth]{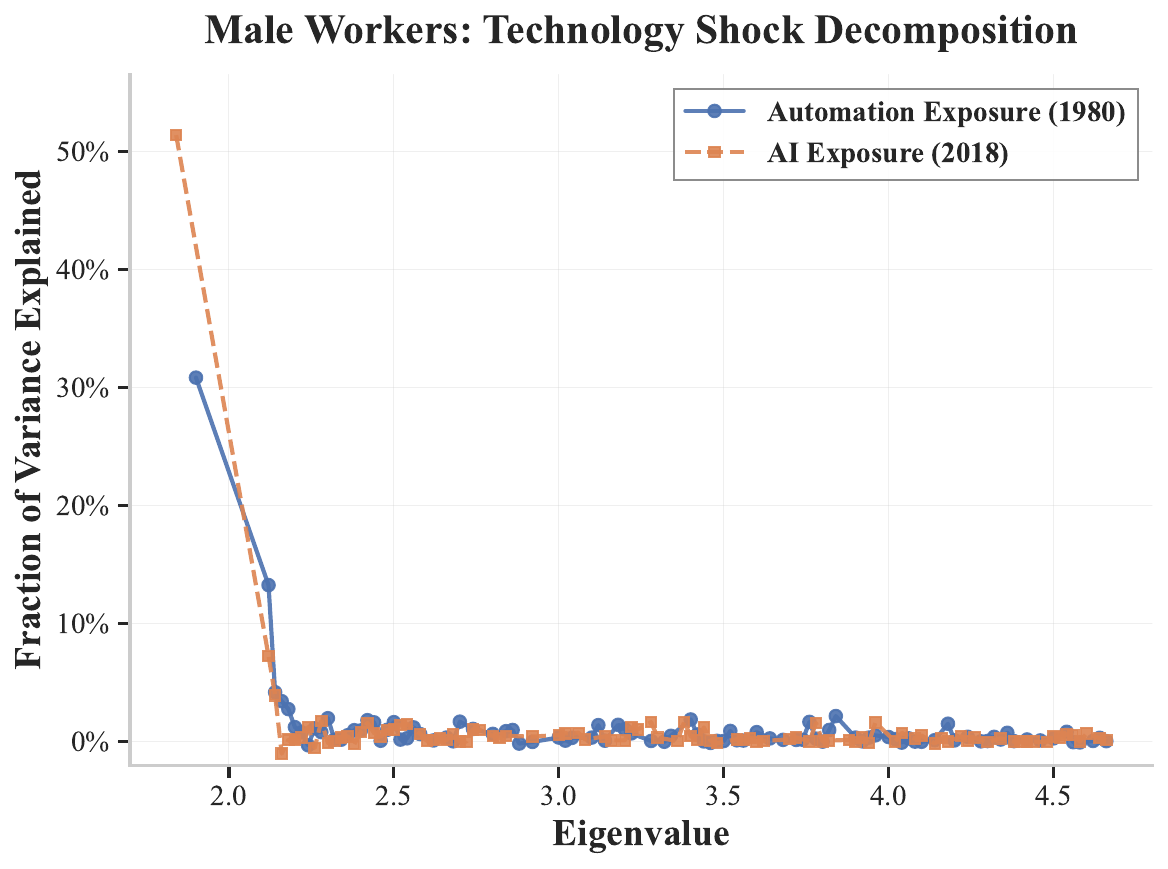}}\hfill
    \subcaptionbox{Female Workers}{\includegraphics[width=0.48\linewidth]{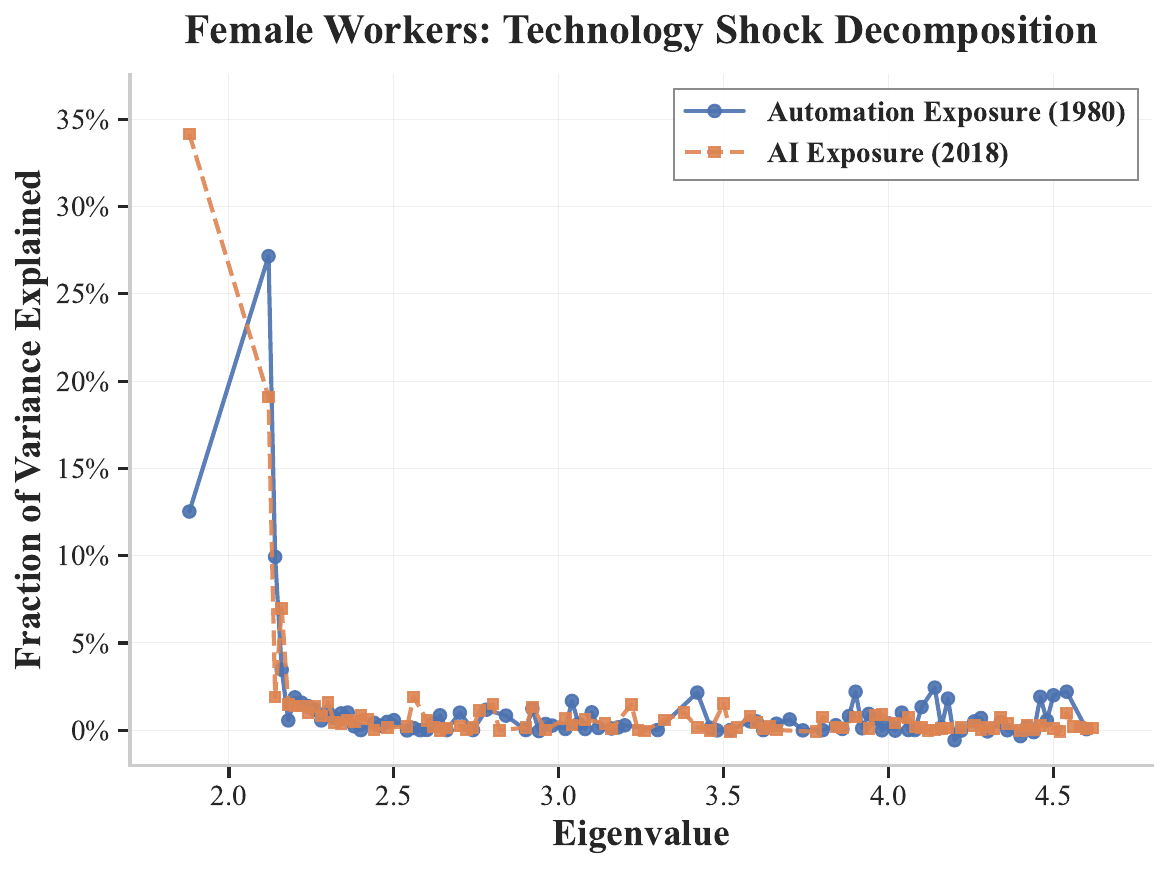}}
    \caption{Technology Shock Decomposition by Gender}
    \label{f:spectral_gender}
    \note{\textit{Notes:}Variance decomposition of automation and AI exposures into eigenshocks for male and female workers using substitution matrices based on 1980 (automation) and 2018 (AI) employment shares. Gender-specific employment distributions generate distinct substitution structures and eigenvalue patterns.}
\end{figure}

\paragraph{Automation's Gender-Differentiated Impact}
Figure~\ref{f:spectral_gender} reveals that automation constrains male workers more severely than female workers. For men (Panel A), automation concentrates 31\% of variance on the smallest eigenvalue (1.95), compared to 27\% at eigenvalue 2.15 for women (Panel B). This difference reflects occupational segregation: men dominate production and operative occupations where automation clusters, while women's employment disperses across service, clerical, and professional occupations. 

The eigen-decomposition difference implies different adjustment capacities. Male workers face effective elasticity of approximately 1.9, while female workers retain elasticity near 2.3. This 20\% difference in mobility translates directly to wage incidence: male production workers experience pass-through rates approaching 45\%, while female workers in similar exposure levels face 38\% pass-through. The gender gap emerges not from different skill transferability but from employment concentration—men's overrepresentation in manual-intensive occupations creates fewer escape routes when automation strikes.

\paragraph{AI's Convergent Pattern}
In contrast, AI exposure shows remarkable similarity across genders. Both panels display extreme concentration on the smallest eigenvalue: 51\% for men and 34\% for women, both at eigenvalue 1.95. Despite women's higher representation in cognitive-intensive occupations potentially affected by AI, the clustering pattern remains universal. This convergence suggests AI's broad reach across cognitive tasks affects both gender-segregated and integrated occupations equally.

The similar eigenvalue loading implies comparable mobility constraints. Both male and female workers in AI-exposed occupations face effective elasticities around 2.4, generating pass-through rates of 20 - 50\%. Unlike automation, where occupational segregation provides some insulation for female workers, AI's cognitive focus creates uniform rigidity across gender lines.

\paragraph{Implications for Distributional Analysis}
These gender-specific patterns highlight how initial employment distributions shape technological incidence. Automation's concentration in male-dominated manual occupations amplifies its impact on male workers through both direct exposure and constrained mobility. Female workers' diversification across skill clusters—partly reflecting historical exclusion from manufacturing—inadvertently provides better adjustment options. For AI, the universal nature of cognitive tasks eliminates this protective dispersion, suggesting future technological shocks may generate more uniform gender impacts while maintaining severe clustering effects.
\subsection{The Employment Effects of Automation and AI} \label{s:emp_effect_auto_ai}
As discussed in the main text, clustering shocks lead to smaller employment adjustments while exacerbating wage disparities. Panel (a) of Figure \ref{f:emp_effects_auto_ai} illustrates the relationship between changes in log employment shares and relative wage changes for automation. The CES benchmark (dashed line) rotates counterclockwise, overstating employment shifts, particularly for negatively impacted occupations. This suggests that the CES framework underestimates the rigidity in labor reallocation caused by clustering shocks.

Panel (b) presents the same employment effects for AI exposure, revealing a similar pattern. The CES model again overstates employment adjustments, failing to account for the constrained worker mobility induced by the skill-clustering nature of AI-exposed occupations. These findings highlight the importance of incorporating a richer substitution structure, as captured by DIDES, to better reflect labor market frictions in response to technological change.

\begin{figure}[ht]
    \centering
    \subcaptionbox{Automation}{\includegraphics[scale=0.39]{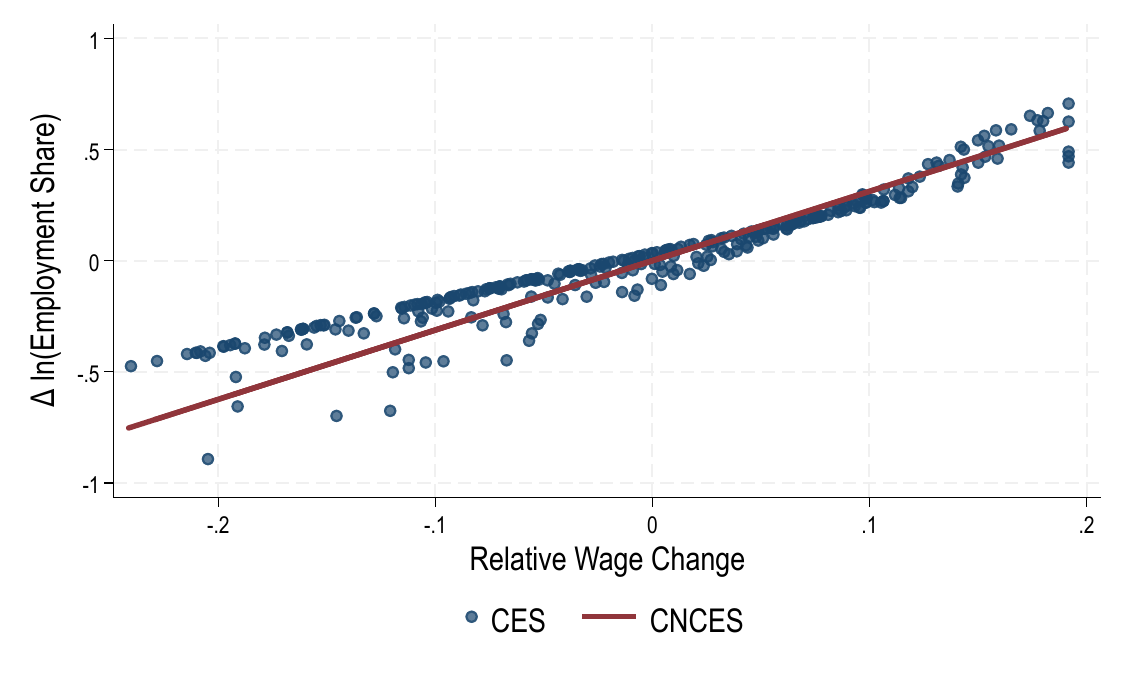}}\hfill
    \subcaptionbox{AI}{\includegraphics[scale=0.39]{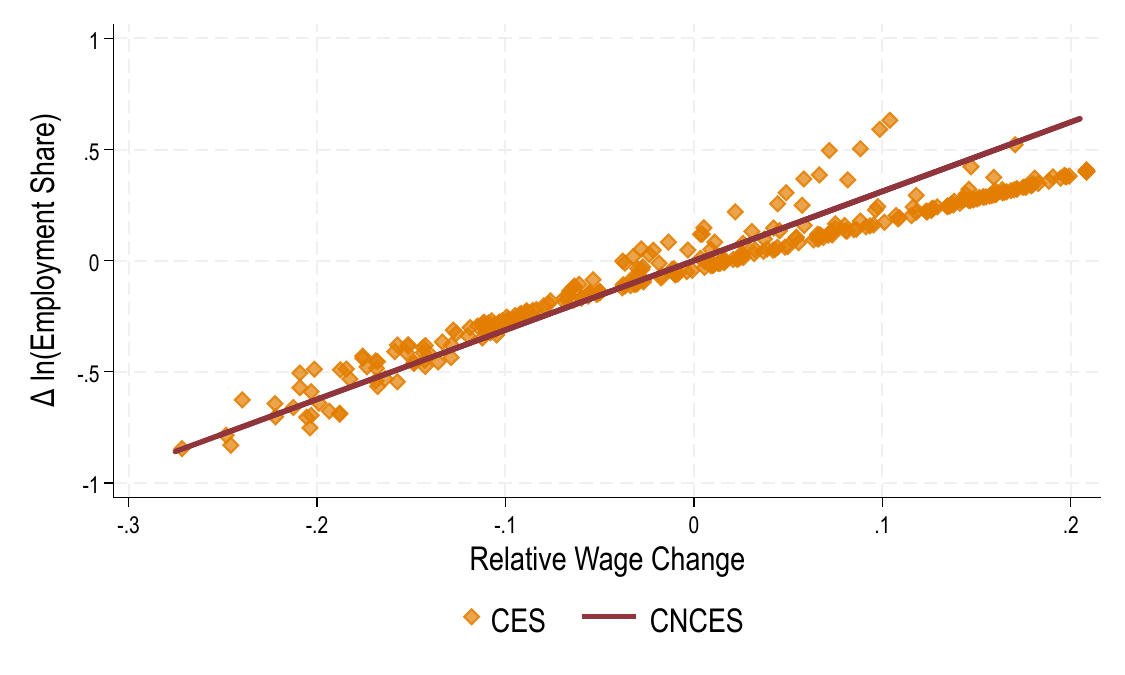}}
    \caption{Employment Effects of Technological Shocks}
    \note{\textit{Notes:} This figure compares employment effects of automation (Panel a) and AI (Panel b) against the CES benchmark. The CES framework overestimates employment adjustments, particularly in negatively impacted occupations, due to its failure to account for clustering shocks that restrict labor mobility.}
    \label{f:emp_effects_auto_ai}
\end{figure}
\subsection{Wage Pass-Through with Productivity in Production} \label{app:passthrough_efficiency}

When $\delta = 1$ (all workers contribute productivity to production), the labor supply elasticity becomes:
\begin{equation*}
    \Theta_{oo'}^{\text{eff}} = \begin{cases}
        \theta\frac{x_oF_{oo}}{F_o} + (\theta-1)(1-\pi_o) & \text{if } o = o' \\
        \theta\frac{x_{o'}F_{oo'}}{F_o} - (\theta-1)\pi_{o'} & \text{if } o \neq o'
    \end{cases}
\end{equation*}

With $\theta = 1.10$, the baseline substitution term $(\theta-1) = 0.10$ nearly vanishes, leaving primarily the correlation-driven term $\theta x_{o'}F_{oo'}/F_o$ to govern labor market adjustment. This dramatically amplifies wage pass-through relative to our baseline specification.

\begin{figure}[ht]
    \centering
    \begin{subfigure}[b]{0.48\textwidth}
        \centering
        \includegraphics[width=\textwidth]{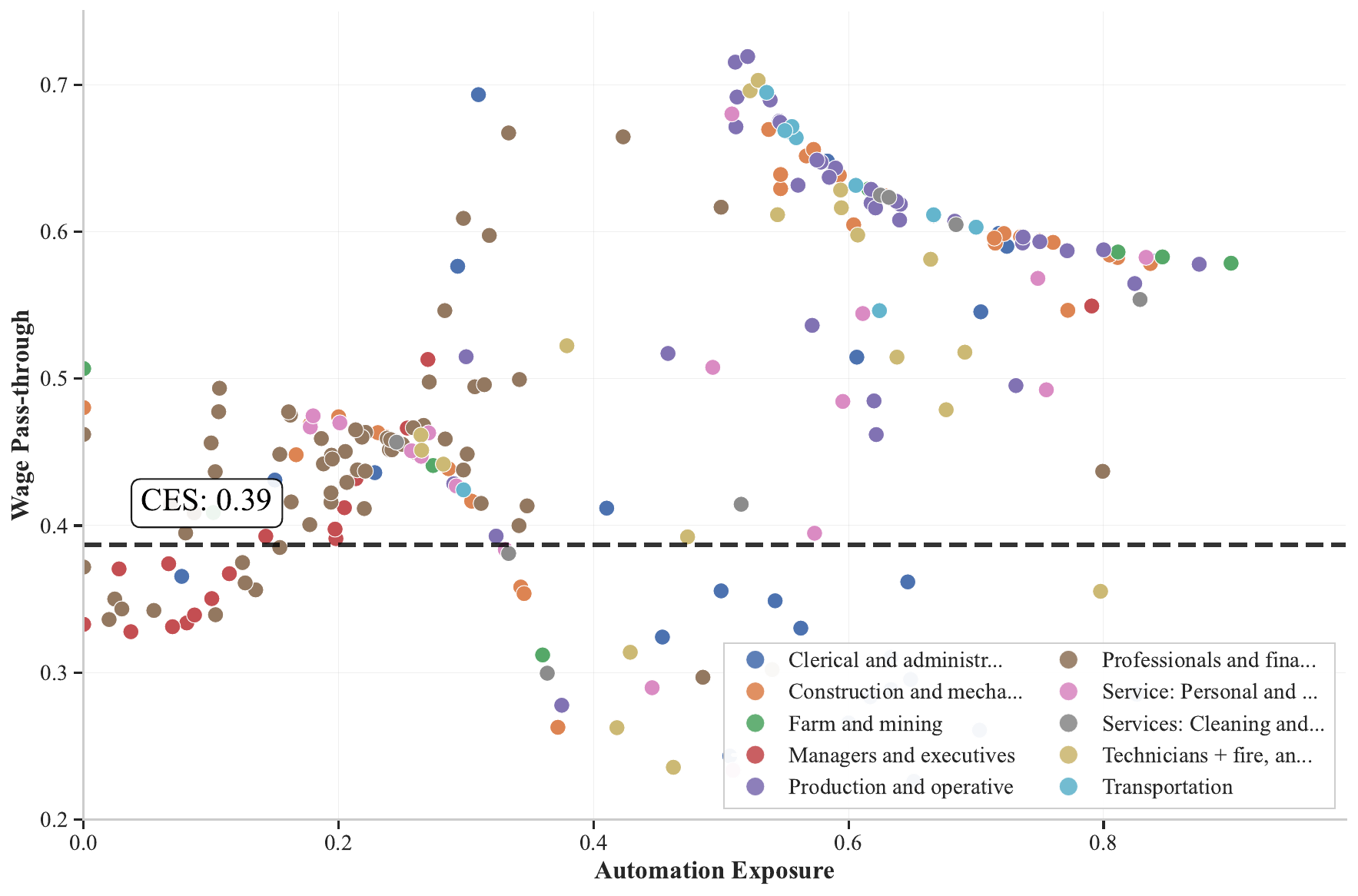}
        \caption{Automation exposure and wage pass-through}
        \label{fig:robot_passthrough_efficiency}
    \end{subfigure}
    \hfill
    \begin{subfigure}[b]{0.48\textwidth}
        \centering
        \includegraphics[width=\textwidth]{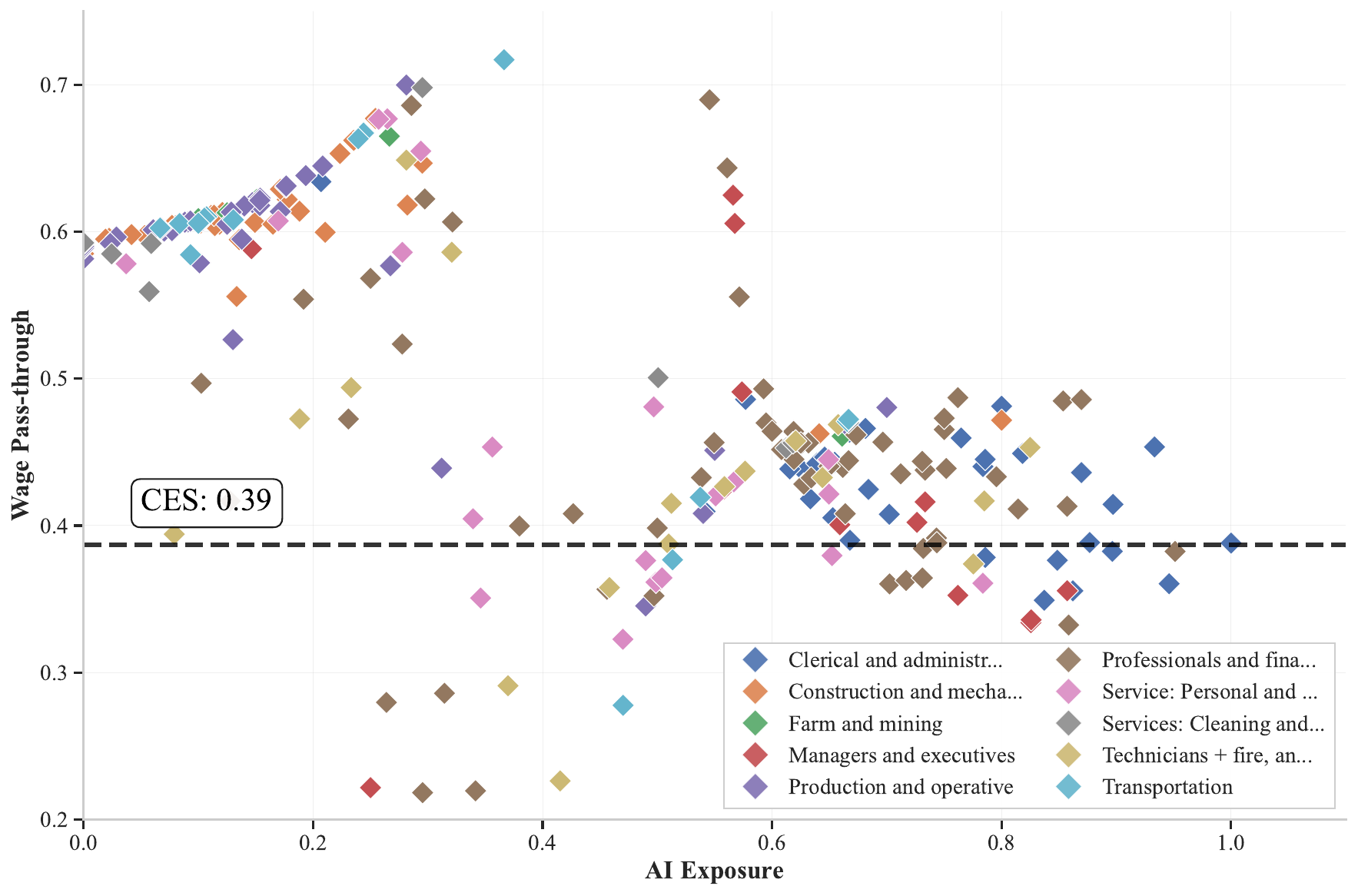}
        \caption{AI exposure and wage pass-through}
        \label{fig:ai_passthrough_efficiency}
    \end{subfigure}
    \caption{Wage Pass-Through under $\delta = 1$ (Productivity in Production)}
    \label{fig:passthrough_efficiency}
\end{figure}

Figure \ref{fig:passthrough_efficiency} illustrates this amplification. Under $\delta = 1$, wage pass-through ranges from 0.20 to 0.70 across occupations, compared to the uniform CES benchmark of 0.387. The substantial variation—entirely absent in the CES model—stems from heterogeneous correlation patterns. 
Two patterns emerge distinctly:
\begin{enumerate}
    \item \textbf{Clustering amplification:} Technology-exposed occupations cluster in pass-through space, with AI-exposed occupations (Panel a) showing particularly tight grouping around 0.60-0.65, well above the CES benchmark
    \item \textbf{CES misestimation:} The uniform CES pass-through of 0.387 masks enormous heterogeneity, overestimating effects for some occupations by 50\% while underestimating for others by 80\%
\end{enumerate}

These results underscore that our baseline estimates ($\delta = 0$) provide conservative bounds on technological incidence. If workers' productivity enters production even partially, wage inequality from technological clustering substantially exceeds our documented effects, with correlation patterns becoming the dominant force in wage determination.

\subsection{Construct Job Transition with CPS} \label{b:appendix:CPS}

Our estimation strategy hinges on observing aggregate job flows across occupations. To construct our occupation-level panel for the period 1980–2018, we rely on individual-level data from the US Census Bureau’s March Current Population Survey (CPS). Each March CPS provides detailed information on respondents’ current occupation as well as the occupation in which they spent most of the previous calendar year. We restrict our sample to individuals aged 25–64 who are employed full-time and have worked at least 26 weeks in the preceding year, thereby ensuring the reliability of our occupational transition estimates. We also exclude observations with extreme or inconsistent income values to mitigate measurement error. Using these data, we construct annual job flow rates for occupations.

Employing a consistent occupation coding scheme, we generate a balanced panel of 306 three-digit occupations. Given the sparsity of observed transitions at this detailed level, we further aggregate these occupations into 15 clusters using a k-means algorithm based on occupational skill intensities. This intuitive clustering groups together occupations with similar skill profiles, ensuring robust estimates of aggregate job flows and facilitating subsequent analyses.

Furthermore, as noted by \cite{Artuc2010-xh}, the retrospective design of the March CPS captures job transitions over a period shorter than a full year—respondents report the longest-held job from the previous calendar year, typically reflecting employment around mid-year. To correct for this timing bias, we annualize the observed job transition probabilities using the transformation \footnote{This approach ensures that no annual job-to-job flows are missing.}
\begin{equation*}
    \mu^{\text{ANN}}_t = \mu_t^{2}.
\end{equation*}

\newpage

\section{Additional Materials}\label{c:appendix}

\setcounter{theorem}{0}
\setcounter{proposition}{0} 
\setcounter{lemma}{0}
\setcounter{corollary}{0}
\setcounter{definition}{0}
\setcounter{assumption}{0}
\setcounter{remark}{0}
\setcounter{table}{0}
\setcounter{figure}{0}
\setcounter{equation}{0} 
%
\renewcommand{\thetheorem}{C\arabic{theorem}}
\renewcommand{\theproposition}{C\arabic{proposition}}
\renewcommand{\thelemma}{C\arabic{lemma}}
\renewcommand{\thecorollary}{C\arabic{corollary}}
\renewcommand{\thedefinition}{C\arabic{definition}}
\renewcommand{\theassumption}{C\arabic{assumption}}
\renewcommand{\theremark}{C\arabic{remark}}
\renewcommand{\thetable}{C\arabic{table}}
\renewcommand{\thefigure}{C\arabic{figure}}
\renewcommand{\theequation}{C\arabic{equation}}

\subsection{Derivation of Employment Shares} \label{app_ss:emp_share_derivation}

This appendix derives the closed-form expression for occupational employment shares under the multivariate Fréchet productivity distribution. 

\begin{proposition}
Given the joint productivity distribution:
\begin{equation*}
    \Pr[\epsilon_1(i) \leq \epsilon_1, \ldots, \epsilon_O(i) \leq \epsilon_O] = \exp\left[-F\left(A_1\epsilon_1^{-\theta}, \ldots, A_O\epsilon_O^{-\theta}\right)\right]
\end{equation*}
the share of workers choosing occupation $o$ is:
\begin{equation*}
    \pi_o = \frac{A_o w_o^{\theta} F_o(A_1w_1^{\theta}, \ldots, A_Ow_O^{\theta})}{F(A_1w_1^{\theta}, \ldots, A_Ow_O^{\theta})}
\end{equation*}
\end{proposition}

\begin{proof}
We derive the probability that a worker chooses occupation $o$, which occurs when $w_o\epsilon_o(i) \geq w_{o'}\epsilon_{o'}(i)$ for all $o' \neq o$.

First, consider the joint probability that occupation $o$ yields utility less than $t$ and is optimal:
\begin{align*}
    &\Pr[w_o\epsilon_o(i) < t \text{ and } w_o\epsilon_o(i) = \max_{o'} w_{o'}\epsilon_{o'}(i)] \nonumber \\
    &= \Pr[w_o\epsilon_o(i) < t \text{ and } w_o\epsilon_o(i) \geq w_{o'}\epsilon_{o'}(i), \forall o' \neq o]
\end{align*}

This equals the probability that all occupational utilities are below $t$, with occupation $o$ being the highest. Using the law of total probability:
\begin{align*}
    &= \int_0^t \frac{\partial}{\partial z} \Pr[w_{o'}\epsilon_{o'} \leq z, \forall o'] \bigg|_{z} dz
\end{align*}

Substituting the joint distribution and differentiating:
\begin{align*}
    &= \int_0^t \frac{\partial}{\partial z} \exp\left[-F\left(A_1w_1^{\theta}z^{-\theta}, \ldots, A_Ow_O^{\theta}z^{-\theta}\right)\right] dz \nonumber \\
    &= \int_0^t A_ow_o^{\theta} F_o\left(A_1w_1^{\theta}z^{-\theta}, \ldots, A_Ow_O^{\theta}z^{-\theta}\right) \nonumber \\
    &\quad \times \exp\left[-F\left(A_1w_1^{\theta}z^{-\theta}, \ldots, A_Ow_O^{\theta}z^{-\theta}\right)\right] \theta z^{-\theta-1} dz
\end{align*}

Using the homogeneity of degree one property of $F$:
\begin{align*}
    F\left(A_1w_1^{\theta}z^{-\theta}, \ldots, A_Ow_O^{\theta}z^{-\theta}\right) = z^{-\theta} F\left(A_1w_1^{\theta}, \ldots, A_Ow_O^{\theta}\right)
\end{align*}

Since $F_o$ is homogeneous of degree zero (as the derivative of a degree-one homogeneous function):
\begin{align*}
    F_o\left(A_1w_1^{\theta}z^{-\theta}, \ldots, A_Ow_O^{\theta}z^{-\theta}\right) = F_o\left(A_1w_1^{\theta}, \ldots, A_Ow_O^{\theta}\right)
\end{align*}

Substituting these properties:
\begin{align*}
    &= \int_0^t A_ow_o^{\theta} F_o(A_1w_1^{\theta}, \ldots, A_Ow_O^{\theta}) \exp\left[-F(A_1w_1^{\theta}, \ldots, A_Ow_O^{\theta})z^{-\theta}\right] \theta z^{-\theta-1} dz \nonumber \\
    &= \frac{A_ow_o^{\theta} F_o(A_1w_1^{\theta}, \ldots, A_Ow_O^{\theta})}{F(A_1w_1^{\theta}, \ldots, A_Ow_O^{\theta})} \nonumber \\
    &\quad \times \int_0^t \exp\left[-F(A_1w_1^{\theta}, \ldots, A_Ow_O^{\theta})z^{-\theta}\right] F(A_1w_1^{\theta}, \ldots, A_Ow_O^{\theta}) \theta z^{-\theta-1} dz
\end{align*}

The integral evaluates to:
\begin{align*}
    \int_0^t \exp[-\Lambda z^{-\theta}] \Lambda \theta z^{-\theta-1} dz = 1 - \exp[-\Lambda t^{-\theta}]
\end{align*}
where $\Lambda = F(A_1w_1^{\theta}, \ldots, A_Ow_O^{\theta})$.

Therefore:
\begin{align*}
    \Pr[w_o\epsilon_o(i) < t \text{ and optimal}] = \frac{A_ow_o^{\theta} F_o(A_1w_1^{\theta}, \ldots, A_Ow_O^{\theta})}{F(A_1w_1^{\theta}, \ldots, A_Ow_O^{\theta})} \left(1 - \exp[-\Lambda t^{-\theta}]\right)
\end{align*}

Taking the limit as $t \to \infty$:
\begin{align*}
    \pi_o &= \lim_{t \to \infty} \Pr[w_o\epsilon_o(i) < t \text{ and } w_o\epsilon_o(i) = \max_{o'} w_{o'}\epsilon_{o'}(i)] \nonumber \\
    &= \frac{A_ow_o^{\theta} F_o(A_1w_1^{\theta}, \ldots, A_Ow_O^{\theta})}{F(A_1w_1^{\theta}, \ldots, A_Ow_O^{\theta})}
\end{align*}

This completes the proof.
\end{proof}

\textbf{Remark:} This result crucially depends on the max-stability property of the multivariate Fréchet distribution and the homogeneity properties of the correlation function $F$. The employment share expression shows that occupation $o$'s share depends on its productivity-weighted wage ($A_ow_o^{\theta}$), scaled by how the correlation function responds to changes in that occupation's attractiveness ($F_o/F$).

\subsection{Derivation of Labor Supply Elasticities} \label{app_ss:elasticity_derivation}

This section derives the labor supply elasticity matrix $\Theta$ from the employment share equation \eqref{app_eq:employment_share}.

Starting from the employment share:
\begin{equation*}
    \pi_o = \frac{A_o w_o^{\theta} F_o(A_1w_1^{\theta}, \ldots, A_Ow_O^{\theta})}{F(A_1w_1^{\theta}, \ldots, A_Ow_O^{\theta})}
\end{equation*}

Let $x_o = A_ow_o^{\theta}$ and define the wage index $W = F(x_1, \ldots, x_O)^{1/\theta}$. The elasticity with respect to relative wages is:
\begin{align*}
    \frac{\partial\ln \pi_{o}}{\partial\ln(w_{o'}/W)} &= \frac{\partial}{\partial\ln(w_{o'}/W)} \ln\left[\frac{x_o F_o(x_1, \ldots, x_O)}{F(x_1, \ldots, x_O)}\right] \nonumber \\
    &= \frac{\partial}{\partial\ln(w_{o'}/W)} \ln\left[\left(\frac{w_o}{W}\right)^{\theta} F_o\left(\frac{x_1}{W^{\theta}}, \ldots, \frac{x_O}{W^{\theta}}\right)\right]
\end{align*}

Using the fact that $F$ is homogeneous of degree one (hence $F_o$ is homogeneous of degree zero):
\begin{equation*}
    \frac{\partial\ln \pi_{o}}{\partial\ln(w_{o'}/W)} = 
    \begin{cases}
        \theta\frac{x_{o'}F_{oo'}}{F_o} & \text{if } o' \neq o \\
        \theta\frac{x_o F_{oo}}{F_o} + \theta & \text{if } o' = o
    \end{cases}
\end{equation*}

To obtain the elasticity with respect to absolute wages, we use:
\begin{equation*}
    \frac{\partial\ln \pi_{o}}{\partial\ln w_{o'}} = \frac{\partial\ln \pi_{o}}{\partial\ln(w_{o'}/W)} \cdot \frac{\partial\ln(w_{o'}/W)}{\partial\ln w_{o'}} + \frac{\partial\ln \pi_{o}}{\partial\ln W} \cdot \frac{\partial\ln W}{\partial\ln w_{o'}}
\end{equation*}

Since $\frac{\partial\ln(w_{o'}/W)}{\partial\ln w_{o'}} = 1 - \frac{\partial\ln W}{\partial\ln w_{o'}}$ and $\frac{\partial\ln W}{\partial\ln w_{o'}} = \pi_{o'}$:
\begin{align*}
    \frac{\partial\ln \pi_{o}}{\partial\ln w_{o'}} &= \theta\frac{x_{o'}F_{oo'}}{F_o} - \theta\pi_{o'} \quad \text{for } o' \neq o \\
    \frac{\partial\ln \pi_{o}}{\partial\ln w_{o}} &= \theta\frac{x_o F_{oo}}{F_o} + \theta - \theta\pi_o = \theta\frac{x_o F_{oo}}{F_o} + \theta(1 - \pi_o)
\end{align*}

Since $L_o = \pi_o\bar{L}$, the labor supply elasticity is:
\begin{equation*}
    \Theta_{oo'} = \frac{\partial \ln L_o}{\partial \ln w_{o'}} = \frac{\partial \ln \pi_o}{\partial \ln w_{o'}} = 
    \begin{cases}
        \theta\left[\frac{x_{o'}F_{oo'}}{F_o}\bigg|_{x_j = A_jw_j^{\theta}} - \pi_{o'}\right] & \text{if } o \neq o' \\
        \theta\left[\frac{x_oF_{oo}}{F_o}\bigg|_{x_j = A_jw_j^{\theta}} + 1 - \pi_o\right] & \text{if } o = o'
    \end{cases}
\end{equation*}

This completes the derivation of the labor supply elasticity matrix in equation \eqref{app_eq:elasticity_matrix}.

\subsection{Zero Row Sum Property of the Elasticity Matrix} \label{app_ss:zero_row_sum}

A crucial property of the labor supply elasticity matrix $\Theta$ is that its row sums equal zero, implying that uniform wage changes do not affect relative employment. This result follows directly from the homogeneity of degree one property of the correlation function $F$.

\begin{proposition}
For the elasticity matrix $\Theta$ defined in equation \eqref{app_eq:elasticity_matrix}, we have:
\begin{equation*}
    \sum_{o'=1}^O \Theta_{oo'} = 0 \quad \text{for all } o
\end{equation*}
\end{proposition}

\begin{proof}
Starting from the definition of employment shares in equation \eqref{app_eq:employment_share}:
\begin{equation*}
    \pi_o = \frac{A_o w_o^{\theta} F_o(A_1w_1^{\theta}, \ldots, A_Ow_O^{\theta})}{F(A_1w_1^{\theta}, \ldots, A_Ow_O^{\theta})}
\end{equation*}

Let $x_o = A_o w_o^{\theta}$ for notational simplicity. The sum of elasticities for row $o$ is:
\begin{align*}
    \sum_{o'=1}^O \Theta_{oo'} &= \sum_{o'=1}^O \frac{\partial \ln \pi_o}{\partial \ln w_{o'}} \\
    &= \sum_{o'=1}^O \theta \frac{\partial \ln \pi_o}{\partial \ln x_{o'}} \\
    &= \theta \sum_{o'=1}^O \left[ \frac{x_{o'} F_{oo'}}{F_o} - \pi_{o'} \right] + \theta \cdot \mathbf{1}_{o=o'} \\
    &= \frac{\theta}{F_o} \sum_{o'=1}^O x_{o'} F_{oo'} + \theta(1 - \pi_o) - \theta \sum_{o' \neq o} \pi_{o'}
\end{align*}

Since $\sum_{o'=1}^O \pi_{o'} = 1$, we have $1 - \pi_o = \sum_{o' \neq o} \pi_{o'}$, which simplifies the expression to:
\begin{equation*}
    \sum_{o'=1}^O \Theta_{oo'} = \frac{\theta}{F_o} \sum_{o'=1}^O x_{o'} F_{oo'}
\end{equation*}

Now we invoke Euler's theorem for homogeneous functions. Since $F$ is homogeneous of degree one, its partial derivative $F_o$ is homogeneous of degree zero. By Euler's theorem applied to $F_o$:
\begin{equation*}
    \sum_{o'=1}^O x_{o'} \frac{\partial F_o}{\partial x_{o'}} = 0 \cdot F_o = 0
\end{equation*}

But $\frac{\partial F_o}{\partial x_{o'}} = F_{oo'}$ by definition, therefore:
\begin{equation*}
    \sum_{o'=1}^O x_{o'} F_{oo'} = 0
\end{equation*}

This immediately implies:
\begin{equation*}
    \sum_{o'=1}^O \Theta_{oo'} = \frac{\theta}{F_o} \cdot 0 = 0
\end{equation*}
\end{proof}

This zero row sum property has important economic implications. It ensures that proportional wage increases—such as those resulting from aggregate productivity growth—do not induce occupational reallocation. Only relative wage changes, such as those caused by asymmetric technological shocks, trigger worker mobility across occupations. This property is essential for the model's consistency and ensures that the labor market responds only to distributional shocks rather than level effects.

\subsection{Proof of Equilibrium Existence and Uniqueness} \label{app_ss:equilibrium_proof}

\begin{proposition}
Given the production structure in Section \ref{app_ss:production} and labor supply in Section \ref{app_ss:labor_supply}, a unique competitive equilibrium exists.
\end{proposition}

\begin{proof}
\textbf{Existence:}  
Define employment shares $\lambda_o = L_o/\bar L$ and note that market clearing requires
\[
\lambda = \pi\!\left(w(\lambda)\right),
\]
where $\pi(\cdot)$ are labor supply shares from \eqref{app_eq:employment_share} and $w(\cdot)$ are occupational wages from \eqref{app_eq:wage_equation}. The mapping $T:\Delta^{O-1}\to \Delta^{O-1}$ defined by $T(\lambda)=\pi(w(\lambda))$ is continuous: (i) $\pi(w)$ is continuous and strictly positive by the properties of the correlation function $F$, and (ii) $w(\lambda)$ is continuous from the CES production structure. Since $T$ maps the compact convex simplex $\Delta^{O-1}$ into itself, Brouwer's fixed-point theorem guarantees at least one equilibrium $\lambda^*$ and corresponding $w^*$ (unique up to a scalar normalization).

\textbf{Uniqueness:}  
Two features rule out multiple equilibria.  

(i) \emph{Labor supply:} By the sign-switching property of $F$, the elasticity matrix
\[
\Theta_{oo'} = \frac{\partial \ln L_o}{\partial \ln w_{o'}}
\]
satisfies $\Theta_{oo}>0$, $\Theta_{oo'}<0$ for $o\neq o'$, and $\sum_{o'}\Theta_{oo'}=0$. Thus, occupations are gross substitutes from the workers' perspective.  

(ii) \emph{Labor demand:} From \eqref{app_eq:wage_equation}, own-wage elasticities are negative ($\partial\ln w_o/\partial\ln L_o=-1/\sigma<0$), while cross-elasticities are positive for $o\neq o'$ under $\sigma>1$, implying that occupations are gross substitutes in production.  

Taken together, excess demand in log-space has a Jacobian that is a $P$-matrix (positive diagonal dominance with gross-substitute sign pattern). By the Gale--Nikaid\^o global univalence theorem (or equivalently Kelso--Crawford/Gul--Stacchetti arguments for gross substitutes), the fixed point $\lambda^*$ is unique, and so are relative wages $w^*$ once a normalization is imposed.

\textbf{Conclusion}  
Under $\sigma>1$ and the sign-switching (gross substitutes) property of $F$, a competitive equilibrium always exists and is unique in relative wages and employment shares.

\end{proof}

\subsection{Derivation of Wage Incidence in Proposition \ref{prop:tech_incidence}} \label{app_ss:incidence_proof}

Starting from the equilibrium conditions:
\begin{align*}
    d\ln\boldsymbol{w} &= \frac{1}{\sigma}d\ln y \cdot \mathbf{1} - \frac{1}{\sigma}d\ln\boldsymbol{\alpha} - \frac{1}{\sigma}d\ln\boldsymbol{L} \quad \text{(labor demand)} \\
    d\ln\boldsymbol{L} &= \Theta \cdot d\ln\boldsymbol{w} \quad \text{(labor supply)}
\end{align*}

Substituting the labor supply response into the demand equation:
\begin{align*}
    d\ln\boldsymbol{w} &= \frac{1}{\sigma}d\ln y \cdot \mathbf{1} - \frac{1}{\sigma}d\ln\boldsymbol{\alpha} - \frac{1}{\sigma}\Theta \cdot d\ln\boldsymbol{w} \\
    \left(\mathbf{I} + \frac{\Theta}{\sigma}\right) d\ln\boldsymbol{w} &= \frac{1}{\sigma}d\ln y \cdot \mathbf{1} - \frac{1}{\sigma}d\ln\boldsymbol{\alpha}
\end{align*}

Since $\sum_{o'}\Theta_{oo'} = 0$ for all $o$ (see Appendix \ref{app_ss:zero_row_sum}), the matrix $(\mathbf{I} + \Theta/\sigma)$ is invertible. Solving for wages:
\begin{equation*}
    d\ln\boldsymbol{w} = \frac{1}{\sigma}d\ln y \cdot \mathbf{1} - \underbrace{\left(\mathbf{I} + \frac{\Theta}{\sigma}\right)^{-1}}_{\equiv \Delta} \cdot \frac{d\ln\boldsymbol{\alpha}}{\sigma}
\end{equation*}

This establishes equation \eqref{eq:wage_incidence} with the pass-through matrix $\Delta = (\mathbf{I} + \Theta/\sigma)^{-1}$.

\subsection{Derivation of Mobility Gains} \label{app_ss:mobility_gain}

Consider a marginal worker initially in occupation $o$ who transitions to occupation $o'$ following the shock. Before the shock, this worker was indifferent between the two occupations:
\begin{equation*}
    \ln w_o + \ln\epsilon_o(i) = \ln w_{o'} + \ln\epsilon_{o'}(i)
\end{equation*}

After the shock, the worker strictly prefers $o'$:
\begin{equation*}
    \ln w_o + d\ln w_o + \ln\epsilon_o(i) < \ln w_{o'} + d\ln w_{o'} + \ln\epsilon_{o'}(i)
\end{equation*}

The equivalent variation (EV) for this marginal switcher satisfies:
\begin{equation*}
    \ln w_o + d\ln w_o + \ln\epsilon_o(i) + \text{EV}(i) = \ln w_{o'} + d\ln w_{o'} + \ln\epsilon_{o'}(i)
\end{equation*}

Using the initial indifference condition:
\begin{equation*}
    \text{EV}(i) = d\ln w_{o'} - d\ln w_o
\end{equation*}

For small changes, the share of workers transitioning from $o$ to $o'$ when $d\ln w_{o'} > d\ln w_o$ is:
\begin{equation*}
    \mu_{oo'} = -\Theta_{oo'}(d\ln w_{o'} - d\ln w_o)
\end{equation*}

Note that $\Theta_{oo'} < 0$ for $o \neq o'$, so $\mu_{oo'} > 0$ when wages rise more in $o'$.

The average mobility gain for workers initially in occupation $o$ is:
\begin{align*}
    \text{Mobility Gain}_o &= \sum_{o': d\ln w_{o'} > d\ln w_o} \mu_{oo'} \cdot \text{EV}_{oo'} \\
    &= \sum_{o': d\ln w_{o'} > d\ln w_o} [-\Theta_{oo'}(d\ln w_{o'} - d\ln w_o)] \cdot (d\ln w_{o'} - d\ln w_o) \\
    &= -\sum_{o'} \Theta_{oo'}(d\ln w_{o'} - d\ln w_o)^2 \cdot \mathbf{1}_{d\ln w_{o'} > d\ln w_o}
\end{align*}

This establishes equation \eqref{eq:mobility_gains}.

\subsection{Proof of Proposition \ref{prop:spectral}} \label{app_ss:spectral_proof}

We derive the spectral decomposition of wage incidence using the eigendecomposition $\Theta = U\Lambda V$, where $V = U^{-1}$.

\textbf{Step 1: Decompose the technological shock.}
Since the eigenvectors $\{\boldsymbol{u}_n\}$ form a basis for $\mathbb{R}^O$, we can write:
\begin{equation*}
    \frac{d\ln\boldsymbol{\alpha}}{\sigma} = \sum_{n=1}^O b_n \boldsymbol{u}_n
\end{equation*}
where the coefficients are $\boldsymbol{b}= \left(U^{\prime} U\right)^{-1} U^{\prime} \frac{d\ln\boldsymbol{\alpha}}{\sigma} $.

\textbf{Step 2: Apply the eigendecomposition to the pass-through matrix.}
The pass-through matrix can be written as:
\begin{align*}
    \Delta &= \left(\mathbf{I} + \frac{\Theta}{\sigma}\right)^{-1} = \left(\mathbf{I} + \frac{U\Lambda V}{\sigma}\right)^{-1} \\
    &= U\left(\mathbf{I} + \frac{\Lambda}{\sigma}\right)^{-1}V = \sum_{n=1}^O \frac{1}{1 + \lambda_n/\sigma} \boldsymbol{u}_n\boldsymbol{v}_n'
\end{align*}

\textbf{Step 3: Compute the wage response.}
Substituting into the wage incidence equation:
\begin{align*}
    d\ln\boldsymbol{w} &= \frac{d\ln y}{\sigma}\mathbf{1} - \Delta \cdot \frac{d\ln\boldsymbol{\alpha}}{\sigma} \\
    &= \frac{d\ln y}{\sigma}\mathbf{1} - \sum_{n=1}^O \frac{1}{1 + \lambda_n/\sigma} \boldsymbol{u}_n\boldsymbol{v}_n' \cdot \left(\sum_{m=1}^O b_m \boldsymbol{u}_m\right) \\
    &= \frac{d\ln y}{\sigma}\mathbf{1} - \sum_{n=1}^O \frac{b_n}{1 + \lambda_n/\sigma} \boldsymbol{u}_n
\end{align*}

The last equality uses the orthogonality property $\boldsymbol{v}_n' \cdot \boldsymbol{u}_m = \delta_{nm}$.

This completes the proof, showing that each eigenshock $\boldsymbol{u}_n$ passes through to wages with a dampening factor $(1 + \lambda_n/\sigma)^{-1}$.
\subsection{Proof of Eigenvalue Properties} \label{app_ss:eigenvalue_proof}

This section proves the eigenvalue properties of the labor supply elasticity matrix $\Theta$ stated in Lemma \ref{lemma:eigenvalues}.

\begin{proof}
We establish each property in turn.

\textbf{Part 1: Existence of zero eigenvalue with uniform eigenvector.}

From Section \ref{app_ss:zero_row_sum}, we know that $\sum_{o'} \Theta_{oo'} = 0$ for all $o$. This implies:
\begin{equation*}
    \Theta \cdot \mathbf{1} = \mathbf{0}
\end{equation*}
where $\mathbf{1} = [1, 1, \ldots, 1]'$. Therefore, $\lambda = 0$ is an eigenvalue with right eigenvector $\boldsymbol{u}_1 = \mathbf{1}/\sqrt{O}$ (normalized).

\textbf{Part 2: Non-negativity of all eigenvalues.}

The matrix $\Theta$ has the structure:
\begin{equation*}
    \Theta_{oo'} = \begin{cases}
        \theta\left[\frac{x_oF_{oo}}{F_o} + 1 - \pi_o\right] > 0 & \text{if } o = o' \\
        \theta\left[\frac{x_{o'}F_{oo'}}{F_o} - \pi_{o'}\right] < 0 & \text{if } o \neq o'
    \end{cases}
\end{equation*}

The sign-switching property of $F$ ensures $F_{oo'} \leq 0$ for $o \neq o'$, making $\Theta$ a matrix with positive diagonal and negative off-diagonal elements. Additionally, the diagonal dominance condition holds:
\begin{equation*}
    \Theta_{oo} = -\sum_{o' \neq o} \Theta_{oo'} > 0
\end{equation*}

By the Gershgorin circle theorem, all eigenvalues lie in the union of discs:
\begin{equation*}
    \lambda \in \bigcup_{o} \left\{z \in \mathbb{C}: |z - \Theta_{oo}| \leq \sum_{o' \neq o} |\Theta_{oo'}|\right\}
\end{equation*}

Since $\Theta_{oo} = \sum_{o' \neq o} |\Theta_{oo'}|$ (from the zero row sum), each disc is centered at a positive point with radius equal to the center. Therefore, all discs lie in the right half-plane: $\text{Re}(\lambda) \geq 0$.

For a real matrix with real eigenvalues (which $\Theta$ has due to its economic interpretation), this implies $\lambda \geq 0$.

\textbf{Part 3: Uniqueness of zero eigenvalue.}

Suppose $\lambda = 0$ has geometric multiplicity greater than 1. Then there exists a non-uniform vector $\boldsymbol{x} \neq c\mathbf{1}$ such that $\Theta\boldsymbol{x} = \mathbf{0}$.

Without loss of generality, normalize $\boldsymbol{x}$ so that $\max_o x_o = 1$ and $\min_o x_o < 1$. Let $o^* = \arg\max_o x_o$. Then:
\begin{equation*}
    0 = (\Theta\boldsymbol{x})_{o^*} = \Theta_{o^*o^*} + \sum_{o' \neq o^*} \Theta_{o^*o'} x_{o'}
\end{equation*}

Since $x_{o'} < x_{o^*} = 1$ for at least one $o'$ and $\Theta_{o^*o'} < 0$ for all $o' \neq o^*$:
\begin{equation*}
    \sum_{o' \neq o^*} \Theta_{o^*o'} x_{o'} > \sum_{o' \neq o^*} \Theta_{o^*o'} = -\Theta_{o^*o^*}
\end{equation*}

This gives $(\Theta\boldsymbol{x})_{o^*} > 0$, contradicting $\Theta\boldsymbol{x} = \mathbf{0}$. Therefore, the zero eigenvalue has geometric (and algebraic) multiplicity 1.

\end{proof}

\subsection{Additional Derivations for Dynamic Model} \label{c:appendix:add_math_dynamic}

\subsubsection{Derivation of System in Changes (Section \ref{a:appendix:sys_changes_dynamic})}

This section provides detailed derivations for the dynamic system expressed in growth rates.

\paragraph{Production equilibrium in changes}
Log-differentiating the wage equation yields:
\begin{equation*}
\sigma \ln \dot{w}_{o,t+1} + \ln \dot{L}_{o,t+1} = \ln \dot{Y}_{t+1} + \ln \dot{\alpha}_{o,t+1}
\end{equation*}

This implies wages adjust according to:
\begin{equation*}
\dot{\boldsymbol{w}}_{t+1} = \dot{w}(\dot{\boldsymbol{L}}_{t+1}, \dot{\boldsymbol{\Psi}}_{t+1})
\end{equation*}
where $\dot{\boldsymbol{\Psi}}_{t+1}$ represents changes in fundamentals.

\paragraph{Evolution of adjusted mobility}
Starting from the definition:
\begin{equation*}
\tilde{\mu}_{oo',t} = \frac{A_{o',t}Z_{oo',t}^{\theta/\kappa}}{F(A_{1,t}Z_{o1,t}^{\theta/\kappa}, \ldots, A_{O,t}Z_{oO,t}^{\theta/\kappa})}
\end{equation*}

Taking the ratio across time:
\begin{align*}
\frac{\tilde{\mu}_{oo',t}}{\tilde{\mu}_{oo',t-1}} &= \frac{A_{o',t}Z_{oo',t}^{\theta/\kappa}/A_{o',t-1}Z_{oo',t-1}^{\theta/\kappa}}{F(A_{1,t}Z_{o1,t}^{\theta/\kappa}, \ldots)/F(A_{1,t-1}Z_{o1,t-1}^{\theta/\kappa}, \ldots)}
\end{align*}

Using $Z_{oo',t} = \exp(\beta V_{o',t+1} + \ln w_{o',t} - \tau_{oo'})$ and $\dot{u}_{o,t} = \exp(V_{o,t} - V_{o,t-1})$:
\begin{align*}
\frac{\tilde{\mu}_{oo',t}}{\tilde{\mu}_{oo',t-1}} &= \frac{\dot{A}_{o',t} \dot{u}_{o',t+1}^{\beta\theta/\kappa} \dot{w}_{o',t}^{\theta/\kappa}}{F(\{A_{o'',t}Z_{oo'',t}^{\theta/\kappa}/F_{t-1}\}_{o''=1}^O)}
\end{align*}

where $F_{t-1} = F(A_{1,t-1}Z_{o1,t-1}^{\theta/\kappa}, \ldots, A_{O,t-1}Z_{oO,t-1}^{\theta/\kappa})$.

Using the homogeneity of $F$ and noting that $A_{o'',t}Z_{oo'',t}^{\theta/\kappa}/F_{t-1} = \tilde{\mu}_{oo'',t-1} \cdot \dot{A}_{o'',t} \dot{u}_{o'',t+1}^{\beta\theta/\kappa} \dot{w}_{o'',t}^{\theta/\kappa}$:
\begin{equation*}
\boxed{\frac{\tilde{\mu}_{oo',t}}{\tilde{\mu}_{oo',t-1}} = \frac{\dot{A}_{o',t} \dot{u}_{o',t+1}^{\beta\theta/\kappa} \dot{w}_{o',t}^{\theta/\kappa}}{F(\{\tilde{\mu}_{oo'',t-1} \dot{A}_{o'',t} \dot{u}_{o'',t+1}^{\beta\theta/\kappa} \dot{w}_{o'',t}^{\theta/\kappa}\}_{o''=1}^O)}}
\end{equation*}

\paragraph{Evolution of expected utility}
From the value function:
\begin{equation*}
V_{o,t} = \frac{\kappa}{\theta}\ln F(A_{1,t}Z_{o1,t}^{\theta/\kappa}, \ldots, A_{O,t}Z_{oO,t}^{\theta/\kappa}) + \bar{\gamma}\frac{\kappa}{\theta}
\end{equation*}

The change in value is:
\begin{align*}
V_{o,t+1} - V_{o,t} &= \frac{\kappa}{\theta}\ln\frac{F(A_{1,t+1}Z_{o1,t+1}^{\theta/\kappa}, \ldots)}{F(A_{1,t}Z_{o1,t}^{\theta/\kappa}, \ldots)}
\end{align*}

Using homogeneity to factor out $F(A_{1,t}Z_{o1,t}^{\theta/\kappa}, \ldots)$:
\begin{align*}
V_{o,t+1} - V_{o,t} &= \frac{\kappa}{\theta}\ln F\left(\left\{\frac{A_{o'',t+1}Z_{oo'',t+1}^{\theta/\kappa}}{F(A_{1,t}Z_{o1,t}^{\theta/\kappa}, \ldots)}\right\}_{o''=1}^O\right)
\end{align*}

Substituting $A_{o'',t+1}Z_{oo'',t+1}^{\theta/\kappa} = A_{o'',t}Z_{oo'',t}^{\theta/\kappa} \cdot \dot{A}_{o'',t+1} \dot{Z}_{oo'',t+1}^{\theta/\kappa}$ and recognizing that $A_{o'',t}Z_{oo'',t}^{\theta/\kappa}/F_t = \tilde{\mu}_{oo'',t}$:
\begin{equation*}
\boxed{\dot{u}_{o,t+1} = F(\{\tilde{\mu}_{oo'',t} \dot{A}_{o'',t+1} \dot{u}_{o'',t+2}^{\beta\theta/\kappa} \dot{w}_{o'',t+1}^{\theta/\kappa}\}_{o''=1}^O)^{\kappa/\theta}}
\end{equation*}

\paragraph{Labor market clearing}
Employment evolves through transitions:
\begin{equation*}
L_{o,t} = \sum_{o'} \mu_{o'o,t} L_{o',t-1}
\end{equation*}
where observed transitions relate to adjusted rates via:
\begin{equation*}
\mu_{oo',t} = \tilde{\mu}_{oo',t} F_{o'}(\tilde{\mu}_{o1,t}, \ldots, \tilde{\mu}_{oO,t})
\end{equation*}

These equations form a complete system characterizing the dynamic equilibrium in growth rates, preserving the DIDES structure through the correlation function $F$.
\subsubsection{Derivation of Dynamic Hat Algebra (Section \ref{a:appendix:dynamic_hat_algebra})}

This section derives the counterfactual evolution equations for the dynamic model.

\paragraph{Counterfactual wage determination.}
From the production equilibrium, counterfactual wages relate to baseline wages through:
\begin{equation*}
\hat{w}_{o,t+1} = \frac{\dot{w}_{o,t+1}'}{\dot{w}_{o,t+1}} = \left(\frac{\hat{Y}_{t+1} \hat{\alpha}_{o,t+1}}{\hat{L}_{o,t+1}}\right)^{\frac{1}{\sigma}}
\end{equation*}
where hats denote ratios of counterfactual to baseline growth rates.

\paragraph{Evolution of counterfactual transition probabilities.}
Starting from the ratio of counterfactual to baseline growth rates:
\begin{align*}
\frac{\tilde{\mu}_{oo',t}'}{\tilde{\mu}_{oo',t-1}'} &= \frac{\dot{\tilde{\mu}}_{oo',t} \hat{A}_{o',t} \hat{u}_{o',t+1}^{\beta\theta/\kappa} \hat{w}_{o',t}^{\theta/\kappa}}{\text{[Denominator]}}
\end{align*}

The denominator requires careful manipulation. Using the ratio of counterfactual to baseline correlation functions:
\begin{align*}
\text{[Denominator]} &= \frac{F(\{\tilde{\mu}_{oo'',t-1}' \dot{A}_{o'',t}' \dot{u}_{o'',t+1}'^{\beta\theta/\kappa} \dot{w}_{o'',t}'^{\theta/\kappa}\}_{o''})}
{F(\{\tilde{\mu}_{oo'',t-1} \dot{A}_{o'',t} \dot{u}_{o'',t+1}^{\beta\theta/\kappa} \dot{w}_{o'',t}^{\theta/\kappa}\}_{o''})}
\end{align*}

Recognizing that counterfactual growth rates equal baseline growth times hat values, and using homogeneity of $F$:
\begin{align*}
&= F\left(\left\{\frac{\tilde{\mu}_{oo'',t-1}'}{\tilde{\mu}_{oo'',t-1}} \cdot \tilde{\mu}_{oo'',t} \cdot \hat{A}_{o'',t} \hat{u}_{o'',t+1}^{\beta\theta/\kappa} \hat{w}_{o'',t}^{\theta/\kappa}\right\}_{o''}\right)
\end{align*}

Therefore, the recursive formula is:
\begin{equation*}
\boxed{\tilde{\mu}_{oo',t}' = \frac{\tilde{\mu}_{oo',t-1}' \dot{\tilde{\mu}}_{oo',t} \hat{A}_{o',t} \hat{u}_{o',t+1}^{\beta\theta/\kappa} \hat{w}_{o',t}^{\theta/\kappa}}{F(\{\tilde{\mu}_{oo'',t-1}' \dot{\tilde{\mu}}_{oo'',t} \hat{A}_{o'',t} \hat{u}_{o'',t+1}^{\beta\theta/\kappa} \hat{w}_{o'',t}^{\theta/\kappa}\}_{o''=1}^O)}}
\end{equation*}

\paragraph{Evolution of counterfactual expected utility.}
Following similar steps for the utility growth rates:
\begin{align*}
\hat{u}_{o,t+1} &= \frac{\dot{u}_{o,t+1}'}{\dot{u}_{o,t+1}} \\
&= \frac{F(\{\tilde{\mu}_{oo'',t}' \dot{A}_{o'',t+1}' \dot{u}_{o'',t+2}'^{\beta\theta/\kappa} \dot{w}_{o'',t+1}'^{\theta/\kappa}\}_{o''})^{\kappa/\theta}}{F(\{\tilde{\mu}_{oo'',t} \dot{A}_{o'',t+1} \dot{u}_{o'',t+2}^{\beta\theta/\kappa} \dot{w}_{o'',t+1}^{\theta/\kappa}\}_{o''})^{\kappa/\theta}}
\end{align*}

Using the same homogeneity argument:
\begin{equation*}
\boxed{\hat{u}_{o,t+1} = F(\{\tilde{\mu}_{oo'',t}' \dot{\tilde{\mu}}_{oo'',t+1} \hat{A}_{o'',t+1} \hat{u}_{o'',t+2}^{\beta\theta/\kappa} \hat{w}_{o'',t+1}^{\theta/\kappa}\}_{o''=1}^O)^{\kappa/\theta}}
\end{equation*}

\paragraph{Observed transitions and employment evolution.}
The observed counterfactual transitions incorporate correlation effects:
\begin{equation*}
\mu_{oo',t}' = \tilde{\mu}_{oo',t}' F_{o'}(\tilde{\mu}_{o1,t}', \ldots, \tilde{\mu}_{oO,t}')
\end{equation*}

Employment evolves through the transition matrix plus exogenous flows:
\begin{equation*}
L_{o,t}' = \sum_{o'} \mu_{o'o,t}' L_{o',t-1}' + \Delta L_{o,t}
\end{equation*}

These equations provide a complete characterization of counterfactual dynamics, preserving the DIDES structure through the correlation function $F$ while enabling analysis of alternative technological scenarios.
\subsubsection{Initial Dynamics with Unexpected Shocks}

For unexpected shocks at $t=1$, the economy begins at baseline equilibrium with $\hat{u}_{o,0} = 1$, $\mu_{oo',0}' = \mu_{oo',0}$, and $L_{o,0}' = L_{o,0}$.

\paragraph{Deriving the initial utility adjustment.}
The baseline expected utility at $t=0$ is:
\begin{equation*}
u_{o,0} = F(\{A_{o',0}Z_{oo',0}^{\theta/\kappa}\}_{o'=1}^O)^{\kappa/\theta}
\end{equation*}

Since initial conditions are identical ($A_{o',0}' = A_{o',0}$, $w_{o',0}' = w_{o',0}$), we can rewrite using counterfactual notation:
\begin{equation*}
u_{o,0} = F\left(\left\{\frac{A_{o',0}}{A_{o',0}'} \cdot \frac{Z_{oo',0}}{Z_{oo',0}'} \cdot A_{o',0}'Z_{oo',0}'^{\theta/\kappa}\right\}_{o'=1}^O\right)^{\kappa/\theta}
\end{equation*}

Since the ratios equal unity at $t=0$:
\begin{equation*}
u_{o,0} = F(\{A_{o',0}'Z_{oo',0}'^{\theta/\kappa}\}_{o'=1}^O)^{\kappa/\theta}
\end{equation*}

After the shock, counterfactual utility at $t=1$ is:
\begin{equation*}
u_{o,1}' = F(\{A_{o',1}'Z_{oo',1}'^{\theta/\kappa}\}_{o'=1}^O)^{\kappa/\theta}
\end{equation*}

Taking the ratio and using homogeneity of $F$:
\begin{align*}
\frac{u_{o,1}'}{u_{o,0}} &= \frac{F(\{A_{o',1}'Z_{oo',1}'^{\theta/\kappa}\}_{o'})^{\kappa/\theta}}{F(\{A_{o',0}'Z_{oo',0}'^{\theta/\kappa}\}_{o'})^{\kappa/\theta}} \\
&= F\left(\left\{\frac{A_{o',1}'Z_{oo',1}'^{\theta/\kappa}}{F(\{A_{o'',0}'Z_{oo'',0}'^{\theta/\kappa}\}_{o''})} \right\}_{o'=1}^O\right)^{\kappa/\theta}
\end{align*}

\paragraph{Connecting to baseline transition probabilities.}
Note that at $t=0$:
\begin{align*}
\tilde{\mu}_{oo',0} &= \frac{A_{o',0}Z_{oo',0}^{\theta/\kappa}}{F(\{A_{o'',0}Z_{oo'',0}^{\theta/\kappa}\}_{o''})} \\
&= \frac{\frac{A_{o',0}Z_{oo',0}^{\theta/\kappa}}{A_{o',0}'Z_{oo',0}'^{\theta/\kappa}} \cdot A_{o',0}'Z_{oo',0}'^{\theta/\kappa}}{F\left(\left\{\frac{A_{o'',0}Z_{oo'',0}^{\theta/\kappa}}{A_{o'',0}'Z_{oo'',0}'^{\theta/\kappa}} \cdot A_{o'',0}'Z_{oo'',0}'^{\theta/\kappa}\right\}_{o''}\right)}
\end{align*}

Since initial conditions are identical, the ratios equal unity, yielding:
\begin{equation*}
\tilde{\mu}_{oo',0} = \frac{A_{o',0}'Z_{oo',0}'^{\theta/\kappa}}{F(\{A_{o'',0}'Z_{oo'',0}'^{\theta/\kappa}\}_{o''})}
\end{equation*}

Combining this with the utility ratio:
\begin{align*}
\frac{u_{o,1}'}{u_{o,0}} &= F\left(\left\{\tilde{\mu}_{oo',0} \cdot \frac{A_{o',1}'Z_{oo',1}'^{\theta/\kappa}}{A_{o',0}'Z_{oo',0}'^{\theta/\kappa}}\right\}_{o'=1}^O\right)^{\kappa/\theta} \\
&= F\left(\left\{\tilde{\mu}_{oo',0} \cdot \dot{A}_{o',1}' \dot{Z}_{oo',1}'^{\theta/\kappa}\right\}_{o'=1}^O\right)^{\kappa/\theta}
\end{align*}

Since $Z_{oo',t} = \exp(\beta V_{o',t+1} + \ln w_{o',t} - \tau_{oo'})$:
\begin{equation*}
\dot{u}_{o,1}' = F(\{\tilde{\mu}_{oo',0} \dot{A}_{o',1}' \dot{w}_{o',1}'^{\theta/\kappa} \dot{u}_{o',2}'^{\beta\theta/\kappa}\}_{o'=1}^O)^{\kappa/\theta}
\end{equation*}

\paragraph{Computing the initial hat values.}
The baseline utility growth follows:
\begin{equation*}
\dot{u}_{o,1} = F(\{\tilde{\mu}_{oo',0} \dot{A}_{o',1} \dot{w}_{o',1}^{\theta/\kappa} \dot{u}_{o',2}^{\beta\theta/\kappa}\}_{o'=1}^O)^{\kappa/\theta}
\end{equation*}

The hat value is:
\begin{align*}
\hat{u}_{o,1} &= \frac{\dot{u}_{o,1}'}{\dot{u}_{o,1}} \\
&= F\left(\left\{\frac{\tilde{\mu}_{oo',0} \dot{A}_{o',1}' \dot{w}_{o',1}'^{\theta/\kappa} \dot{u}_{o',2}'^{\beta\theta/\kappa}}{F(\{\tilde{\mu}_{oo'',0} \dot{A}_{o'',1} \dot{w}_{o'',1}^{\theta/\kappa} \dot{u}_{o'',2}^{\beta\theta/\kappa}\}_{o''})}\right\}_{o'=1}^O\right)^{\kappa/\theta}
\end{align*}

Recognizing that $\dot{A}_{o',1}'/\dot{A}_{o',1} = \hat{A}_{o',1}$ and similarly for other variables:
\begin{align*}
\hat{u}_{o,1} &= F\left(\left\{\tilde{\mu}_{oo',1} \cdot \hat{A}_{o',1} \hat{w}_{o',1}^{\theta/\kappa} \hat{u}_{o',2}^{\beta\theta/\kappa} \cdot \frac{u_{o',1}'/u_{o',1}}{u_{o',1}'/u_{o',1}}\right\}_{o'=1}^O\right)^{\kappa/\theta}
\end{align*}

Note that $u_{o',1}'/u_{o',1} = (u_{o',1}'/u_{o',0}) \cdot (u_{o',0}/u_{o',1}) = \dot{u}_{o',1}'/\dot{u}_{o',1} = \hat{u}_{o',1}$.

Defining $\vartheta_{oo',1} = \tilde{\mu}_{oo',1}\hat{u}_{o',1}^{\beta\theta/\kappa}$:
\begin{equation*}
\boxed{\hat{u}_{o,1} = F(\{\vartheta_{oo',1} \hat{A}_{o',1} \hat{w}_{o',1}^{\theta/\kappa} \hat{u}_{o',2}^{\beta\theta/\kappa}\}_{o'=1}^O)^{\kappa/\theta}}
\end{equation*}

\paragraph{Initial transition probabilities.}
Following parallel derivation for transition probabilities, starting from:
\begin{equation*}
\frac{\tilde{\mu}_{oo',1}'}{\tilde{\mu}_{oo',1}} = \frac{A_{o',1}'Z_{oo',1}'^{\theta/\kappa}/A_{o',1}Z_{oo',1}^{\theta/\kappa}}{F(\{A_{o'',1}'Z_{oo'',1}'^{\theta/\kappa}\}_{o''})/F(\{A_{o'',1}Z_{oo'',1}^{\theta/\kappa}\}_{o''})}
\end{equation*}

After similar manipulations:
\begin{equation*}
\boxed{\tilde{\mu}_{oo',1}' = \frac{\vartheta_{oo',1} \hat{A}_{o',1} \hat{w}_{o',1}^{\theta/\kappa} \hat{u}_{o',2}^{\beta\theta/\kappa}}{F(\{\vartheta_{oo'',1} \hat{A}_{o'',1} \hat{w}_{o'',1}^{\theta/\kappa} \hat{u}_{o'',2}^{\beta\theta/\kappa}\}_{o''=1}^O)}}
\end{equation*}

The adjustment factor $\vartheta_{oo',1}$ captures the combined effect of the initial shock and forward-looking expectations, encoding how unexpected changes propagate through the DIDES structure.


\subsection{Derivation of Effective Labor Supply Elasticities} \label{app:eff_elasticity_derivation}

This appendix derives the effective labor supply elasticities in equation \eqref{app:eq:eff_elasticity} when a fraction $\delta \in [0,1]$ of workers contribute their idiosyncratic productivity directly to production, as discussed in Appendix \ref{app:efficiency}.

\subsubsection{Workers with Productivity in Production}

For workers whose productivity enters production, occupation $o$ provides wage $w_o\epsilon_o(i)$ per unit of labor. Given the multivariate Fréchet distribution from equation \eqref{eq:prod_distr}:
\begin{equation*}
    \Pr[\epsilon_1(i) \leq \epsilon_1, \ldots, \epsilon_O(i) \leq \epsilon_O] = \exp[-F(A_1\epsilon_1^{-\theta}, \ldots, A_O\epsilon_O^{-\theta})]
\end{equation*}

Workers choose occupation $o^*(i) = \arg\max_o \{w_o\epsilon_o(i)\}$, yielding employment share:
\begin{equation*}
    \pi_o = \frac{A_ow_o^{\theta}F_o(A_1w_1^{\theta}, \ldots, A_Ow_O^{\theta})}{F(A_1w_1^{\theta}, \ldots, A_Ow_O^{\theta})}
\end{equation*}

The key distinction arises in aggregating efficiency units. The conditional expectation of productivity for workers choosing occupation $o$ is:
\begin{equation*}
    \mathbb{E}[\epsilon_o(i)|o^*(i)=o] = \Gamma\left(1-\frac{1}{\theta}\right)F(A_1w_1^{\theta}, \ldots, A_Ow_O^{\theta})^{1/\theta}/w_o
\end{equation*}

Therefore, efficiency units supplied to occupation $o$ by productivity workers are:
\begin{equation*}
    \ell_o^{\text{prod}} = \delta \cdot \pi_o \cdot \mathbb{E}[\epsilon_o(i)|o^*(i)=o] = \delta \cdot \Gamma\left(1-\frac{1}{\theta}\right)w_o^{\theta-1}F_o(x_1, \ldots, x_O) \cdot F(x_1, \ldots, x_O)^{(1-\theta)/\theta}
\end{equation*}
where $x_o = A_ow_o^{\theta}$.

To derive elasticities, we differentiate with respect to wages. Define the wage index $W = F(x_1, \ldots, x_O)^{1/\theta}$. The elasticity with respect to relative wages is:
\begin{align*}
    \frac{\partial\ln\ell_o^{\text{prod}}}{\partial\ln(w_{o'}/W)} &= \frac{\partial}{\partial\ln(w_{o'}/W)}\ln\left[\left(\frac{w_o}{W}\right)^{\theta-1}F_o\left(\frac{x_1}{W^{\theta}}, \ldots, \frac{x_O}{W^{\theta}}\right)\right]\\
    &= \begin{cases}
        (\theta-1) + \theta\frac{x_oF_{oo}}{F_o} & \text{if } o' = o \\
        \theta\frac{x_{o'}F_{oo'}}{F_o} & \text{if } o' \neq o
    \end{cases}
\end{align*}

Since $\frac{\partial\ln W}{\partial\ln w_{o'}} = \pi_{o'}$, the absolute wage elasticity becomes:
\begin{equation}\label{app:eq:theta_prod}
    \Theta_{oo'}^{\text{prod}} = \frac{\partial\ln\ell_o^{\text{prod}}}{\partial\ln w_{o'}} = \begin{cases}
        \theta\frac{x_oF_{oo}}{F_o} + (\theta-1)(1-\pi_o) & \text{if } o' = o \\
        \theta\frac{x_{o'}F_{oo'}}{F_o} - (\theta-1)\pi_{o'} & \text{if } o' \neq o
    \end{cases}
\end{equation}

\subsubsection{Workers with Productivity in Preferences}

For workers whose productivity affects effort costs, efficiency units supplied are $(1-\delta) \cdot \pi_o$. The derivation follows Appendix \ref{app_ss:elasticity_derivation}, yielding:
\begin{equation}\label{app:eq:theta_pref}
    \Theta_{oo'}^{\text{pref}} = \begin{cases}
        \theta\left[\frac{x_oF_{oo}}{F_o} + 1 - \pi_o\right] & \text{if } o' = o \\
        \theta\left[\frac{x_{o'}F_{oo'}}{F_o} - \pi_{o'}\right] & \text{if } o' \neq o
    \end{cases}
\end{equation}

\subsubsection{Weighted Effective Elasticity}

Total efficiency units in occupation $o$ are:
\begin{equation*}
    \ell_o^{\text{eff}} = \delta \cdot \pi_o \cdot \mathbb{E}[\epsilon_o(i)|o^*(i)=o] + (1-\delta) \cdot \pi_o
\end{equation*}

Define $s_o^{\delta}$ as the share of efficiency units from productivity workers:
\begin{equation*}
    s_o^{\delta} = \frac{\delta \cdot \mathbb{E}[\epsilon_o(i)|o^*(i)=o]}{\delta \cdot \mathbb{E}[\epsilon_o(i)|o^*(i)=o] + (1-\delta)}
\end{equation*}

The effective elasticity is the weighted average by efficiency unit shares:
\begin{align}
    \Theta_{oo'}^{\text{eff}} &= s_o^{\delta} \cdot \Theta_{oo'}^{\text{prod}} + (1-s_o^{\delta}) \cdot \Theta_{oo'}^{\text{pref}} \label{app:eq:theta_weighted}
\end{align}

Substituting the individual elasticities:
\begin{align}
    \Theta_{oo'}^{\text{eff}} &= \begin{cases}
        s_o^{\delta}\left[\theta\frac{x_oF_{oo}}{F_o} + (\theta-1)(1-\pi_o)\right] + (1-s_o^{\delta})\left[\theta\frac{x_oF_{oo}}{F_o} + \theta(1-\pi_o)\right] & \text{if } o' = o \nonumber\\
        s_{o'}^{\delta}\left[\theta\frac{x_{o'}F_{oo'}}{F_o} - (\theta-1)\pi_{o'}\right] + (1-s_{o}^{\delta})\left[\theta\frac{x_{o'}F_{oo'}}{F_o} - \theta\pi_{o'}\right] & \text{if } o' \neq o \nonumber
    \end{cases}
\end{align}

Simplifying:
\begin{align}
    \Theta_{oo'}^{\text{eff}} &= \begin{cases}
        \theta\frac{x_oF_{oo}}{F_o} + [s_o^{\delta}(\theta-1) + (1-s_o^{\delta})\theta](1-\pi_o) & \text{if } o' = o \\
        \theta\frac{x_{o'}F_{oo'}}{F_o} - [s_{o'}^{\delta}(\theta-1) + (1-s_{o}^{\delta})\theta]\pi_{o'} & \text{if } o' \neq o
    \end{cases} \nonumber\\
    &= \begin{cases}
        \theta\frac{x_oF_{oo}}{F_o} + (\theta-s_o^{\delta})(1-\pi_o) & \text{if } o' = o \\
        \theta\frac{x_{o'}F_{oo'}}{F_o} - (\theta-s_{o}^{\delta})\pi_{o'} & \text{if } o' \neq o
    \end{cases} \label{app:eq:theta_eff_final}
\end{align}

This establishes the result in equation \eqref{app:eq:eff_elasticity}, demonstrating that the correlation term $\theta x_{o'}F_{oo'}/F_o$ remains unaffected by worker composition while the baseline substitution term decreases from $\theta$ to $(\theta-s_o^{\delta})$ as the share of efficiency units from productivity workers increases.

\end{document}